\documentclass[reqno,fleqn]{amsart}
\usepackage[T1]{fontenc}
\usepackage[latin9]{inputenc}
\usepackage{array}
\usepackage{enumitem}
\usepackage{amstext}
\usepackage{amsthm}
\usepackage{amssymb}
\usepackage{marvosym}
\PassOptionsToPackage{normalem}{ulem}
\usepackage{ulem}

\makeatletter

\providecommand{\tabularnewline}{\\}

\numberwithin{equation}{section}
\numberwithin{figure}{section}
\theoremstyle{plain}
\newtheorem{thm}{\protect\theoremname}[section]
  \theoremstyle{plain}
  \newtheorem{cor}[thm]{\protect\corollaryname}
  \theoremstyle{plain}
  \newtheorem{prop}[thm]{\protect\propositionname}
  \theoremstyle{plain}
  \newtheorem{fact}[thm]{\protect\factname}
  \theoremstyle{plain}
  \newtheorem{lem}[thm]{\protect\lemmaname}
  \theoremstyle{remark}
  \newtheorem{claim}[thm]{\protect\claimname}
  \theoremstyle{plain}
  \newtheorem*{thm*}{\protect\theoremname}
  \theoremstyle{remark}
  \newtheorem{rem}[thm]{\protect\remarkname}

\usepackage{amsmath}
\usepackage{sherstov}
\usepackage{small-caption}
\usepackage{xcolor}
\usepackage{setspace}
\usepackage{centernot}
\usepackage{bbm}

\newcommand{\moo}{\{-1,+1\}}
\newcommand{\classAC}{\operatorname{AC}}
\newcommand{\orth}{\operatorname{orth}}

\usepackage{hyperref}
\hypersetup{
pdfauthor = {A. A. Sherstov},
plainpages = false,
bookmarks = true,
bookmarksopen = true,
pdftex,
colorlinks=false,
citecolor = teal,
linkcolor = teal,
urlcolor  = teal,
}

\newtheoremstyle{myplain}      {10pt}{10pt}{\itshape}{}{\scshape}{.}{.5em}{}
\newtheoremstyle{mydefinition} {10pt}{10pt}{}{}{\scshape}{.}{.5em}{}
\newtheoremstyle{myremark} {10pt}{10pt}{}{}{\itshape}{.}{.5em}{}

\@namedef{thm}{\@thm{\let \thm@swap \@gobble \th@myplain }{thm}{Theorem}}
\@namedef{lem}{\@thm{\let \thm@swap \@gobble \th@myplain }{thm}{Lemma}}
\@namedef{cor}{\@thm{\let \thm@swap \@gobble \th@myplain }{thm}{Corollary}}
\@namedef{prop}{\@thm{\let \thm@swap \@gobble \th@myplain }{thm}{Proposition}}
\@namedef{claim}{\@thm{\let \thm@swap \@gobble \th@myplain }{thm}{Claim}}
\@namedef{fact}{\@thm{\let \thm@swap \@gobble \th@myplain }{thm}{Fact}}
\@namedef{defn}{\@thm{\let \thm@swap \@gobble \th@mydefinition }{thm}{Definition}}
\@namedef{rem}{\@thm{\let \thm@swap \@gobble \th@mydefinition }{thm}{Remark}}
\@namedef{example}{\@thm{\let \thm@swap \@gobble \th@mydefinition }{thm}{Example}}

\@namedef{thm*}{\@thm {\th@myplain }{}{Theorem}}
\@namedef{lem*}{\@thm {\th@myplain }{}{Lemma}}
\@namedef{cor*}{\@thm {\th@myplain }{}{Corollary}}
\@namedef{prop*}{\@thm {\th@myplain }{}{Proposition}}
\@namedef{claim*}{\@thm {\th@myplain }{}{Claim}}
\@namedef{fact*}{\@thm {\th@myplain }{}{Fact}}
\@namedef{defn*}{\@thm {\th@mydefinition }{}{Definition}}
\@namedef{rem*}{\@thm {\th@mydefinition }{}{Remark}}
\@namedef{example*}{\@thm {\th@mydefinition }{}{Example}}

\renewcommand{\mathcal}[1]{\mathscr{#1}}

\makeatletter
\def\@seccntformat#1{%
  \protect\textup{%
    \protect\@secnumfont
    \expandafter\protect\csname format#1\endcsname 
    \csname the#1\endcsname
    \protect\@secnumpunct
  }%
}

\makeatletter
\newcommand \SparseDotfill {\leavevmode \leaders \hb@xt@ .7em{\hss .\hss }\hfill \kern \z@}
\makeatother	

\makeatletter
\def\@tocline#1#2#3#4#5#6#7{\relax
  \ifnum #1>\c@tocdepth 
  \else
    \par \addpenalty\@secpenalty\addvspace{\ifnum #1=1 2mm \else #2\fi}%
    \begingroup \hyphenpenalty\@M
    \@ifempty{#4}{%
      \@tempdima\csname r@tocindent\number#1\endcsname\relax
    }{%
      \@tempdima#4\relax
    }%
    \parindent\z@ \leftskip#3\relax \advance\leftskip\@tempdima\relax
    \rightskip\@pnumwidth plus4em \parfillskip-\@pnumwidth
          \ifnum #1=1 \bfseries #5\else #5\fi 
   \leavevmode\hskip-\@tempdima
      \ifcase #1
       \or\or \hskip 1em \or \hskip 2em \else \hskip 3em \fi%
#6     \nobreak\relax
{\ifnum #1=1\hfill \else \SparseDotfill\fi}
 \hbox to\@pnumwidth{\@tocpagenum{
    \ifnum #1=1 \bfseries \fi #7}}\par
    \nobreak
    \endgroup
  \fi}
\makeatother

\usepackage{enumitem}
\setenumerate{label=\rm (\roman*),labelsep=5mm}

\usepackage{graphicx}

\newcommand{\SURJ}{\operatorname{SURJ}}

\providecommand{\noopsort}[1]{}

\newcommand{\cube}{\operatorname{cube}}
\newcommand{\diam}{\operatorname{diam}}

\newcommand{\Smooth}{\mathfrak{S}}
\newcommand{\SmoothFunction}{\operatorname{Smooth}}
\newcommand{\Distribution}{\mathfrak{D}}

\DeclareMathOperator{\Sgn}{\widetilde{sgn}}

\newcommand{\VV}{\mathcal{V}}
\DeclareMathOperator{\srank}{\rk_\pm}
\DeclareMathOperator{\pospart}{pos}
\DeclareMathOperator{\negpart}{neg}

\DeclareMathOperator{\cone}{cone}

\newcommand{\WeaklyBounded}{\mathfrak{B}}
\newcommand{\StronglyBounded}{\mathfrak{B}^*}
\newcommand{\tallbar}{\scalebox{0.9}{|}}

\newcommand{\pomo}{\{-1,+1\}}

\newcommand{\pomon}{\pomo^n}

\DeclareMathOperator{\pp}{PP}
\DeclareMathOperator{\upp}{UPP}

\AtBeginDocument{
  
}

\makeatother

  \providecommand{\claimname}{Claim}
  \providecommand{\corollaryname}{Corollary}
  \providecommand{\factname}{Fact}
  \providecommand{\lemmaname}{Lemma}
  \providecommand{\propositionname}{Proposition}
  \providecommand{\remarkname}{Remark}
  \providecommand{\theoremname}{Theorem}
\providecommand{\theoremname}{Theorem}

\begin{document}

\title[The Threshold Degree and Sign-Rank of AC$^{0}$]{Near-Optimal Lower Bounds on the Threshold Degree and Sign-Rank of
AC$^{0}$}

\author{Alexander A. Sherstov and Pei Wu}

\thanks{$^{*}$ Computer Science Department, UCLA, Los Angeles, CA~90095.
Supported by NSF grant CCF-1814947 and an Alfred P. Sloan Foundation
Research Fellowship.\\
 {\large{}\Letter ~}\texttt{\{sherstov,pwu\}@cs.ucla.edu }}
\begin{abstract}
The \emph{threshold degree} of a Boolean function $f\colon\zoon\to\zoo$
is the minimum degree of a real polynomial $p$ that represents $f$
in sign: $\sign p(x)=(-1)^{f(x)}.$ A related notion is \emph{sign-rank},
defined for a Boolean matrix $F=[F_{ij}]$ as the minimum rank of
a real matrix $M$ with $\sign M_{ij}=(-1)^{F_{ij}}$. Determining
the maximum threshold degree and sign-rank achievable by constant-depth
circuits ($\classAC^{0}$) is a well-known and extensively studied
open problem, with complexity-theoretic and algorithmic applications.

We give an essentially optimal solution to this problem. For any $\epsilon>0,$
we construct an $\classAC^{0}$ circuit in $n$~variables that has
threshold degree $\Omega(n^{1-\epsilon})$ and sign-rank $\exp(\Omega(n^{1-\epsilon})),$
improving on the previous best lower bounds of $\Omega(\sqrt{n})$
and $\exp(\tilde{\Omega}(\sqrt{n}))$, respectively. Our results subsume
\emph{all} previous lower bounds on the threshold degree and sign-rank
of $\classAC^{0}$ circuits of any given depth, with a strict improvement
starting at depth~$4$. As a corollary, we also obtain near-optimal
bounds on the discrepancy, threshold weight, and threshold density
of $\classAC^{0}$, strictly subsuming previous work on these quantities.
Our work gives some of the strongest lower bounds to date on the communication
complexity of $\classAC^{0}$.
\end{abstract}

\maketitle
\belowdisplayskip=12pt plus 1pt minus 3pt 
\abovedisplayskip=12pt plus 1pt minus 3pt 
\thispagestyle{empty}

\newpage{}\thispagestyle{empty}
\hypersetup{linkcolor=black} 
~\vspace{-10mm}
\tableofcontents{}\newpage{}

\hypersetup{linkcolor=teal} 
\thispagestyle{empty}

\hyphenation{com-po-nent-wise}

\section{Introduction}

A real polynomial $p$ is said to \emph{sign-represent} the Boolean
function $f\colon\zoon\to\zoo$ if $\sign p(x)=(-1)^{f(x)}$ for every
input $x\in\zoon.$ The \emph{threshold degree} of $f$, denoted $\degthr(f)$,
is the minimum degree of a multivariate real polynomial that sign-represents
$f$. Equivalent terms in the literature include \emph{strong degree}~\cite{aspnes91voting},
\emph{voting polynomial degree}~\cite{krause94depth2mod}, \emph{PTF
degree}~\cite{OS-extremal-ptf}, and \emph{sign degree}~\cite{buhrman07pp-upp}.
Since any function $f\colon\zoon\to\zoo$ can be represented exactly
by a real polynomial of degree at most $n,$ the threshold degree
of $f$ is an integer between $0$ and $n.$ Viewed as a computational
model, sign-representation is remarkably powerful because it corresponds
to the strongest form of pointwise approximation. The formal study
of threshold degree began in 1969 with the pioneering work of Minsky
and Papert~\cite{minsky88perceptrons} on limitations of perceptrons.
The authors of~\cite{minsky88perceptrons} famously proved that the
parity function on $n$ variables has the maximum possible threshold
degree, $n$. They obtained lower bounds on the threshold degree of
several other functions, including DNF formulas and intersections
of halfspaces. Since then, sign-representing polynomials have found
applications far beyond artificial intelligence. In theoretical computer
science, applications of threshold degree include circuit lower bounds~\cite{krause94depth2mod,KP98threshold,sherstov07ac-majmaj,arkadev07multiparty,beame-huyn-ngoc09multiparty-focs},
size-depth trade-offs~\cite{paturi-saks94rational,siu-roy-kailath94rational},
communication complexity~\cite{sherstov07ac-majmaj,arkadev07multiparty,sherstov07quantum,RS07dc-dnf,beame-huyn-ngoc09multiparty-focs,sherstov12mdisj,sherstov13directional},
structural complexity~\cite{aspnes91voting,beigel91rational}, and
computational learning~\cite{KS01dnf,KOS:02,odonnell03degree,ACRSZ07nand,sherstov09hshs,sherstov09opthshs,bun-thaler13amplification,sherstov14sign-deg-ac0,thaler14omb}.

The notion of threshold degree has been especially influential in
the study of $\classAC^{0}$, the class of constant-depth polynomial-size
circuits with $\wedge,\vee,\neg$ gates of unbounded fan-in. The first
such result was obtained by Aspnes et al.~\cite{aspnes91voting},
who used sign-representing polynomials to give a beautiful new proof
of classic lower bounds for $\classAC^{0}$. In communication complexity,
the notion of threshold degree played a central role in the first
construction~\cite{sherstov07ac-majmaj,sherstov07quantum} of an
$\classAC^{0}$ circuit with exponentially small discrepancy and hence
large communication complexity in nearly every model. That discrepancy
result was used in~\cite{sherstov07ac-majmaj} to show the optimality
of Allender's classic simulation of $\classAC^{0}$ by majority circuits,
solving the open problem~\cite{krause94depth2mod} on the relation
between the two circuit classes. Subsequent work~\cite{arkadev07multiparty,beame-huyn-ngoc09multiparty-focs,sherstov12mdisj,sherstov13directional}
resolved other questions in communication complexity and circuit complexity
related to constant-depth circuits by generalizing the threshold degree
method of~\cite{sherstov07ac-majmaj,sherstov07quantum}.

Sign-representing polynomials also paved the way for \emph{algorithmic}
breakthroughs in the study of constant-depth circuits. Specifically,
any function of threshold degree $d$ can be viewed as a halfspace
in $\binom{n}{0}+\binom{n}{1}+\cdots+\binom{n}{d}$ dimensions, corresponding
to the monomials in a sign-representation of $f$. As a result, a
class of functions of threshold degree at most $d$ can be learned
in the standard PAC model under arbitrary distributions in time polynomial
in $\binom{n}{0}+\binom{n}{1}+\cdots+\binom{n}{d}.$ Klivans and Servedio~\cite{KS01dnf}
used this threshold degree approach to give what is currently the
fastest algorithm for learning polynomial-size DNF formulas, with
running time $\exp(\tilde{O}(n^{1/3}))$. Another learning-theoretic
breakthrough based on threshold degree is the fastest algorithm for
learning Boolean formulas, obtained by O'Donnell and Servedio~\cite{odonnell03degree}
for formulas of constant depth and by Ambainis et al.~\cite{ACRSZ07nand}
for arbitrary depth. Their algorithm runs in time $\exp(\tilde{O}(n^{(2^{k-1}-1)/(2^{k}-1)}))$
for formulas of size $n$ and constant depth~$k$, and in time $\exp(\tilde{O}(\sqrt{n}))$
for formulas of unbounded depth. In both cases, the bound on the running
time follows from the corresponding upper bound on the threshold degree.

A far-reaching generalization of threshold degree is the matrix-analytic
notion of \emph{sign-rank}, which allows sign-representation out of
arbitrary low-dimensional subspaces rather than the subspace of low-degree
polynomials. The contribution of this paper is to prove essentially
optimal lower bounds on the threshold degree and sign-rank of $\classAC^{0}$,
which in turn imply lower bounds on other fundamental complexity measures
of interest in communication complexity and learning theory. In the
remainder of this section, we give a detailed overview of the previous
work, present our main results, and discuss our proofs. 

\begin{table}[b]
~\\
~\\
~\\

\begin{tabular}{ll>{\raggedright}p{1mm}l>{\raggedright}p{3mm}ll}
\noalign{\vskip\doublerulesep}
 & \textbf{Depth} &  & \textbf{Threshold degree} &  & \textbf{Reference} & \tabularnewline[\doublerulesep]
\hline 
\noalign{\vskip\doublerulesep}
\hline 
\noalign{\vskip5pt}
 & $2$ &  & $\Omega(n^{1/3})$ &  & Minsky and Papert~\cite{minsky88perceptrons} & \tabularnewline[\doublerulesep]
\noalign{\vskip\doublerulesep}
 & $2$ &  & $O(n^{1/3}\log n)$ &  & Klivans and Servedio~\cite{KS01dnf} & \tabularnewline[\doublerulesep]
\noalign{\vskip\doublerulesep}
 & $k+2$ &  & $\Omega(n^{1/3}\log^{2k/3}n)$ &  & O'Donnell and Servedio~\cite{odonnell03degree} & \tabularnewline[\doublerulesep]
\noalign{\vskip\doublerulesep}
 & $k$ &  & $\Omega(n^{\frac{k-1}{2k-1}})$ &  & Sherstov~\cite{sherstov14sign-deg-ac0} & \tabularnewline[\doublerulesep]
\noalign{\vskip\doublerulesep}
 & $4$ &  & $\Omega(\sqrt{n})$ &  & Sherstov~\cite{sherstov15asymmetry} & \tabularnewline[\doublerulesep]
\noalign{\vskip\doublerulesep}
 & $3$ &  & $\tilde{\Omega}(\sqrt{n})$ &  & Bun and Thaler~\cite{BT18ac0-large-error} & \tabularnewline[\doublerulesep]
\noalign{\vskip\doublerulesep}
 & $k$ &  & $\tilde{\Omega}(n^{\frac{k-1}{k+1}})$ &  & This paper & \tabularnewline[\doublerulesep]
\hline 
 &  &  &  &  &  & \tabularnewline
\end{tabular}

\caption{\label{tab:threshold-degree}Known bounds on the maximum threshold
degree of $\wedge,\vee,\neg$-circuits of polynomial size and constant
depth. In all bounds, $n$ denotes the number of variables, and $k$
denotes an arbitrary positive integer.}
\end{table}

\subsection{Threshold degree of AC\protect\textsuperscript{0}}

Determining the maximum threshold degree of an $\classAC^{0}$ circuit
in $n$ variables is a longstanding open problem in the area. It is
motivated by algorithmic and complexity-theoretic applications~\cite{KS01dnf,odonnell03degree,klivans-servedio06decision-lists,RS07dc-dnf,bun-thaler13amplification},
in addition to being a natural question in its own right. Table~\ref{tab:threshold-degree}
gives a quantitative summary of the results obtained to date. In their
seminal monograph, Minsky and Papert~\cite{minsky88perceptrons}
proved a lower bound of $\Omega(n^{1/3})$ on the threshold degree
of the following DNF formula in $n$ variables: 
\[
f(x)=\bigwedge_{i=1}^{n^{1/3}}\;\bigvee_{j=1}^{n^{2/3}}x_{i,j}.
\]
Three decades later, Klivans and Servedio~\cite{KS01dnf} obtained
an $O(n^{1/3}\log n)$ upper bound on the threshold degree of any
polynomial-size DNF formula in $n$ variables, essentially matching
Minsky and Papert's result and resolving the problem for depth~$2$.
Determining the threshold degree of circuits of depth $k\geq3$ proved
to be challenging. The only upper bound known to date is the trivial
$O(n),$ which follows directly from the definition of threshold degree.
In particular, it is consistent with our knowledge that there are
$\classAC^{0}$ circuits with linear threshold degree. On the lower
bounds side, the only progress for a long time was due to O'Donnell
and Servedio~\cite{odonnell03degree}, who constructed for any $k\geq1$
a circuit of depth $k+2$ with threshold degree $\Omega(n^{1/3}\log^{2k/3}n).$
The authors of~\cite{odonnell03degree} formally posed the problem
of obtaining a polynomial improvement on Minsky and Papert's lower
bound. Such an improvement was obtained in~\cite{sherstov14sign-deg-ac0},
with a threshold degree lower bound of $\Omega(n^{(k-1)/(2k-1)})$
for circuits of depth~$k.$ A polynomially stronger result was obtained
in~\cite{sherstov15asymmetry}, with a lower bound of $\Omega(\sqrt{n})$
on the threshold degree of an explicit circuit of depth~$4$. Bun
and Thaler~\cite{BT18ac0-large-error} recently used a different,
depth-$3$ circuit to give a much simpler proof of the $\tilde{\Omega}(\sqrt{n})$
lower bound for $\classAC^{0}$. We obtain a quadratically stronger,
and near-optimal, lower bound on the threshold degree of $\classAC^{0}$.
\begin{thm}
\label{thm:MAIN-degthr-ac0}Let $k\geq1$ be a fixed integer. Then
there is an $($explicitly given$)$ Boolean circuit family $\{f_{n}\}_{n=1}^{\infty},$
where $f_{n}\colon\zoon\to\zoo$ has polynomial size, depth $k,$
and threshold degree 
\[
\degthr(f_{n})=\Omega\left(n^{\frac{k-1}{k+1}}\cdot(\log n)^{-\frac{1}{k+1}\lceil\frac{k-2}{2}\rceil\lfloor\frac{k-2}{2}\rfloor}\right).
\]
Moreover, $f_{n}$ has bottom fan-in $O(\log n)$ for all $k\ne2.$
\end{thm}

\noindent For large $k,$ Theorem~\ref{thm:MAIN-degthr-ac0} essentially
matches the trivial upper bound of $n$ on the threshold degree of
any function. For any fixed depth $k,$ Theorem~\ref{thm:MAIN-degthr-ac0}
subsumes all\emph{ }previous lower bounds on the threshold degree
of $\classAC^{0},$ with a polynomial improvement starting at depth
$k=4.$ In particular, the lower bounds due to Minsky and Papert~\cite{minsky88perceptrons}
and Bun and Thaler~\cite{BT18ac0-large-error} are subsumed as the
special cases $k=2$ and $k=3$, respectively. From a computational
learning perspective, Theorem~\ref{thm:MAIN-degthr-ac0} definitively
rules out the threshold degree approach to learning constant-depth
circuits.

\noindent 
\begin{table}[b]
~\\
~\\
~\\

\begin{tabular}{ll>{\raggedright}p{1mm}l>{\raggedright}p{3mm}ll}
\noalign{\vskip\doublerulesep}
 & \textbf{Depth} &  & \textbf{Sign-rank} &  & \textbf{Reference} & \tabularnewline[\doublerulesep]
\hline 
\noalign{\vskip\doublerulesep}
\hline 
\noalign{\vskip5pt}
 & $3$ &  & $\exp(\Omega(n^{1/3}))$ &  & Razborov and Sherstov~\cite{RS07dc-dnf} & \tabularnewline[\doublerulesep]
\noalign{\vskip\doublerulesep}
 & $3$ &  & $\exp(\tilde{\Omega}(n^{2/5}))$ &  & Bun and Thaler~\cite{BT16sign-rank-ac0} & \tabularnewline[\doublerulesep]
\noalign{\vskip\doublerulesep}
 & 7 &  & $\exp(\tilde{\Omega}(\sqrt{n}))$ &  & Bun and Thaler~\cite{BT18ac0-large-error} & \tabularnewline[\doublerulesep]
\noalign{\vskip\doublerulesep}
 & $3k$ &  & $\exp(\tilde{\Omega}(n^{1-\frac{1}{k+1}}))$ &  & This paper & \tabularnewline[\doublerulesep]
\noalign{\vskip\doublerulesep}
 & $3k+1$ &  & $\exp(\tilde{\Omega}(n^{1-\frac{1}{k+1.5}}))$ &  & This paper & \tabularnewline[\doublerulesep]
\hline 
 &  &  &  &  &  & \tabularnewline
\end{tabular}

\caption{\label{tab:sign-rank}Known lower bounds on the maximum sign-rank
of $\wedge,\vee,\neg$-circuits $F\colon\zoon\times\zoon\to\zoo$
of polynomial size and constant depth. In all bounds, $k$ denotes
an arbitrary positive integer.}
\end{table}

\subsection{Sign-rank of AC\protect\textsuperscript{0}}

The \emph{sign-rank} of a matrix $A=[A_{ij}]$ without zero entries,
denoted $\srank(A),$ is the least rank of a real matrix $M=[M_{ij}]$
with $\sign M_{ij}=\sign A_{ij}$ for all $i,j.$ In other words,
the sign-rank of $A$ is the minimum rank of a matrix that can be
obtained by making arbitrary sign-preserving changes to the entries
of $A$. The sign-rank of a Boolean function $F\colon\zoon\times\zoon\to\zoo$
is defined in the natural way as the sign-rank of the matrix $[(-1)^{F(x,y)}]_{x,y}.$
In particular, the sign-rank of $F$ is an integer between $1$ and
$2^{n}$. This fundamental notion has been studied in contexts as
diverse as matrix analysis, communication complexity, circuit complexity,
and learning theory; see~\cite{RS07dc-dnf} for a bibliographic overview.
 To a complexity theorist, sign-rank is a vastly more challenging
quantity to analyze than threshold degree. Indeed, a sign-rank lower
bound rules out sign-representation out of \emph{every }linear subspace
of given dimension, whereas a threshold degree lower bound rules out
sign-representation specifically by linear combinations of monomials
up to a given degree. 

Unsurprisingly, progress in understanding sign-rank has been slow
and difficult. No nontrivial lower bounds were known for any explicit
matrices until the breakthrough work of Forster~\cite{forster02linear},
who proved strong lower bounds on the sign-rank of Hadamard matrices
and more generally all sign matrices with small spectral norm. The
sign-rank of constant-depth circuits $F\colon\zoon\times\zoon\to\zoo$
has since seen considerable work, as summarized in Table~\ref{tab:sign-rank}.
The first exponential lower bound on the sign-rank of an $\classAC^{0}$
circuit was obtained by Razborov and Sherstov~\cite{RS07dc-dnf},
solving a~$22$-year-old problem due to Babai, Frankl, and Simon~\cite{BFS86cc}.
The authors of~\cite{RS07dc-dnf} constructed a polynomial-size circuit
of depth~$3$ with sign-rank $\exp(\Omega(n^{1/3}))$. In follow-up
work, Bun and Thaler~\cite{BT16sign-rank-ac0} constructed a polynomial-size
circuit of depth~$3$ with sign-rank $\exp(\tilde{\Omega}(n^{2/5}))$.
A more recent and incomparable result, also due to Bun and Thaler~\cite{BT18ac0-large-error},
is a sign-rank lower bound of $\exp(\tilde{\Omega}(\sqrt{n}))$ for
a circuit of polynomial size and depth~$7$. No nontrivial upper
bounds are known on the sign-rank of $\classAC^{0}$. Closing this
gap between the best lower bound of $\exp(\tilde{\Omega}(\sqrt{n}))$
and the trivial upper bound of $2^{n}$ has been a challenging open
problem. We solve  this problem almost completely, by constructing
for any $\epsilon>0$ a constant-depth circuit with sign-rank $\exp(\Omega(n^{1-\epsilon})).$
In quantitative detail, our results on the sign-rank of $\classAC^{0}$
are the following two theorems. 
\begin{thm}
\label{thm:MAIN-sign-rank-3k}Let $k\geq1$ be a given integer. Then
there is an $($explicitly given$)$ Boolean circuit family $\{F_{n}\}_{n=1}^{\infty},$
where $F_{n}\colon\zoon\times\zoon\to\zoo$ has polynomial size, depth
$3k,$ and sign-rank
\[
\srank(F_{n})=\exp\left(\Omega\left(n^{1-\frac{1}{k+1}}\cdot(\log n)^{-\frac{k(k-1)}{2(k+1)}}\right)\right).
\]
\end{thm}

\noindent As a companion result, we prove the following qualitatively
similar but quantitatively incomparable theorem. 
\begin{thm}
\label{thm:MAIN-sign-rank-3k-plus-1}Let $k\geq1$ be a given integer.
Then there is an $($explicitly given$)$ Boolean circuit family $\{G_{n}\}_{n=1}^{\infty},$
where $G_{n}\colon\zoon\times\zoon\to\zoo$ has polynomial size, depth
$3k+1,$ and sign-rank
\[
\srank(G_{n})=\exp\left(\Omega\left(n^{1-\frac{1}{k+1.5}}\cdot(\log n)^{-\frac{k^{2}}{2k+3}}\right)\right).
\]
\end{thm}

\noindent For large $k$, the lower bounds of Theorems~\ref{thm:MAIN-sign-rank-3k}
and~\ref{thm:MAIN-sign-rank-3k-plus-1} approach the trivial upper
bound of $2^{n}$ on the sign-rank of any Boolean function $\zoon\times\zoon\to\zoo$.
For any given depth, Theorems~\ref{thm:MAIN-sign-rank-3k} and~\ref{thm:MAIN-sign-rank-3k-plus-1}
subsume all previous lower bounds on the sign-rank of $\classAC^{0},$
with a strict improvement starting at depth~$3$. From a computational
learning perspective, Theorems~\ref{thm:MAIN-sign-rank-3k} and~\ref{thm:MAIN-sign-rank-3k-plus-1}
state that $\classAC^{0}$ has near-maximum \emph{dimension complexity}~\cite{sherstov07halfspace-mat,sherstov07cc-prod-nonprod,RS07dc-dnf,BT18ac0-large-error},
namely, $\exp(\Omega(n^{1-\epsilon}))$ for any constant $\epsilon>0.$\emph{
}This rules out the possibility of learning $\classAC^{0}$ circuits
via dimension complexity~\cite{RS07dc-dnf}, a far-reaching generalization
of the threshold degree approach from the monomial basis to arbitrary
bases.

\subsection{Communication complexity}

Theorems~\ref{thm:MAIN-degthr-ac0}\textendash \ref{thm:MAIN-sign-rank-3k-plus-1}
imply strong new lower bounds on the communication complexity of $\classAC^{0}$.
We adopt the standard randomized model of Yao~\cite{ccbook}, with
players Alice and Bob and a Boolean function $F\colon X\times Y\to\zoo.$
On input $(x,y)\in X\times Y,$ Alice and Bob receive the arguments
$x$ and $y,$ respectively, and communicate back and forth according
to an agreed-upon protocol. Each player privately holds an unlimited
supply of uniformly random bits that he or she can use when deciding
what message to send at any given point in the protocol. The \emph{cost}
of a protocol is the total number of bits communicated in a worst-case
execution. The\emph{ $\epsilon$-error randomized communication complexity
$R_{\epsilon}(F)$} of $F$ is the least cost of a protocol that computes
$F$ with probability of error at most $\epsilon$ on every input.

Of particular interest to us are communication protocols with error
probability close to that of random guessing, $1/2.$ There are two
standard ways to formalize the complexity of a communication problem
$F$ in this setting, both inspired by probabilistic polynomial time
$\mathsf{PP}$ for Turing machines: 
\[
\upp(F)=\inf_{0\leq\epsilon<1/2}R_{\epsilon}(F)
\]
and
\[
\pp(F)=\inf_{0\leq\epsilon<1/2}\left\{ R_{\epsilon}(F)+\log_{2}\left(\frac{1}{\frac{1}{2}-\epsilon}\right)\right\} .
\]
The former quantity, introduced by Paturi and Simon~\cite{paturi86cc},
is called the \emph{communication complexity of $F$ with unbounded
error}, in reference to the fact that the error probability can be
arbitrarily close to $1/2.$ The latter quantity is called the \emph{communication
complexity of $F$ with weakly unbounded error}. Proposed by Babai
et al.~\cite{BFS86cc}, it features an additional penalty term that
depends on the error probability. It is clear that 
\[
\upp(F)\leq\pp(F)\leq n+2
\]
for every communication problem $F\colon\zoon\times\zoon\to\zoo$,
with an exponential gap achievable between the two complexity measures~\cite{buhrman07pp-upp,sherstov07halfspace-mat}.
These two models occupy a special place in the study of communication
because they are more powerful than any other standard model (deterministic,
nondeterministic, randomized, quantum with or without entanglement).
Moreover, unbounded-error protocols represent a frontier in communication
complexity theory in that they are the most powerful protocols for
which explicit lower bounds are currently known. Our results imply
that even for such protocols, $\classAC^{0}$ has near-maximal communication
complexity. 

To begin with, combining Theorem~\ref{thm:MAIN-degthr-ac0} with
the \emph{pattern matrix method}~\cite{sherstov07ac-majmaj,sherstov07quantum}
gives:
\begin{thm}
\label{thm:MAIN-disc-ac0}Let $k\geq3$ be a fixed integer. Then there
is an $($explicitly given$)$ Boolean circuit family $\{F_{n}\}_{n=1}^{\infty},$
where $F_{n}\colon\zoon\times\zoon\to\zoo$ has polynomial size, depth
$k,$ communication complexity 
\[
\pp(F_{n})=\Omega\left(n^{\frac{k-1}{k+1}}\cdot(\log n)^{-\frac{1}{k+1}\lceil\frac{k-2}{2}\rceil\lfloor\frac{k-2}{2}\rfloor}\right)
\]
and discrepancy
\[
\disc(F_{n})=\exp\left(-\Omega\left(n^{\frac{k-1}{k+1}}\cdot(\log n)^{-\frac{1}{k+1}\lceil\frac{k-2}{2}\rceil\lfloor\frac{k-2}{2}\rfloor}\right)\right).
\]
\end{thm}

\noindent \emph{Discrepancy} is a combinatorial complexity measure
of interest in communication complexity theory and other research
areas; see Section~\ref{subsec:Discrepancy-and-sign-rank} for a
formal definition. As $k$ grows, the bounds of Theorem~\ref{thm:MAIN-disc-ac0}
approach the best possible bounds for any communication problem $F_{n}\colon\zoon\times\zoon\to\zoo.$
The same \emph{qualitative} behavior was achieved in previous work
by Bun and Thaler~\cite{BT18ac0-large-error}, who constructed, for
any constant $\epsilon>0$, a constant-depth circuit $F_{n}\colon\zoon\times\zoon\to\zoo$
with communication complexity $\pp(F_{n})=\Omega(n^{1-\epsilon})$
and discrepancy $\disc(F_{n})=\exp(-\Omega(n^{1-\epsilon}))$. Theorem~\ref{thm:MAIN-disc-ac0}
strictly subsumes the result of Bun and Thaler~\cite{BT18ac0-large-error}
and all other prior work on the discrepancy and $\pp$-complexity
of constant-depth circuits~\cite{sherstov07ac-majmaj,buhrman07pp-upp,sherstov07quantum,sherstov14sign-deg-ac0,sherstov15asymmetry}.
For any fixed depth greater than~$3$, the bounds of Theorem~\ref{thm:MAIN-disc-ac0}
are a polynomial improvement in $n$ over all previous work. We further
show that Theorem~\ref{thm:MAIN-disc-ac0} carries over to the \emph{number-on-the-forehead
model}, the strongest formalism of multiparty communication. This
result, presented in detail in Section~\ref{subsec:Results-for-AC},
uses the multiparty version~\cite{sherstov13directional} of the
pattern matrix method.

Our work also gives near-optimal lower bounds for $\classAC^{0}$
in the much more powerful unbounded-error model. Specifically, it
is well-known~\cite{paturi86cc} that the unbounded-error communication
complexity of any Boolean function $F\colon X\times Y\to\zoo$ coincides
up to an additive constant with the logarithm of the sign-rank of
$F.$ As a result, Theorems~\ref{thm:MAIN-sign-rank-3k} and~\ref{thm:MAIN-sign-rank-3k-plus-1}
imply: 
\begin{thm}
\label{thm:MAIN-unbounded}Let $k\geq1$ be a given integer. Let $\{F_{n}\}_{n=1}^{\infty}$
and $\{G_{n}\}_{n=1}^{\infty}$ be the polynomial-size circuit families
of depth $3k$ and $3k+1,$ respectively, constructed in Theorems~\emph{\ref{thm:MAIN-sign-rank-3k}}
and~\emph{\ref{thm:MAIN-sign-rank-3k-plus-1}}. Then
\begin{align*}
\upp(F_{n}) & =\Omega\left(n^{1-\frac{1}{k+1}}\cdot(\log n)^{-\frac{k(k-1)}{2(k+1)}}\right),\\
\upp(G_{n}) & =\Omega\left(n^{1-\frac{1}{k+1.5}}\cdot(\log n)^{-\frac{k^{2}}{2k+3}}\right).
\end{align*}
\end{thm}

\noindent For large $k,$ the lower bounds of Theorem~\ref{thm:MAIN-unbounded}
essentially match the trivial upper bound of $n+1$ on the unbounded-error
communication complexity of any function $F\colon\zoon\times\zoon\to\zoo.$
Theorem~\ref{thm:MAIN-unbounded} strictly subsumes all previous
lower bounds on the unbounded-error communication complexity of $\classAC^{0}$,
with a polynomial improvement for any depth greater than~$2$. The
best lower bound on the unbounded-error communication complexity of
$\classAC^{0}$ prior to our work was $\tilde{\Omega}(\sqrt{n})$
for a circuit of depth~$7$, due to Bun and Thaler~\cite{BT18ac0-large-error}.
Finally, we remark that Theorem~\ref{thm:MAIN-unbounded} gives essentially
the strongest possible separation of the communication complexity
classes $\PH$ and $\UPP$. We refer the reader to the work of Babai
et al.~\cite{BFS86cc} for definitions and detailed background on
these classes.

Qualitatively, Theorem~\ref{thm:MAIN-unbounded} is stronger than
Theorem~\ref{thm:MAIN-disc-ac0} because communication protocols
with unbounded error are significantly more powerful than those with
weakly unbounded error. On the other hand, Theorem~\ref{thm:MAIN-disc-ac0}
is stronger quantitatively for any fixed depth and has the additional
advantage of generalizing to the multiparty setting.

\subsection{Threshold weight and threshold density}

By well-known reductions, Theorem~\ref{thm:MAIN-degthr-ac0} implies
a number of other lower bounds for the representation of $\classAC^{0}$
circuits by polynomials. For the sake of completeness, we mention
two such consequences. The \emph{threshold density }of a Boolean function
$f\colon\zoon\to\zoo,$ denoted $\dns(f),$ is the minimum size of
a set family $\Scal\subseteq\Pcal(\{1,2,\ldots,n\})$ such that 
\[
\sign\left(\sum_{S\in\Scal}\lambda_{S}(-1)^{\sum_{i\in S}x_{i}}\right)\equiv(-1)^{f(x)}
\]
for some reals $\lambda_{S}.$ A related complexity measure is \emph{threshold
weight}, denoted $W(f)$ and defined as the minimum sum $\sum_{S\subseteq\oneton}|\lambda_{S}|$
over all integers $\lambda_{S}$ such that 
\[
\sign\left(\sum_{S\subseteq\{1,2,\ldots,n\}}\lambda_{S}(-1)^{\sum_{i\in S}x_{i}}\right)\equiv(-1)^{f(x)}.
\]
It is not hard to see that the threshold density and threshold weight
of $f$ correspond to the minimum size of a threshold-of-parity and
majority-of-parity circuit for $f,$ respectively. The definitions
imply that $\dns(f)\leq W(f)$ for every $f,$ and a little more thought
reveals that $1\leq\dns(f)\leq2^{n}$ and $1\leq W(f)\leq4^{n}.$
These complexity measures have seen extensive work, motivated by applications
to computational learning and circuit complexity. For a bibliographic
overview, we refer the reader to~\cite[Section~8.2]{sherstov14sign-deg-ac0}.

Krause and Pudlák~\cite[Proposition~2.1]{krause94depth2mod} gave
an ingenious method for transforming threshold degree lower bounds
into lower bounds on threshold density and thus also threshold weight.
Specifically, let $f\colon\zoon\to\zoo$ be a Boolean function of
interest. The authors of~\cite{krause94depth2mod} considered the
related function $F\colon(\zoo^{n})^{3}\to\zoo$ given by $F(x,y,z)=f(\dots,(\overline{z_{i}}\wedge x_{i})\vee(z_{i}\wedge y_{i}),\dots)$,
and proved that $\dns(F)\geq2^{\degthr(f)}.$ In this light, Theorem~\ref{thm:MAIN-degthr-ac0}
implies that the threshold density of $\classAC^{0}$ is $\exp(\Omega(n^{1-\epsilon}))$
for any constant $\epsilon>0$:
\begin{cor}
\label{cor:MAIN-thr-density}Let $k\geq3$ be a fixed integer. Then
there is an $($explicitly given$)$ Boolean circuit family $\{F_{n}\}_{n=1}^{\infty},$
where $F_{n}\colon\zoon\to\zoo$ has polynomial size and depth $k$
and satisfies 
\begin{align*}
W(F_{n}) & \geq\dns(F_{n})\\
 & =\exp\left(\Omega\left(n^{\frac{k-1}{k+1}}\cdot(\log n)^{-\frac{1}{k+1}\lceil\frac{k-2}{2}\rceil\lfloor\frac{k-2}{2}\rfloor}\right)\right).
\end{align*}
\end{cor}

\noindent For large $k,$ the lower bounds on the threshold weight
and density in Corollary~\ref{cor:MAIN-thr-density} essentially
match the trivial upper bounds. Observe that the circuit family $\{F_{n}\}_{n=1}^{\infty}$
of Corollary~\ref{cor:MAIN-thr-density} has the same depth as the
circuit family $\{f_{n}\}_{n=1}^{\infty}$ of Theorem~\ref{thm:MAIN-degthr-ac0}.
This is because $f_{n}$ has bottom fan-in $O(\log n)$, and thus
the Krause-Pudlák transformation $f_{n}\mapsto F_{n}$ can be ``absorbed''
into the bottom two levels of $f_{n}$. Corollary~\ref{cor:MAIN-thr-density}
subsumes all previous lower bounds~\cite{krause94depth2mod,bun-thaler13amplification,sherstov14sign-deg-ac0,sherstov15asymmetry,BT18ac0-large-error}
on the threshold weight and density of $\classAC^{0},$ with a polynomial
improvement for every $k\geq4.$ The improvement is particularly noteworthy
in the case of threshold density, where the best previous lower bound~\cite{sherstov15asymmetry,BT18ac0-large-error}
was $\exp(\Omega(\sqrt{n}))$. 

\subsection{Previous approaches}

In the remainder of this section, we discuss our proofs of Theorems~\ref{thm:MAIN-degthr-ac0}\textendash \ref{thm:MAIN-sign-rank-3k-plus-1}.
The notation that we use here is standard, and we defer its formal
review to Section~\ref{sec:Preliminaries}. We start with necessary
approximation-theoretic background, then review relevant previous
work, and finally contrast it with the approach of this paper. To
sidestep minor technicalities, we will represent Boolean functions
in this overview as mappings $\pomon\to\pomo.$ We alert the reader
that we will revert to the standard $\zoon\to\zoo$ representation
starting with Section~\ref{sec:Preliminaries}.

\subsubsection*{Background}

Recall that our results concern the sign-representation of Boolean
functions and matrices. To properly set the stage for our proofs,
however, we need to consider the more general notion of pointwise
approximation~\cite{nisan-szegedy94degree}. Let $f\colon\pomon\to\pomo$
be a Boolean function of interest. The \emph{$\epsilon$-approximate
degree of $f,$} denoted $\degeps(f),$ is the minimum degree of a
real polynomial that approximates $f$ within~$\epsilon$ pointwise:
$\degeps(f)=\min\{\deg p\colon\|f-p\|_{\infty}\leq\epsilon\}.$ The
regimes of most interest are \emph{bounded-error approximation}, corresponding
to constants $\epsilon\in(0,1)$; and \emph{large-error approximation},
corresponding to $\epsilon=1-o(1).$ In the former case, the choice
of error parameter $\epsilon\in(0,1)$ is immaterial and affects the
approximate degree of a Boolean function by at most a multiplicative
constant. It is clear that pointwise approximation is a stronger requirement
than sign-representation, and thus $\degthr(f)\leq\degeps(f)$ for
all $0\leq\epsilon<1.$ A moment's thought reveals that threshold
degree is in fact the limiting case of $\epsilon$-approximate degree
as the error parameter approaches $1$:
\begin{equation}
\degthr(f)=\lim_{\epsilon\nearrow1}\degeps(f).\label{eq:degthr-limiting}
\end{equation}

Both approximate degree and threshold degree have dual characterizations~\cite{sherstov07quantum},
obtained by appeal to linear programming duality. Specifically, $\degeps(f)\geq d$
if and only if there is a function $\phi\colon\pomon\to\Re$ with
the following two properties: $\langle\phi,f\rangle>\epsilon\|\phi\|_{1}$;
and $\langle\phi,p\rangle=0$ for every polynomial $p$ of degree
less than $d$. Rephrasing, $\phi$ must have large correlation with
$f$ but zero correlation with every low-degree polynomial. By weak
linear programming duality, $\phi$ constitutes a proof that $\degeps(f)\geq d$
and for that reason is said to \emph{witness} the lower bound $\degeps(f)\geq d.$
In view of~(\ref{eq:degthr-limiting}), this discussion generalizes
to threshold degree. The dual characterization here states that $\degthr(f)\geq d$
if and only if there is a nonzero function $\phi\colon\pomon\to\Re$
with the following two properties: $\phi(x)f(x)\geq0$ for all $x;$
and $\langle\phi,p\rangle=0$ for every polynomial $p$ of degree
less than $d$. In this dual characterization, $\phi$ agrees in sign
with $f$ and is additionally orthogonal to polynomials of degree
less than $d.$ The sign-agreement property can be restated in terms
of correlation, as $\langle\phi,f\rangle=\|\phi\|_{1}$. As before,
$\phi$ is called a threshold degree \emph{witness~}for~$f.$

What distinguishes the dual characterizations of approximate degree
and threshold degree is how the dual object $\phi$ relates to $f$.
Specifically, a threshold degree witness must agree in sign with $f$
at every point. An approximate degree witness, on the other hand,
need only exhibit such sign-agreement with $f$ at \emph{most} points,
in that the points where the sign of $\phi$ is correct should account
for most of the $\ell_{1}$ norm of $\phi.$ As a result, constructing
dual objects for threshold degree is significantly more difficult
than for approximate degree. This difficulty is to be expected because
 the gap between threshold degree and approximate degree can be arbitrary,
e.g., $1$~versus~$\Theta(n)$ for the majority function on $n$
bits~\cite{paturi92approx}.

\subsubsection*{Hardness amplification via block-composition}

Much of the recent work on approximate degree and threshold degree
is concerned with composing functions in ways that amplify their hardness.
Of particular significance here is \emph{block-composition}, defined
for functions $f\colon\pomon\to\pomo$ and $g\colon X\to\pomo$ as
the Boolean function $f\circ g\colon X^{n}\to\pomo$ given by $\text{(}f\circ g)(x_{1},\ldots,x_{n})=f(g(x_{1}),\ldots,g(x_{n})).$
Block-composition works particularly well for threshold degree. To
use an already familiar example, the block-composition $\AND_{n^{1/3}}\circ\OR_{n^{2/3}}$
has threshold degree $\Omega(n^{1/3})$ whereas the constituent functions
$\AND_{n^{1/3}}$ and $\OR_{n^{2/3}}$ have threshold degree~$1.$
As a more extreme example, Sherstov~\cite{sherstov09opthshs} obtained
a lower bound of $\Omega(n)$ on the threshold degree of the conjunction
$h_{1}\wedge h_{2}$ of two halfspaces $h_{1},h_{2}\colon\zoon\to\zoo$,
each of which by definition has threshold degree~$1$. The fact that
threshold degree can increase spectacularly under block-composition
is the basis of much previous work, including the best previous lower
bounds~\cite{sherstov14sign-deg-ac0,sherstov15asymmetry} on the
threshold degree of $\classAC^{0}.$ Apart from threshold degree,
block-composition has yielded strong results for approximate degree
in various error regimes, including direct sum theorems~\cite{sherstov09hshs},
direct product theorems~\cite{sherstov11quantum-sdpt}, and error
amplification results~\cite{sherstov11quantum-sdpt,bun-thaler13amplification,thaler14omb,bun-thaler16approx-degree-depth3}. 

How, then, does one prove lower bounds on the threshold degree or
approximate degree of a composed function $f\circ g$? It is here
that the dual characterizations take center stage: they make it possible
to prove lower bounds \emph{algorithmically}, by constructing the
corresponding dual object for the composed function. Such algorithmic
proofs run the gamut in terms of technical sophistication, from straightforward
to highly technical, but they have some structure in common. In most
cases, one starts by obtaining dual objects $\phi$ and $\psi$ for
the constituent functions $f$ and $g$, respectively, either by direct
construction or by appeal to linear programming duality. They are
then combined to yield a dual object $\Phi$ for the composed function,
using \emph{dual block-composition~\cite{sherstov09hshs,lee09formulas}}:
\begin{equation}
\!\!\!\Phi(x_{1},x_{2},\ldots,x_{n})=\phi(\sign\psi(x_{1}),\ldots,\sign\psi(x_{n}))\prod_{i=1}^{n}|\psi(x_{i})|.\label{eq:dual-block-compose}
\end{equation}
This composed dual object often requires additional work to ensure
sign-agreement or correlation with the composed Boolean function.
Among the generic tools available to assist in this process is a ``corrector''
object $\zeta$ due to Razborov and Sherstov~\cite{RS07dc-dnf},
with the following four properties: (i)~$\zeta$ is orthogonal to
low-degree polynomials; (ii)~$\zeta$ takes on~$1$ at a prescribed
point of the hypercube; (iii)~$\zeta$ is bounded on inputs of low
Hamming weight; and (iv)~$\zeta$ vanishes on all other points of
the hypercube. Using the Razborov\textendash Sherstov object, suitably
shifted and scaled, one can surgically correct the behavior of a given
dual object $\Phi$ on a substantial fraction of inputs, thus modifying
its metric properties without affecting its orthogonality to low-degree
polynomials. This technique has played an important role in recent
work, e.g.,~\cite{BT16sign-rank-ac0,bun-thaler17adeg-ac0,BKT17poly-strikes-back,BT18ac0-large-error}.

\subsubsection*{Hardness amplification for approximate degree}

While block-composition has produced a treasure trove of results on
polynomial representations of Boolean functions, it is of limited
use when it comes to constructing functions with high \emph{bounded-error}
approximate degree. To illustrate the issue, consider arbitrary functions
$f\colon\pomo^{n_{1}}\to\pomo$ and $g\colon\pomo^{n_{2}}\to\pomo$
with $1/3$-approximate degrees $n_{1}^{\alpha_{1}}$ and $n_{2}^{\alpha_{2}},$
respectively, for some $0<\alpha_{1}<1$ and $0<\alpha_{2}<1.$ It
is well-known~\cite{sherstov12noisy} that the composed function
$f\circ g$ on $n_{1}n_{2}$ variables has $1/3$-approximate degree
$O(n_{1}^{\alpha_{1}}n_{2}^{\alpha_{2}})=O(n_{1}n_{2})^{\max\{\alpha_{1},\alpha_{2}\}}.$
This means that relative to the new number of variables, the block-composed
function $f\circ g$ is asymptotically no harder to approximate to
bounded error than the constituent functions $f$ and $g$. In particular,
one cannot use block-composition to transform functions on $n$ bits
with $1/3$-approximate degree at most $n^{\alpha}$ into functions
on $N\geq n$ bits with $1/3$-approximate degree $\omega(N^{\alpha}).$ 

Until recently, the best lower bound on the bounded-error approximate
degree of $\classAC^{0}$ was $\Omega(n^{2/3})$, due to Aaronson
and Shi~\cite{aaronson-shi04distinctness}. Breaking this $n^{2/3}$
barrier was a fundamental problem in its own right, in addition to
being a hard prerequisite for \emph{threshold} degree lower bounds
for $\classAC^{0}$ better than $\Omega(n^{2/3}).$ This barrier was
overcome in a brilliant paper of Bun and Thaler~\cite{bun-thaler17adeg-ac0},
who proved, for any constant $\epsilon>0,$ an $\Omega(n^{1-\epsilon})$
lower bound on the $1/3$-approximate degree of $\classAC^{0}$. Their
hardness amplification for approximate degree works as follows. Let
$f\colon\pomo^{n}\to\pomo$ be given, with $1/3$-approximate degree
$n^{\alpha}$ for some $0\leq\alpha<1$. Bun and Thaler consider the
block-composition $F=f\circ\AND_{\Theta(\log m)}\circ\OR_{m}$, for
an appropriate parameter $m=\poly(n).$ As shown in earlier work~\cite{sherstov09hshs,bun-thaler13amplification}
on approximate degree, dual block-composition witnesses the lower
bound $\deg_{1/3}(F)=\Omega(\deg_{1/3}(\OR_{m})\deg_{1/3}(f))=\Omega(\sqrt{m}\deg_{1/3}(f)).$
Next, Bun and Thaler make the crucial observation that the dual object
for $\OR_{m}$ has most of its $\ell_{1}$ mass on inputs of Hamming
weight $O(1)$, which in view of~(\ref{eq:dual-block-compose}) implies
that the dual object for $F$ places most of its $\ell_{1}$ mass
on inputs of Hamming weight $\tilde{O}(n).$ The authors of~\cite{bun-thaler17adeg-ac0}
then use the Razborov\textendash Sherstov corrector object to transfer
the small amount of $\ell_{1}$ mass that the dual object for $F$
places on inputs of high Hamming weight, to inputs of low Hamming
weight. The resulting dual object for $F$ is supported entirely on
inputs of low Hamming weight and therefore witnesses a lower bound
on the $1/3$-approximate degree of the \emph{restriction} $F'$ of
$F$ to inputs of low Hamming weight. By re-encoding the input to
$F'$, one finally obtains a function $F''$ on $\tilde{O}(n)$ variables
with $1/3$-approximate degree polynomially larger than that of $f.$
This passage from $f$ to $F''$ is the desired hardness amplification
for approximate degree. We find it helpful to think of Bun and Thaler's
technique as block-composition followed by input compression, to reduce
the number of input variables in the block-composed function. To obtain
an $\Omega(n^{1-\epsilon})$ lower bound on the approximate degree
of $\classAC^{0}$, the authors of~\cite{bun-thaler17adeg-ac0} start
with a trivial circuit and iteratively apply the hardness amplification
step a constant number of times, until approximate degree $\Omega(n^{1-\epsilon})$
is reached. 

In follow-up work, Bun, Kothari, and Thaler~\cite{BKT17poly-strikes-back}
refined the technique of~\cite{bun-thaler17adeg-ac0} by deriving
optimal concentration bounds for the dual object for $\OR_{m}.$ They
thereby obtained tight or nearly tight lower bounds on the $1/3$-approximate
degree of \emph{surjectivity}, \emph{element distinctness}, and other
important problems. The most recent contribution to this line of work
is due to Bun and Thaler~\cite{BT18ac0-large-error}, who prove an
$\Omega(n^{1-\epsilon})$ lower bound on the $(1-2^{-n^{1-\epsilon}})$-approximate
degree of $\classAC^{0}$ by combining the method of~\cite{bun-thaler17adeg-ac0}
with Sherstov's work~\cite{sherstov11quantum-sdpt} on direct product
theorems for approximate degree. This near-linear lower bound substantially
strengthens the authors' previous result~\cite{bun-thaler17adeg-ac0}
on the \emph{bounded-error} approximate degree of $\classAC^{0}$,
but does not address the threshold degree.

\subsection{Our approach}

\subsubsection*{Threshold degree of AC\protect\textsuperscript{$\,0$}}

Bun and Thaler~\cite{BT18ac0-large-error} refer to obtaining an
$\Omega(n^{1-\epsilon})$ threshold degree lower bound for $\classAC^{0}$
as the ``main glaring open question left by our work.'' It is important
to note here that lower bounds on approximate degree, even with the
error parameter exponentially close to $1$ as in~\cite{BT18ac0-large-error},
have no implications for threshold degree. For example, there are
functions~\cite{sherstov09opthshs} with $(1-2^{-\Theta(n)})$-approximate
degree $\Theta(n)$ but threshold degree~$1$. Our proof of Theorem~\ref{thm:MAIN-degthr-ac0}
is unrelated to the most recent work of Bun and Thaler~\cite{BT18ac0-large-error}
on the large-error approximate degree of $\classAC^{0}$ and instead
builds on their earlier and simpler ``block-composition followed
by input compression'' approach~\cite{bun-thaler17adeg-ac0}. The
centerpiece of our proof is a hardness amplification result for threshold
degree, whereby any function $f$ with threshold degree $n^{\alpha}$
for a constant $0\leq\alpha<1$ can be transformed efficiently and
within $\classAC^{0}$ into a function $F$ with polynomially larger
threshold degree. 

In more detail, let $f\colon\pomon\to\pomo$ be a function of interest,
with threshold degree $n^{\alpha}$. We consider the block-composition
$f\circ\MP_{m},$ where $m=n^{O(1)}$ is an appropriate parameter
and $\MP_{m}=\AND_{m}\circ\OR_{m^{2}}$ is the Minsky\textendash Papert
function with threshold degree $\Omega(m)$. We construct the dual
object for $\MP_{m}$ from scratch to ensure concentration on inputs
of Hamming weight $\tilde{O}(m)$. By applying dual block-composition
to the threshold degree witnesses of $f$ and $\MP_{m}$, we obtain
a dual object $\Phi$ witnessing the $\Omega(mn^{\alpha})$ threshold
degree of $f\circ\MP_{m}$. So far in the proof, our differences from~\cite{bun-thaler17adeg-ac0}
are as follows: (i)~since our goal is amplification of threshold
degree, we work with witnesses of threshold degree rather than approximate
degree; (ii)~to ensure rapid growth of threshold degree, we use block-composition
with inner function $\MP_{m}=\AND_{m}\circ\OR_{m^{2}}$ of threshold
degree $\Theta(m)$, in place of Bun and Thaler's inner function $\AND_{\Theta(\log m)}\circ\OR_{m}$
of threshold degree $\Theta(\log m).$

Since the dual object for $\MP_{m}$ by construction has most of its
$\ell_{1}$ norm on inputs of Hamming weight $\tilde{O}(m)$, the
dual object $\Phi$ for the composed function has most of its $\ell_{1}$
norm on inputs of Hamming weight $\tilde{O}(nm)$. Analogous to~\cite{bun-thaler17adeg-ac0,BKT17poly-strikes-back,BT18ac0-large-error},
we would like to use the Razborov\textendash Sherstov corrector object
to \emph{remove} the $\ell_{1}$ mass that $\Phi$ has on the inputs
of high Hamming weight, transferring it to inputs of low Hamming weight.
This brings us to the novel and technically demanding part of our
proof. Previous works~\cite{bun-thaler17adeg-ac0,BKT17poly-strikes-back,BT18ac0-large-error}
transferred the $\ell_{1}$ mass from the inputs of high Hamming weight
to the neighborhood of the all-zeroes input $(0,0,\ldots,0).$ An
unavoidable feature of the Razborov\textendash Sherstov transfer process
is that it amplifies the $\ell_{1}$ mass being transferred. When
the transferred mass finally reaches its destination, it overwhelms
$\Phi$'s original values at the local points, destroying $\Phi$'s
sign-agreement with the composed function $f\circ\MP_{m}$. It is
this difficulty that most prevented earlier works~\cite{bun-thaler17adeg-ac0,BKT17poly-strikes-back,BT18ac0-large-error}
from obtaining a strong threshold degree lower bound for $\classAC^{0}$. 

We proceed differently. Instead of transferring the $\ell_{1}$ mass
of $\Phi$ from the inputs of high Hamming weight to the neighborhood
of $(0,0,\dots,0)$, we transfer it simultaneously to \emph{exponentially
many} strategically chosen neighborhoods. Split this way across many
neighborhoods, the transferred mass does not overpower the original
values of $\Phi$ and in particular does not change any signs. Working
out the details of this transfer scheme requires subtle and lengthy
calculations; it was not clear to us until the end that such a scheme
exists. Once the transfer process is complete, we obtain a witness
for the $\Omega(mn^{\alpha})$ threshold degree of $f\circ\MP_{m}$
restricted to inputs of low Hamming weight. Compressing the input
as in~\cite{bun-thaler17adeg-ac0,BKT17poly-strikes-back}, we obtain
an amplification theorem for threshold degree. With this work behind
us, the proof of Theorem~\ref{thm:MAIN-degthr-ac0} for any depth
$k$ amounts to starting with a trivial circuit and amplifying its
threshold degree $O(k)$ times.

\subsubsection*{Sign-rank of AC\protect\textsuperscript{$\,0$}}

It is not known how to ``lift'' a threshold degree lower bound in
a black-box manner to a sign-rank lower bound. In particular, Theorem~\ref{thm:MAIN-degthr-ac0}
has no implications a priori for the sign-rank of $\classAC^{0}$.
Our proofs of Theorems~\ref{thm:MAIN-sign-rank-3k} and~\ref{thm:MAIN-sign-rank-3k-plus-1}
are completely disjoint from Theorem~\ref{thm:MAIN-degthr-ac0} and
are instead based on a stronger approximation-theoretic quantity that
we call \emph{$\gamma$-smooth threshold degree.} Formally, the $\gamma$-smooth
threshold degree of a Boolean function $f\colon X\to\pomo$ is the
largest $d$ for which there is a nonzero function $\phi\colon X\to\Re$
with the following two properties: $\phi(x)f(x)\geq\gamma\cdot\|\phi\|_{1}/|X|$
for all $x\in X;$ and $\langle\phi,p\rangle=0$ for every polynomial
$p$ of degree less than $d$. Taking $\gamma=0$ in this formalism,
one recovers the standard dual characterization of the threshold degree
of $f.$ In particular, threshold degree is synonymous with $0$-smooth
threshold degree. The general case of $\gamma$-smooth threshold degree
for $\gamma>0$ requires threshold degree witnesses $\phi$ that are
\emph{min-smooth}, in that the absolute value of $\phi$ at any given
point is at least a $\gamma$ fraction of the average absolute value
of $\phi$ over all points. A substantial advantage of \emph{smooth}
threshold degree is that it has immediate sign-rank implications.
Specifically, any lower bound of $d$ on the $2^{-O(d)}$-smooth threshold
degree can be converted efficiently and in a black-box manner into
a sign-rank lower bound of $2^{\Omega(d)}$, using a combination of
the pattern matrix method~\cite{sherstov07ac-majmaj,sherstov07quantum}
and Forster's spectral lower bound on sign-rank~\cite{forster02linear,forster01relations}.
Accordingly, we obtain Theorems~\ref{thm:MAIN-sign-rank-3k} and~\ref{thm:MAIN-sign-rank-3k-plus-1}
by proving an $\Omega(n^{1-\epsilon})$ lower bound on the $2^{-O(n^{1-\epsilon})}$-smooth
threshold degree of $\classAC^{0}$, for any constant $\epsilon>0$. 

At the core of our result is an amplification theorem for smooth threshold
degree, whose repeated application makes it possible to prove arbitrarily
strong lower bounds for $\classAC^{0}$. Amplifying smooth threshold
degree is a complex juggling act due to the presence of two parameters\textemdash degree
and smoothness\textemdash that must evolve in coordinated fashion.
The approach of Theorem~\ref{thm:MAIN-degthr-ac0} is not useful
here because the threshold degree witnesses that arise from the proof
of Theorem~\ref{thm:MAIN-degthr-ac0} are highly nonsmooth. In more
detail, when amplifying the threshold degree of a function $f$ as
in the proof of Theorem~\ref{thm:MAIN-degthr-ac0}, two phenomena
adversely affect the smoothness parameter. The first is block-composition
itself as a composition technique, which in the regime of interest
to us transforms \emph{every} threshold degree witness for $f$ into
a hopelessly nonsmooth witness for the composed function. The other
culprit is the input compression step, which re-encodes the input
and thereby affects the smoothness in ways that are hard to control.
To overcome these difficulties, we develop a novel approach based
on what we call \emph{local smoothness.} 

Formally, let $\Phi\colon\NN^{n}\to\Re$ be a function of interest.
For a subset $X\subseteq\NN^{n}$ and a real number $K\geq1,$ we
say that $\Phi$ is \emph{$K$-smooth on $X$} if $|\Phi(x)|\leq K^{|x-x'|}|\Phi(x')|$
for all $x,x'\in X.$ Put another way, for any two points of $X$
at $\ell_{1}$~distance $d,$ the corresponding values of $\Phi$
differ in magnitude by a factor of at most $K^{d}.$ In and of itself,
a locally smooth function $\Phi$ need not be min-smooth because for
a pair of points that are far from each other, the corresponding $\Phi$-values
can differ by many orders of magnitude. However, locally smooth functions
exhibit extraordinary plasticity. Specifically, we show how to modify
a locally smooth function's metric properties\textemdash such as its
support or the distribution of its $\ell_{1}$ mass\textemdash without
the change being detectable by low-degree polynomials. This apparatus
makes it possible to restore min-smoothness to the dual object $\Phi$
that results from the block-composition step and preserve that min-smoothness
throughout the input compression step, eliminating the two obstacles
to min-smoothness in the earlier proof of Theorem~\ref{thm:MAIN-degthr-ac0}.
The new block-composition step uses a \emph{locally smooth} witness
for the threshold degree of $\MP_{m}$, which needs to be built from
scratch and is quite different from the witness in the proof of Theorem~\ref{thm:MAIN-degthr-ac0}.

Our described approach departs considerably from previous work on
the sign-rank of constant-depth circuits~\cite{RS07dc-dnf,BT16sign-rank-ac0,BT18ac0-large-error}.
The analytic notion in those earlier papers is weaker than $\gamma$-smooth
threshold degree and in particular allows the dual object to be \emph{arbitrary}
on a $\gamma$ fraction of the inputs. This weaker property is acceptable
when the main result is proved in one shot, with a closed-form construction
of the dual object. By contrast, we must construct dual objects iteratively,
with each iteration increasing the degree parameter and proportionately
decreasing the smoothness parameter. This iterative process requires
that the dual object in each iteration be min-smooth on the entire
domain. Perhaps unexpectedly, we find $\gamma$-smooth threshold degree
easier to work with than the weaker notion in previous work~\cite{RS07dc-dnf,BT16sign-rank-ac0,BT18ac0-large-error}.
In particular, we are able to give a new and short proof of the $\exp(\Omega(n^{1/3}))$
lower bound on the sign-rank of $\classAC^{0}$, originally obtained
by Razborov and Sherstov~\cite{RS07dc-dnf} with a much more complicated
approach. The new proof can be found in Section~\ref{subsec:RS-simplified},
where it serves as a prelude to our main result on the sign-rank of
$\classAC^{0}$.

\section{\label{sec:Preliminaries}Preliminaries}

\subsection{General}

For a string $x\in\zoon$ and a set $S\subseteq\{1,2,\ldots,n\},$
we let $x|_{S}$ denote the restriction of $x$ to the indices in
$S.$ In other words, $x|_{S}=x_{i_{1}}x_{i_{2}}\ldots x_{i_{|S|}},$
where $i_{1}<i_{2}<\cdots<i_{|S|}$ are the elements of $S.$ The
\emph{characteristic function} of a set $S\subseteq\{1,2,\ldots,n\}$
is given by
\[
\1_{S}(x)=\begin{cases}
1 & \text{if }x\in S,\\
0 & \text{otherwise.}
\end{cases}
\]
For a logical condition $C,$ we use the Iverson bracket
\[
\I[C]=\begin{cases}
1 & \text{if \ensuremath{C} holds,}\\
0 & \text{otherwise.}
\end{cases}
\]
We let $\NN=\{0,1,2,3,\ldots\}$ denote the set of natural numbers.
The following well-known bound~\cite[Proposition~1.4]{jukna11extremal-2nd-edition}
is used in our proofs without further mention: 
\begin{align}
\sum_{i=0}^{k}{n \choose i}\leq\left(\frac{\e n}{k}\right)^{k}, &  & k=0,1,2,\dots,n,\label{eq:entropy-bound-binomial}
\end{align}
where $\e=2.7182\ldots$ denotes Euler's number.

We adopt the extended real number system $\Re\cup\{-\infty,\infty\}$
in all calculations, with the additional convention that $0/0=0.$
We use the comparison operators in a unary capacity to denote one-sided
intervals of the real line. Thus, ${<}a,$ ${\leq}a,$ ${>}a,$ ${\geq}a$
stand for $(-\infty,a),$ $(-\infty,a],$ $(a,\infty),$ $[a,\infty),$
respectively. We let $\ln x$ and $\log x$ stand for the natural
logarithm of $x$ and the logarithm of $x$ to base $2,$ respectively.
We use the following two versions of the sign function:
\begin{align*}
\sign x=\begin{cases}
-1 & \text{if }x<0,\\
0 & \text{if }x=0,\\
1 & \text{if }x>0,
\end{cases}\qquad\qquad\qquad & \Sgn x=\begin{cases}
-1 & \text{if }x<0,\\
1 & \text{if }x\geq0.
\end{cases}
\end{align*}
The term \emph{Euclidean space} refers to $\Re^{n}$ for some positive
integer $n.$ We let $e_{i}$ denote the vector whose $i$th component
is $1$ and the others are $0.$ Thus, the vectors $e_{1},e_{2},\dots,e_{n}$
form the standard basis for $\Re^{n}.$ For vectors $x$ and $y,$
we write $x\leq y$ to mean that $x_{i}\leq y_{i}$ for each $i.$
The relations $\geq,$ $<,$ $>$ on vectors are defined analogously.

We frequently omit the argument in equations and inequalities involving
functions, as in $\sign p=(-1)^{f}$. Such statements are to be interpreted
pointwise. For example, the statement ``$f\geq2|g|$ on $X$'' means
that $f(x)\geq2|g(x)|$ for every $x\in X.$ The positive and negative
parts of a function $f\colon X\to\Re$ are denoted $\pospart f=\max\{f,0\}$
and $\negpart f=\max\{-f,0\}$, respectively.

\subsection{Boolean functions and circuits}

We view Boolean functions as mappings $X\to\zoo$ for some finite
set $X.$ More generally, we consider \emph{partial} Boolean functions
$f\colon X\to\{0,1,*\},$ with the output value $*$ used for don't-care
inputs. The negation of a Boolean function $f$ is denoted as usual
by $\overline{f}=1-f.$ The familiar functions $\OR_{n}\colon\zoon\to\zoo$
and $\AND_{n}\colon\zoon\to\zoo$ are given by $\OR_{n}(x)=\bigvee_{i=1}^{n}x_{i}$
and $\AND_{n}(x)=\bigwedge_{i=1}^{n}x_{i}.$ We abbreviate $\NOR_{n}=\neg\OR_{n}.$
The generalized \emph{Minsky\textendash Papert function} $\MP_{m,r}\colon(\zoo^{r})^{m}\to\zoo$
is given by $\MP_{m,r}(x)=\bigwedge_{i=1}^{m}\bigvee_{j=1}^{r}x_{i,j}.$
We abbreviate $\MP_{m}=\MP_{m,m^{2}},$ which is the right setting
of parameters for most of our applications.

We adopt the standard notation for function composition, with $f\circ g$
defined by $(f\circ g)(x)=f(g(x)).$ In addition, we use the $\circ$
operator to denote the \emph{componentwise} composition of Boolean
functions. Formally, the componentwise composition of $f\colon\zoon\to\{0,1\}$
and $g\colon X\to\{0,1\}$ is the function $f\circ g\colon X^{n}\to\{0,1\}$
given by $(f\circ g)(x_{1},x_{2},\ldots,x_{n})=f(g(x_{1}),g(x_{2}),\ldots,g(x_{n})).$
To illustrate, $\MP_{m,r}=\AND_{m}\circ\OR_{r}.$ Componentwise composition
is consistent with standard composition, which in the context of Boolean
functions is only defined for $n=1.$ Thus, the meaning of $f\circ g$
is determined by the range of $g$ and is never in doubt. Componentwise
composition generalizes in the natural manner to partial Boolean functions
$f\colon\zoon\to\{0,1,*\}$ and $g\colon X\to\{0,1,*\}$, as follows:
\[
(f\circ g)(x_{1},\ldots,x_{n})=\begin{cases}
f(g(x_{1}),\ldots,g(x_{n})) & \text{if }x_{1},\ldots,x_{n}\in g^{-1}(0\cup1),\\
* & \text{otherwise.}
\end{cases}
\]
Compositions $f_{1}\circ f_{2}\circ\cdots\circ f_{k}$ of three or
more functions, where each instance of the $\circ$ operator can be
standard or componentwise, are well-defined by associativity and do
not require parenthesization.

For Boolean strings $x,y\in\zoon,$ we let $x\oplus y$ denote their
bitwise XOR. The strings $x\wedge y$ and $x\vee y$ are defined analogously,
with the binary connective applied bitwise. A \emph{Boolean circuit
}$C$ in variables $x_{1},x_{2},\ldots,x_{n}$ is a circuit with inputs
$x_{1},\neg x_{1},x_{2},\neg x_{2},\ldots,x_{n},\neg x_{n}$ and gates
$\wedge$ and $\vee$. The circuit $C$ is \emph{monotone} if it does
not use any of the negated inputs $\neg x_{1},\neg x_{2},\ldots,\neg x_{n}.$
The \emph{fan-in} of $C$ is the maximum in-degree of any $\wedge$
or $\vee$ gate. Unless stated otherwise, we place no restrictions
on the gate fan-in. The \emph{size} of $C$ is the number of $\wedge$
and $\vee$ gates. The \emph{depth} of $C$ is the maximum number
of $\wedge$ and $\vee$ gates on any path from an input to the circuit
output. With this convention, the circuit that computes $(x_{1},x_{2},\ldots,x_{n})\mapsto x_{1}$
has depth~$0$. The circuit class $\classAC^{0}$ consists of function
families $\{f_{n}\}_{n=1}^{\infty}$ such that $f_{n}\colon\zoon\to\zoo$
is computed by a Boolean circuit of size at most $cn^{c}$ and depth
at most~$c$, for some constant $c\geq1$ and all $n.$ We specify
small-depth layered circuits by indicating the type of gate used in
each layer. For example, an \emph{AND-OR-AND circuit} is a depth-$3$
circuit with the top and bottom layers composed of $\wedge$ gates,
and middle layer composed of $\vee$ gates. A \emph{Boolean formula
}is a Boolean circuit in which every gate has fan-out~$1$. Common
examples of Boolean formulas are DNF and CNF formulas.

\subsection{Norms and products}

For a set $X,$ we let $\Re^{X}$ denote the linear space of real-valued
functions on $X.$ The \emph{support} of a function $f\in\Re^{X}$
is denoted $\supp f=\{x\in X:f(x)\ne0\}.$ For real-valued functions
with finite support, we adopt the usual norms and inner product:
\begin{align*}
 & \|f\|_{\infty}=\max_{x\in\supp f}\,|f(x)|,\\
 & \|f\|_{1}=\sum_{x\in\supp f}\,|f(x)|,\\
 & \langle f,g\rangle=\sum_{x\in\supp f\,\cap\,\supp g}f(x)g(x).
\end{align*}
This covers as a special case functions on finite sets. The \emph{tensor
product} of $f\in\Re^{X}$ and $g\in\Re^{Y}$ is denoted $f\otimes g\in\Re^{X\times Y}$
and given by $(f\otimes g)(x,y)=f(x)g(y).$ The tensor product $f\otimes f\otimes\cdots\otimes f$
($n$ times) is abbreviated $f^{\otimes n}.$ For a subset $S\subseteq\{1,2,\ldots,n\}$
and a function $f\colon X\to\Re,$ we define $f^{\otimes S}\colon X^{n}\to\Re$
by $f^{\otimes S}(x_{1},x_{2},\ldots,x_{n})=\prod_{i\in S}f(x_{i}).$
As extremal cases, we have $f^{\otimes\varnothing}\equiv1$ and $f^{\otimes\{1,2,\ldots,n\}}=f^{\otimes n}.$
Tensor product notation generalizes naturally to \emph{sets} of functions:
$F\otimes G=\{f\otimes g:f\in F,g\in G\}$ and $F^{\otimes n}=\{f_{1}\otimes f_{2}\otimes\cdots\otimes f_{n}:f_{1},f_{2},\ldots,f_{n}\in F\}.$
A \emph{conical combination} of $f_{1},f_{2},\dots,f_{k}\in\Re^{X}$
is any function of the form $\lambda_{1}f_{1}+\lambda_{2}f_{2}+\cdots+\lambda_{k}f_{k},$
where $\lambda_{1},\lambda_{2},\ldots,\lambda_{k}$ are nonnegative
reals. A \emph{convex combination} of $f_{1},f_{2},\dots,f_{k}\in\Re^{X}$
is any function $\lambda_{1}f_{1}+\lambda_{2}f_{2}+\cdots+\lambda_{k}f_{k},$
where $\lambda_{1},\lambda_{2},\ldots,\lambda_{k}$ are nonnegative
reals that sum to~$1.$ The \emph{conical hull} of $F\subseteq\Re^{X}$,
denoted $\cone F,$ is the set of all conical combinations of functions
in $F.$ The \emph{convex hull}, denoted $\conv F$, is defined analogously
as the set of all convex combinations of functions in $F.$ For any
set of functions $F\subseteq\Re^{X},$ we have
\begin{align}
(\conv F)^{\otimes n} & \subseteq\conv(F^{\otimes n}).\label{eq:tensor-conv-conv-tensor}
\end{align}

Throughout this manuscript, we view probability distributions as real
functions. This convention makes available the shorthands introduced
above. In particular, for probability distributions $\mu$ and $\lambda,$
the symbol $\supp\mu$ denotes the support of $\mu$, and $\mu\otimes\lambda$
denotes the probability distribution given by $(\mu\otimes\lambda)(x,y)=\mu(x)\lambda(y).$
If $\mu$ is a probability distribution on $X,$ we consider $\mu$
to be defined also on any superset of $X$ with the understanding
that $\mu=0$ outside $X.$ We let $\Distribution(X)$ denote the
family of all finitely supported probability distributions on $X.$
Most of this paper is concerned with the distribution family $\Distribution(\NN^{n})$
and its subfamilies, each of which we will denote with a Fraktur letter.

Analogous to functions, we adopt the familiar norms for vectors $x\in\Re^{n}$
in Euclidean space: $\|x\|_{\infty}=\max_{i=1,\ldots,n}|x_{i}|$ and
$\|x\|_{1}=\sum_{i=1}^{n}|x_{i}|.$ The latter norm is particularly
prominent in this paper, and to avoid notational clutter we use $|x|$
interchangeably with $\|x\|_{1}$. We refer to $|x|=\|x\|_{1}$ as
the \emph{weight} of $x.$ For any sets $X\subseteq\NN^{n}$ and $W\subseteq\Re,$
we define
\[
X|_{W}=\{x\in X:|x|\in W\}.
\]
In the case of a one-element set $W=\{w\}$, we further shorten $X|_{\{w\}}$
to $X|_{w}.$ To illustrate, $\NN^{n}|_{\leq w}$ denotes the set
of vectors whose components are natural numbers and sum to at most
$w$, whereas $\zoon|_{w}$ denotes the set of Boolean strings of
length $n$ and Hamming weight exactly $w.$ For a function $f\colon X\to\Re$
on a subset $X\subseteq\NN^{n},$ we let $f|_{W}$ denote the restriction
of $f$ to $X|_{W}.$ A typical instance of this notation would be
$f|_{\leq w}$ for some real number~$w.$

\subsection{Orthogonal content}

For a multivariate real polynomial $p\colon\Re^{n}\to\Re$, we let
$\deg p$ denote the total degree of $p$, i.e., the largest degree
of any monomial of $p.$ We use the terms \emph{degree }and \emph{total
degree }interchangeably in this paper. It will be convenient to define
the degree of the zero polynomial by $\deg0=-\infty.$ For a real-valued
function $\phi$ supported on a finite subset of $\Re^{n}$, we define
the \emph{orthogonal content of $\phi,$} denoted $\orth\phi$, to
be the minimum degree of a real polynomial $p$ for which $\langle\phi,p\rangle\ne0.$
We adopt the convention that $\orth\phi=\infty$ if no such polynomial
exists. It is clear that $\orth\phi\in\NN\cup\{\infty\},$ with the
extremal cases $\orth\phi=0\;\Leftrightarrow\;\langle\phi,1\rangle\ne0$
and $\orth\phi=\infty\;\Leftrightarrow\;\phi=0.$ Our next three results
record additional facts about orthogonal content.
\begin{prop}
\label{prop:orth}Let $X$ and $Y$ be nonempty finite subsets of
Euclidean space. Then:
\begin{enumerate}
\item \label{item:orth-sum}$\orth(\phi+\psi)\geq\min\{\orth\phi,\orth\psi\}$
for all $\phi,\psi\colon X\to\Re;$
\item \label{item:orth-tensor}$\orth(\phi\otimes\psi)=\orth(\phi)+\orth(\psi)$
for all $\phi\colon X\to\Re$ and $\psi\colon Y\to\Re;$
\item \label{item:orth-difference-of-tensors}$\orth(\phi^{\otimes n}-\psi^{\otimes n})\geq\orth(\phi-\psi)$
for all $\phi,\psi\colon X\to\Re$ and all $n\geq1.$
\end{enumerate}
\end{prop}

\begin{proof}
Item~\ref{item:orth-sum} is immediate, as is the upper bound in~\ref{item:orth-tensor}.
For the lower bound in~\ref{item:orth-tensor}, simply note that
the linearity of inner product makes it possible to restrict attention
to factored polynomials $p(x)q(y)$, where $p$ and $q$ are polynomials
on $X$ and $Y$, respectively. For~\ref{item:orth-difference-of-tensors},
use a telescoping sum to write
\begin{align*}
\phi^{\otimes n}-\psi^{\otimes n} & =\sum_{i=0}^{n-1}(\phi^{\otimes(n-i)}\otimes\psi^{\otimes i}-\phi^{\otimes(n-i-1)}\otimes\psi^{\otimes(i+1)})\\
 & =\sum_{i=0}^{n-1}\phi^{\otimes(n-i-1)}\otimes(\phi-\psi)\otimes\psi^{\otimes i}.
\end{align*}
By~\ref{item:orth-tensor}, each term in the final expression has
orthogonal content at least $\orth(\phi-\psi).$ By~\ref{item:orth-sum},
then, the sum has orthogonal content at least $\orth(\phi-\psi)$
as well.
\end{proof}
\begin{prop}
\label{prop:expect-out}Let $\phi_{0},\phi_{1}\colon X\to\Re$ be
given functions on a finite subset $X$ of Euclidean space with $\orth(\phi_{1}-\phi_{0})>0$.
Then for every polynomial $p\colon X^{n}\to\Re,$ the mapping $z\mapsto\langle\bigotimes_{i=1}^{n}\phi_{z_{i}},p\rangle$
is a polynomial on $\zoon$ of degree at most $(\deg p)/\orth(\phi_{1}-\phi_{0}).$ 
\end{prop}

\begin{proof}
By linearity, it suffices to consider factored polynomials $p(x_{1},\ldots,x_{n})=\prod_{i=1}^{n}p_{i}(x_{i}),$
where each $p_{i}$ is a nonzero polynomial on $X.$ In this setting,
we have
\begin{align}
\left\langle \bigotimes_{i=1}^{n}\phi_{z_{i}},p\right\rangle  & =\prod_{i=1}^{n}\left\langle \phi_{z_{i}},p_{i}\right\rangle .\label{eq:factored-inner-product}
\end{align}
By definition, $\langle\phi_{0},p_{i}\rangle=\langle\phi_{1},p_{i}\rangle$
for any index $i$ with $\deg p_{i}<\orth(\phi_{1}-\phi_{0}).$ As
a result, such indices do not contribute to the degree of the right-hand
side of~(\ref{eq:factored-inner-product}) as a function of $z.$
The contribution of any other index to the degree is clearly at most~$1.$
Summarizing, the right-hand side of~(\ref{eq:factored-inner-product})
is a polynomial in $z\in\zoon$ of degree at most $|\{i:\deg p_{i}\geq\orth(\phi_{1}-\phi_{0})\}|\leq(\deg p)/\orth(\phi_{1}-\phi_{0}).$ 
\end{proof}
\begin{cor}
\label{cor:orth-block-composition}Let $X$ be a finite subset of
Euclidean space. Then for any functions $\phi_{0},\phi_{1}\colon X\to\Re$
and $\psi\colon\zoon\to\Re,$ 
\[
\orth\left(\sum_{z\in\zoon}\psi(z)\bigotimes_{i=1}^{n}\phi_{z_{i}}\right)\geq\orth(\psi)\cdot\orth(\phi_{1}-\phi_{0}).
\]
\end{cor}

\begin{proof}
We may assume that $\orth(\psi)\cdot\orth(\phi_{1}-\phi_{0})>0$ since
the claim holds trivially otherwise. Fix any polynomial $P$ of degree
less than $\orth(\psi)\cdot\orth(\phi_{1}-\phi_{0}).$ The linearity
of inner product leads to
\[
\left\langle \sum_{z\in\zoon}\psi(z)\bigotimes_{i=1}^{n}\phi_{z_{i}},P\right\rangle =\sum_{z\in\zoon}\psi(z)\left\langle \bigotimes_{i=1}^{n}\phi_{z_{i}},P\right\rangle .
\]
By Proposition~\ref{prop:expect-out}, the right-hand side is the
inner product of $\psi$ with a polynomial of degree less than $\orth\psi$
and is therefore zero.
\end{proof}
\noindent Observe that Corollary~\ref{cor:orth-block-composition}
gives an alternate proof of Proposition~\ref{prop:orth}\ref{item:orth-difference-of-tensors}.
Our next proposition uses orthogonal content to give a useful criterion
for a real-valued function to be a probability distribution.
\begin{prop}
\label{prop:distribution-criterion}Let $\Lambda$ be a probability
distribution on a finite subset $X$ of Euclidean space. Let $\tilde{\Lambda}\colon X\to\Re$
be given with $\tilde{\Lambda}\geq0$ and $\orth(\Lambda-\tilde{\Lambda})>0.$
Then $\tilde{\Lambda}$ is a probability distribution on $X.$
\end{prop}

\begin{proof}
By hypothesis, $\tilde{\Lambda}$ is a nonnegative function. Moreover,
$\|\tilde{\Lambda}\|_{1}=\langle\tilde{\Lambda},1\rangle=\langle\Lambda,1\rangle-\langle\Lambda-\tilde{\Lambda},1\rangle=\langle\Lambda,1\rangle=1,$
where the third step uses $\orth(\Lambda-\tilde{\Lambda})>0.$
\end{proof}

\subsection{Sign-representation}

Let $f\colon X\to\zoo$ be a given Boolean function, for a finite
subset $X\subset\Re^{n}$. The \emph{threshold degree} of $f,$ denoted
$\degthr(f),$ is the least degree of a real polynomial $p$ that
represents $f$ in sign: $\sign p(x)=(-1)^{f(x)}$ for each $x\in X.$
The term ``threshold degree'' appears to be due to Saks~\cite{saks93slicing}.
Equivalent terms in the literature include ``strong degree''~\cite{aspnes91voting},
``voting polynomial degree''~\cite{krause94depth2mod}, ``polynomial
threshold function degree''~\cite{odonnell03degree}, and ``sign
degree''~\cite{buhrman07pp-upp}. One of the first results on polynomial
representations of Boolean functions was the following tight lower
bound on the threshold degree of $\MP_{m},$ due to Minsky and Papert~\cite{minsky88perceptrons}.
\begin{thm}[Minsky and Papert]
\label{thm:MP-thrdeg}$\degthr(\MP_{m})=\Omega(m)$.
\end{thm}

\noindent Three new proofs of this lower bound, unrelated to Minsky
and Papert's original proof, were discovered in~\cite{sherstov14sign-deg-ac0}.
Threshold degree admits the following dual characterization, obtained
by appeal to linear programming duality.
\begin{fact}
\label{fact:dual-thrdeg}Let $f\colon X\to\zoo$ be a given Boolean
function on a finite subset $X$ of Euclidean space. Then $\degthr(f)\geq d$
if and only if there exists $\psi\colon X\to\Re$ such that
\begin{align*}
 & (-1)^{f(x)}\psi(x)\geq0, &  & x\in X,\\
 & \orth\psi\geq d,\\
 & \psi\not\equiv0.
\end{align*}
\end{fact}

\noindent The function $\psi$ \emph{witnesses} the threshold degree
of $f$, and is called a \emph{dual polynomial} due to its origin
in a dual linear program. We refer the reader to~\cite{aspnes91voting,odonnell03degree,sherstov09hshs}
for a proof of Fact~\ref{fact:dual-thrdeg}. The following equivalent
statement is occasionally more convenient to work with.
\begin{fact}
\label{fact:dual-thrdeg-distribution-view}For every Boolean function
$f\colon X\to\zoo$ on a finite subset $X$ of Euclidean space,
\begin{equation}
\degthr(f)=\max_{\mu\in\Distribution(X)}\orth((-1)^{f}\cdot\mu).\label{eq:dual-degthr}
\end{equation}
\end{fact}

We now define a generalization of threshold degree inspired by the
dual view in Fact~\ref{fact:dual-thrdeg-distribution-view}. For
a function $f\colon X\to\zoo$ and a real number $0\leq\gamma\leq1,$
let
\begin{equation}
\degthr(f,\gamma)=\max_{\substack{\mu\in\Distribution(X):\\
\mu\geq\gamma/|X|\text{ on }X
}
}\orth((-1)^{f}\cdot\mu).\label{eq:dual-degthr-min-smooth}
\end{equation}
We call this quantity the \emph{$\gamma$-smooth threshold degree
of $f$}, in reference to the fact that the maximization in~(\ref{eq:dual-degthr-min-smooth})
is over probability distributions $\mu$ that assign to every point
of the domain at least a $\gamma$ fraction of the average point's
probability. A glance at~(\ref{eq:dual-degthr}) and~(\ref{eq:dual-degthr-min-smooth})
reveals that $\degthr(f,\gamma)$ is monotonically nonincreasing in
$\gamma,$ with the limiting case $\degthr(f,0)=\degthr(f).$
\begin{fact}
\label{fact:fact-degthr-min-smooth}For every nonconstant function
$f\colon X\to\zoo,$ 
\[
\degthr\!\left(f,\frac{1}{2}\right)\geq1.
\]
\end{fact}

\begin{proof}
Define $\mu=\frac{1}{2}\mu_{0}+\frac{1}{2}\mu_{1},$ where $\mu_{i}$
is the uniform probability distribution on $f^{-1}(i).$ Then clearly
$\orth((-1)^{f}\cdot\mu)\geq1$ and $\mu\geq\frac{1}{2}\max\{\mu_{0},\mu_{1}\}\geq\frac{1}{2|X|}$
on $X.$
\end{proof}
\noindent As one might expect, padding a Boolean function with irrelevant
variables does not decrease its smooth threshold degree. We record
this observation below.
\begin{prop}
\label{prop:smooth-degthr-padding}Fix integers $N\geq n\geq1$ and
a function $f\colon\zoon\to\zoo$. Define $F\colon\zoo^{N}\to\zoo$
by $F(x_{1},x_{2},\ldots,x_{N})=f(x_{1},x_{2},\ldots,x_{n}).$ Then
\begin{align*}
\degthr(F,\gamma) & \geq\degthr(f,\gamma), &  & 0\leq\gamma\leq1.
\end{align*}
In particular,
\[
\degthr(F)\geq\degthr(f).
\]
\end{prop}

\begin{proof}
Fix $0\leq\gamma\leq1$ arbitrarily. Let $\lambda$ be a probability
distribution on $\zoon$ such that $\lambda(x)\geq\gamma2^{-n}$ for
all $x\in\zoon$, and in addition $\orth((-1)^{f}\cdot\lambda)=\degthr(f,\gamma).$
Consider the probability distribution $\Lambda$ on $\zoo^{N}$ given
by $\Lambda(x_{1},x_{2},\ldots,x_{N})=2^{-N+n}\lambda(x_{1},x_{2},\ldots,x_{n}).$
Then $\Lambda(x)\geq2^{-N+n}\cdot\gamma2^{-n}=\gamma2^{-N}$ for all
$x\in\zoo^{N}$, and in addition $\orth((-1)^{F}\cdot\Lambda)=\orth((-1)^{f}\cdot\lambda)=\degthr(f,\gamma).$
\end{proof}

\subsection{Symmetrization}

Let $S_{n}$ denote the symmetric group on $n$ elements. For a permutation
$\sigma\in S_{n}$ and an arbitrary sequence $x=(x_{1},x_{2},\ldots,x_{n}),$
we adopt the shorthand $\sigma x=(x_{\sigma(1)},x_{\sigma(2)},\ldots,x_{\sigma(n)}).$
A function $f(x_{1},x_{2},\ldots,x_{n})$ is called \emph{symmetric}
if it is invariant under permutation of the input variables: $f(x_{1},x_{2},\ldots,x_{n})=f(x_{\sigma(1)},x_{\sigma(2)},\ldots,x_{\sigma(n)})$
for all $x$ and $\sigma.$ Symmetric functions on $\zoon$ are intimately
related to univariate polynomials, as was first observed by Minsky
and Papert in their \emph{symmetrization argument}~\cite{minsky88perceptrons}.
\begin{prop}[Minsky and Papert]
\label{prop:minsky-papert}Let $p\colon\Re^{n}\to\Re$ be a given
polynomial. Then the mapping
\[
t\mapsto\Exp_{\substack{x\in\zoon|_{t}}
}\;p(x)
\]
is a univariate polynomial on $\{0,1,2,\ldots,n\}$ of degree at most
$\deg p.$
\end{prop}

\noindent Minsky and Papert's result generalizes to block-symmetric
functions:
\begin{prop}
\label{prop:symmetrization} Let $n_{1},\dots,n_{k}$ be positive
integers. Let $p\colon\Re^{n_{1}}\times\cdots\times\Re^{n_{k}}\to\Re$
be a given polynomial. Then the mapping
\begin{align*}
(t_{1},t_{2},\ldots,t_{k})\mapsto\Exp_{\substack{x_{1}\in\zoo^{n_{1}}|_{t_{1}}}
}\Exp_{\substack{x_{2}\in\zoo^{n_{2}}|_{t_{2}}}
}\cdots\Exp_{\substack{x_{k}\in\zoo^{n_{k}}|_{t_{k}}}
}\;p(x_{1},x_{2},\ldots,x_{k})
\end{align*}
 is a polynomial on $\{0,1,\ldots,n_{1}\}\times\{0,1,\ldots,n_{2}\}\times\cdots\times\{0,1,\ldots,n_{k}\}$
of degree at most $\deg p.$
\end{prop}

\noindent Proposition~\ref{prop:symmetrization} follows in a straightforward
manner from Proposition~\ref{prop:minsky-papert} by induction on
the number of blocks $k$, as pointed out in~\cite[Proposition~2.3]{RS07dc-dnf}.
The next result is yet another generalization of Minsky and Papert's
symmetrization technique, this time to the setting when $x_{1},x_{2},\ldots,x_{n}$
are vectors rather than bits.
\begin{prop}
\label{prop:ambainis-symmetrization}Let $p\colon(\Re^{m})^{n}\to\Re$
be a polynomial of degree $d.$ Then there is a polynomial $p^{*}\colon\Re^{n}\to\Re$
of degree at most $d$ such that for all $x_{1},x_{2},\ldots,x_{n}\in\{e_{1},e_{2},\ldots,e_{m},0^{m}\},$
\begin{align*}
\Exp_{\sigma\in S_{n}}p(x_{\sigma(1)},x_{\sigma(2)},\ldots,x_{\sigma(n)})=p^{*}(x_{1}+x_{2}+\cdots+x_{n}).
\end{align*}
\end{prop}

\begin{proof}
We closely follow an argument due to Ambainis~\cite[Lemma~3.4]{ambainis05collision},
who proved a related result. Since the components of $x_{1},x_{2},\ldots,x_{n}$
are Boolean-valued, we have $x_{i,j}=x_{i,j}^{2}=x_{i,j}^{3}=\cdots$
and therefore we may assume that $p$ is multilinear. By linearity,
it further suffices to consider the case when $p$ is a single monomial:
\begin{equation}
p(x_{1},x_{2},\ldots,x_{n})=\prod_{j=1}^{m}\prod_{i\in A_{j}}x_{i,j}\label{eq:ambainis-p-defined}
\end{equation}
for some sets $A_{1},A_{2},\ldots,A_{m}\subseteq\{1,2,\ldots,n\}$
with $\sum_{j=1}^{m}|A_{j}|\leq d.$ If some pair of sets $A_{j},A_{j'}$
with $j\ne j'$ have nonempty intersection, then the right-hand side
of~(\ref{eq:ambainis-p-defined}) contains a product of the form
$x_{i,j}x_{i,j'}$ for some $i$ and thus $p\equiv0$ on the domain
in question. As a result, the proposition holds with $p^{*}=0.$ In
the complementary case when $A_{1},A_{2},\ldots,A_{m}$ are pairwise
disjoint, we calculate
\begin{align*}
\Exp_{\sigma\in S_{n}} & p(x_{\sigma(1)},x_{\sigma(2)},\ldots,x_{\sigma(n)})\\
 & =\prod_{j=1}^{m}\Exp_{\sigma\in S_{n}}\left[\prod_{i\in A_{j}}x_{\sigma(i),j}\;\middle|\;\prod_{i\in A_{j'}}x_{\sigma(i),j'}=1\text{ for all }j'<j\right]\\
 & =\prod_{j=1}^{m}\binom{x_{1,j}+x_{2,j}+\cdots+x_{n,j}}{|A_{j}|}\binom{n-|A_{1}|-|A_{2}|-\cdots-|A_{j-1}|}{|A_{j}|}^{-1}.
\end{align*}
Expanding out the binomial coefficients shows that the final expression
is an $m$-variate polynomial whose argument is the vector sum $x_{1}+x_{2}+\cdots+x_{n}\in\Re^{m}$.
Moreover, the degree of this polynomial is $\sum|A_{j}|\leq d.$
\end{proof}
\begin{cor}
\label{cor:ambainis-symmetrization}Let $p\colon(\Re^{m})^{n}\to\Re$
be a given polynomial. Then the mapping
\begin{equation}
v\mapsto\Exp_{\substack{x\in\{0^{m},e_{1},e_{2},\ldots,e_{m}\}^{n}:\\
x_{1}+x_{2}+\cdots+x_{n}=v
}
}\;\;p(x)\label{eq:ambainis-mapping}
\end{equation}
is a polynomial on $\NN^{m}|_{\leq n}$ of degree at most $\deg p.$
\end{cor}

\noindent Minsky and Papert's symmetrization corresponds to $m=1$
in Corollary~\ref{cor:ambainis-symmetrization}.
\begin{proof}[Proof of Corollary~\emph{\ref{cor:ambainis-symmetrization}}.]
 Let $v\in\NN^{m}|_{\leq n}$ be given. Then all representations
$v=x_{1}+x_{2}+\cdots+x_{n}$ with $x_{1},x_{2},\ldots,x_{n}\in\{0^{m},e_{1},e_{2},\ldots,e_{m}\}$
are the same up to the order of the summands. As a result,~(\ref{eq:ambainis-mapping})
is the same mapping as
\[
v\mapsto\Exp_{\sigma\in S_{n}}p(\sigma(\underbrace{e_{1},\ldots,e_{1}}_{v_{1}},\underbrace{e_{2},\ldots,e_{2}}_{v_{2}},\ldots,\underbrace{e_{m},\ldots,e_{m}}_{v_{m}},\underbrace{0^{m},0^{m}\ldots,0^{m}}_{n-v_{1}-\cdots-v_{m}})),
\]
which by Proposition~\ref{prop:ambainis-symmetrization} is a polynomial
in
\[
\underbrace{e_{1}+\cdots+e_{1}}_{v_{1}}+\underbrace{e_{2}+\cdots+e_{2}}_{v_{2}}+\cdots+\underbrace{e_{m}+\cdots+e_{m}}_{v_{m}}+\underbrace{0^{m}+\cdots+0^{m}}_{n-v_{1}-\cdots-v_{m}}=v
\]
of degree at most $\deg p.$ 
\end{proof}
It will be helpful to define symmetric versions of basic Boolean functions.
We define $\AND_{n}^{*},\OR_{n}^{*}\colon\{0,1,2,\ldots,n\}\to\zoo$
by
\[
\AND_{n}^{*}(t)=\begin{cases}
1 & \text{if }t=n,\\
0 & \text{otherwise,}
\end{cases}\qquad\qquad\qquad\OR_{n}^{*}(t)=\begin{cases}
0 & \text{if }t=0,\\
1 & \text{otherwise.}
\end{cases}
\]
The symmetric variant of the Minsky\textendash Papert function is
$\MP_{m,r}^{*}=\AND_{m}\circ\OR_{r}^{*}.$

\subsection{Communication complexity}

An excellent reference on communication complexity is the monograph
by Kushilevitz and Nisan~\cite{ccbook}. In this overview, we will
limit ourselves to key definitions and notation. We adopt the standard
randomized model of multiparty communication, due to Chandra et al.~\cite{cfl83multiparty}.
The model features $\ell$ communicating players, tasked with computing
a Boolean function $F\colon X_{1}\times X_{2}\times\cdots\times X_{\ell}\to\zoo$
for some finite sets $X_{1},X_{2},\dots,X_{\ell}.$ A given input
$(x_{1},x_{2},\dots,x_{\ell})\in X_{1}\times X_{2}\times\cdots\times X_{\ell}$
is distributed among the players by placing $x_{i}$, figuratively
speaking, on the forehead of the $i$th player (for $i=1,2,\dots,\ell$).
In other words, the $i$th player knows the arguments $x_{1},\dots,x_{i-1},x_{i+1},\dots,x_{\ell}$
but not $x_{i}.$ The players communicate by sending broadcast messages,
taking turns according to a protocol agreed upon in advance. Each
of them privately holds an unlimited supply of uniformly random bits,
which he can use along with his available arguments when deciding
what message to send at any given point in the protocol. The players'
objective is to compute $F(x_{1},x_{2},\ldots,x_{\ell})$. An \emph{$\epsilon$-error
protocol} for $F$ is one which, on every input $(x_{1},x_{2},\dots,x_{\ell}),$
produces the correct answer $F(x_{1},x_{2},\dots,x_{\ell})$ with
probability at least $1-\epsilon.$ The \emph{cost} of a protocol
is the total bit length of the messages broadcast by all the players
in the worst case.\footnote{~The contribution of a $b$-bit broadcast to the protocol cost is
$b$ rather than $\ell\cdot b$.} The \emph{$\epsilon$-error randomized communication complexity}
of $F,$ denoted $R_{\epsilon}(F),$ is the least cost of an $\epsilon$-error
randomized protocol for $F$. As a special case of this model for
$\ell=2,$ one recovers the original two-party model of Yao~\cite{yao79cc}
reviewed in the introduction.

Our work focuses on randomized protocols with error probability close
to that of random guessing, $1/2.$ There are two natural ways to
define the communication complexity of a multiparty problem $F$ in
this setting. The \emph{communication complexity of $F$ with unbounded
error}, introduced by Paturi and Simon~\cite{paturi86cc}, is the
quantity 
\begin{equation}
\upp(F)=\inf_{0\leq\epsilon<1/2}R_{\epsilon}(F).\label{eq:upp-def}
\end{equation}
Here, the error is unbounded in the sense that it can be arbitrarily
close to $1/2.$ Babai et al.~\cite{BFS86cc} proposed an alternate
quantity, which includes an additive penalty term that depends on
the error probability: 
\begin{equation}
\pp(F)=\inf_{0\leq\epsilon<1/2}\left\{ R_{\epsilon}(F)+\log\frac{1}{\frac{1}{2}-\epsilon}\right\} .\label{eq:pp-def}
\end{equation}
This quantity is known as the \emph{communication complexity of $F$
with weakly unbounded error.}

\subsection{\label{subsec:Discrepancy-and-sign-rank}Discrepancy and sign-rank}

An \emph{$\ell$-dimensional cylinder intersection} is a function
$\chi\colon X_{1}\times X_{2}\times\cdots\times X_{\ell}\to\zoo$
of the form 
\begin{align*}
\chi(x_{1},x_{2},\dots,x_{\ell})=\prod_{i=1}^{\ell}\chi_{i}(x_{1},\dots,x_{i-1},x_{i+1},\dots,x_{\ell}),
\end{align*}
where $\chi_{i}\colon X_{1}\times\cdots\times X_{i-1}\times X_{i+1}\times\cdots\times X_{\ell}\to\zoo.$
In other words, an $\ell$-dimensional cylinder intersection is the
product of $\ell$ functions with range $\zoo,$ where the $i$th
function does not depend on the $i$th coordinate but may depend arbitrarily
on the other $\ell-1$ coordinates. Introduced by Babai et al.~\cite{bns92},
cylinder intersections are the fundamental building blocks of communication
protocols and for that reason play a central role in the theory. For
a Boolean function $F\colon X_{1}\times X_{2}\times\cdots\times X_{\ell}\to\zoo$
and a probability distribution $P$ on $X_{1}\times X_{2}\times\cdots\times X_{\ell},$
the \emph{discrepancy} \emph{of $F$ with respect to $P$} is given
by 
\begin{align*}
\disc_{P}(F)=\max_{\chi}\left|\sum_{x\in X_{1}\times X_{2}\times\cdots\times X_{\ell}}(-1)^{F(x)}P(x)\chi(x)\right|,
\end{align*}
where the maximum is over cylinder intersections $\chi$. The minimum
discrepancy over all distributions is denoted 
\[
\disc(F)=\min_{P}\;\disc_{P}(F).
\]
The \emph{discrepancy method}~\cite{chor-goldreich88ip,bns92,ccbook}
is a classic technique that bounds randomized communication complexity
from below in terms of discrepancy.
\begin{thm}[Discrepancy method]
\label{thm:dm} Let $F\colon X_{1}\times X_{2}\times\cdots\times X_{\ell}\to\zoo$
be an $\ell$-party communication problem. Then 
\begin{align*}
2^{R_{\epsilon}(F)}\geq\frac{1-2\epsilon}{\disc(F)}.
\end{align*}
\end{thm}

\noindent Combining this theorem with the definition of $\pp(F)$
gives the following corollary.
\begin{cor}
\label{cor:dm}Let $F\colon X_{1}\times X_{2}\times\cdots\times X_{\ell}\to\zoo$
be an $\ell$-party communication problem. Then
\[
\pp(F)\geq\log\frac{2}{\disc(F)}.
\]
 
\end{cor}

The \emph{sign-rank} of a real matrix $A\in\Re^{n\times m}$ with
nonzero entries is the least rank of a matrix $B\in\Re^{n\times m}$
such that $\sign A_{i,j}=\sign B_{i,j}$ for all $i,j$. In general,
the sign-rank of a matrix can be vastly smaller than its rank. For
example, consider the following nonsingular matrices of order $n\geq3$:\[
\begin{bmatrix}
\begin{aligned}
\begin{matrix}
1\\
&1\\
&&1\\
\end{matrix}
\end{aligned} & \scalebox{1.7}{$1$}\\
\begin{matrix}
\\
\scalebox{1.7}{$-1$~}
\end{matrix} & \begin{aligned}
\begin{matrix}
\ddots\\
&1\\
&&1\\
\end{matrix}
\end{aligned}
\end{bmatrix}, \qquad
\begin{bmatrix}
\begin{aligned}
\begin{matrix}
1\\
&1\\
&&1\\
\end{matrix}
\end{aligned} & \scalebox{1.7}{$-1$}\\
\begin{matrix}
\\\scalebox{1.7}{$-1$}
\end{matrix}
& \begin{aligned}
\begin{matrix}
\ddots\\
&1\\
&&1\\
\end{matrix}
\end{aligned}
\end{bmatrix}. \qquad
\]These matrices have sign-rank at most $2$ and $3,$ respectively.
Indeed, the first matrix has the same sign pattern as $[2(j-i)+1]_{i,j}.$
The second has the same sign pattern as $[\langle v_{i},v_{j}\rangle-(1-\epsilon)]_{i,j},$
where $v_{1},v_{2},\dots,v_{n}\in\Re^{2}$ are arbitrary pairwise
distinct unit vectors and $\epsilon$ is a suitably small positive
real, cf.~\cite[Section~5]{paturi86cc}. As a matter of notational
convenience, we extend the notion of sign-rank to Boolean functions
$f\colon X\times Y\to\zoo$ by defining $\srank(f)=\srank(M_{f}),$
where $M_{f}=[(-1)^{f(x,y)}]_{x\in X,y\in Y}$ is the matrix associated
with $f$. A remarkable fact, due to Paturi and Simon~\cite{paturi86cc},
is that the sign-rank of a two-party communication problem fully characterizes
its unbounded-error communication complexity.
\begin{thm}[Paturi and Simon]
\label{thm:srank-upp} Let $F\colon X\times Y\to\zoo$ be a given
communication problem. Then
\[
\log\srank(F)\leq\upp(F)\leq\log\srank(F)+2.
\]
\end{thm}

As Corollary~\ref{cor:dm} and Theorem~\ref{thm:srank-upp} show,
the study of communication with unbounded and weakly unbounded error
is in essence the study of discrepancy and sign-rank. These quantities
are difficult to analyze from first principles. The \emph{pattern
matrix method}, developed in~\cite{sherstov07ac-majmaj,sherstov07quantum},
is a technique that transforms lower bounds for polynomial approximation
into bounds on discrepancy, sign-rank, and various other quantities
in communication complexity. For our discrepancy bounds, we use the
following special case of the pattern matrix method~\cite[Theorem~5.7 and equation~(119)]{sherstov13directional}.
\begin{thm}[Sherstov]
\label{thm:pm-discrepancy-k-party} Let $f\colon\zoon\to\zoo$ be
given. Consider the $\ell$-party communication problem $F\colon(\zoo^{nm})^{\ell}\to\zoo$
given by $F=f\circ\NOR_{m}\circ\AND_{\ell}.$ Then
\begin{align*}
 & \disc(F)\leq\left(\frac{c2^{\ell}\ell}{\sqrt{m}}\right)^{\degthr(f)/2},
\end{align*}
where $c>0$ is a constant independent of $n,m,\ell,f.$
\end{thm}

\noindent We note that the case $\ell=2$ of Theorem~\ref{thm:pm-discrepancy-k-party}
is vastly easier to prove than the general statement; this two-party
result can be found in~\cite[Theorem~7.3 and equation~(7.3)]{sherstov07quantum}.
For our sign-rank lower bounds, we use the following theorem implicit
in~\cite{sherstov07symm-sign-rank}.
\begin{thm}[Sherstov, implicit]
\label{thm:thrdeg-to-sign-rank}Let $f\colon\zoon\to\zoo$ be given.
Suppose that $\degthr(f,\gamma)\geq d,$ where $\gamma$ and $d$
are positive reals. Fix an integer $m\geq2$ and define $F\colon\zoo^{mn}\times\zoo^{mn}\to\zoo$
by $F=f\circ\OR_{m}\circ\AND_{2}.$ Then
\[
\srank(F)\geq\gamma\left\lfloor \frac{m}{2}\right\rfloor ^{d/2}.
\]
\end{thm}

\noindent For the reader's convenience, we include a detailed proof
of Theorem~\ref{thm:thrdeg-to-sign-rank} in Appendix~\ref{sec:Sign-rank-and-smooth-thrdeg}.

\section{Auxiliary results}

In this section, we collect a number of supporting results on polynomial
approximation that have appeared in one form or another in previous
work. For the reader's convenience, we provide self-contained proofs
whenever the precise formulation that we need departs from published
work.

\subsection{Basic dual objects}

As described in the introduction, we prove our main results constructively,
by building explicit dual objects that witness the corresponding lower
bounds. An important tool in this process is the following lemma due
to Razborov and Sherstov~\cite{RS07dc-dnf}. Informally, it is used
to adjust a dual object's metric properties while preserving its orthogonality
to low-degree polynomials. The lemma plays a basic role in several
recent papers~\cite{RS07dc-dnf,BT16sign-rank-ac0,bun-thaler17adeg-ac0,BKT17poly-strikes-back,BT18ac0-large-error}
as well as our work. 
\begin{lem}[Razborov and Sherstov]
\label{lem:razborov-sherstov}Fix integers $d$ and $n,$ where $0\leq d<n.$
Then there is an $($explicitly given$)$ function $\zeta\colon\zoon\to\Re$
such that
\begin{align*}
 & \supp\zeta\subseteq\zoon|_{\leq d}\cup\{1^{n}\},\\
 & \zeta(1^{n})=1,\\
 & \|\zeta\|_{1}\leq1+2^{d}\binom{n}{d},\\
 & \orth\zeta>d.
\end{align*}
\end{lem}

\noindent In more detail, this result corresponds to taking $k=d$
and $\zeta=(-1)^{n}g$ in the proof of Lemma~3.2 of~\cite{RS07dc-dnf}.
We will need the following symmetrized version of Lemma~\ref{lem:razborov-sherstov}.
\begin{lem}
\label{lem:mass-transfer}Fix a point $u\in\NN^{n}$ and a natural
number $d<|u|.$ Then there is $\zeta_{u}\colon\NN^{n}\to\Re$ such
that
\begin{align}
 & \supp\zeta_{u}\subseteq\{u\}\cup\{v\in\NN^{n}:v\leq u\text{ and }|v|\leq d\},\label{eq:zeta-support}\\
 & \zeta_{u}(u)=1,\label{eq:zeta-at-u}\\
 & \|\zeta_{u}\|_{1}\leq1+2^{d}\binom{|u|}{d},\label{eq:zeta-norm}\\
 & \orth\zeta_{u}>d.\label{eq:zeta-orth}
\end{align}
\end{lem}

\begin{proof}
Lemma~\ref{lem:razborov-sherstov} gives a function $\zeta\colon\zoo^{|u|}\to\Re$
such that
\begin{align}
 & \supp\zeta\subseteq\zoo^{|u|}|_{\leq d}\cup\{1^{|u|}\},\label{eq:rs-support}\\
 & \zeta(1^{|u|})=1,\label{eq:rs-at-u}\\
 & \|\zeta\|_{1}\leq1+2^{d}\binom{|u|}{d},\label{eq:rs-norm}\\
 & \orth\zeta>d.\label{eq:rs-orth}
\end{align}
Now define $\zeta_{u}\colon\NN^{n}\to\Re$ by
\[
\zeta_{u}(v)=\sum_{x_{1}\in\zoo^{|u_{1}|}\tallbar_{|v_{1}|}}\cdots\sum_{x_{n}\in\zoo^{|u_{n}|}\tallbar_{|v_{n}|}}\zeta(x_{1}\ldots x_{n}),
\]
where we adopt the convention that the set $\{0,1\}^{0}=\{0,1\}^{0}|_{0}$
has as its only element the empty string, with weight~$0.$ Then
properties~(\ref{eq:zeta-support})\textendash (\ref{eq:zeta-norm})
are immediate from~(\ref{eq:rs-support})\textendash (\ref{eq:rs-norm}),
respectively. To verify the remaining property~(\ref{eq:zeta-orth}),
fix a polynomial $p\colon\Re^{n}\to\Re$ of degree at most $d.$ Then
\begin{align*}
\hspace{-3mm}\langle\zeta_{u},p\rangle & =\sum_{v:v\leq u}\left(\sum_{x_{1}\in\zoo^{|u_{1}|}\tallbar_{|v_{1}|}}\!\cdots\!\sum_{x_{n}\in\zoo^{|u_{n}|}\tallbar_{|v_{n}|}}\!\!\!\zeta(x_{1}\ldots x_{n})\right)p(v_{1},\ldots,v_{n})\\
 & =\sum_{v:v\leq u}\left(\sum_{x_{1}\in\zoo^{|u_{1}|}\tallbar_{|v_{1}|}}\!\cdots\!\sum_{x_{n}\in\zoo^{|u_{n}|}\tallbar_{|v_{n}|}}\!\!\!\zeta(x_{1}\ldots x_{n})p(|x_{1}|,\ldots,|x_{n}|)\!\right)\\
 & =\sum_{x_{1}\in\zoo^{|u_{1}|}}\cdots\sum_{x_{n}\in\zoo^{|u_{n}|}}\zeta(x_{1}\ldots x_{n})p(|x_{1}|,\ldots,|x_{n}|)\\
 & =0,
\end{align*}
where the last step uses~(\ref{eq:rs-orth}).
\end{proof}
When constructing a dual polynomial for a complicated constant-depth
circuit, it is natural to start with a dual polynomial for the OR
function or, equivalently, its counterpart AND. The first such dual
polynomial was constructed by Špalek~\cite{spalek08dual-or}, with
many refinements and generalizations~\cite{bun-thaler13and-or-tree,sherstov14sign-deg-ac0,sherstov15asymmetry,bun-thaler17adeg-ac0,BKT17poly-strikes-back}
obtained in follow-up work. We augment this line of work with yet
another construction, which delivers the exact combination of analytic
and metric properties that we need.
\begin{thm}
\label{thm:dual-OR}Let $0<\epsilon<1$ be given. Then for some constants
$c',c''\in(0,1)$ and all integers $N\geq n\geq1,$ there is an $($explicitly
given$)$ function $\psi\colon\{0,1,2,\ldots,N\}\to\Re$ such that
\begin{align*}
 & \psi(0)>\frac{1-\epsilon}{2},\\
 & \|\psi\|_{1}=1,\\
 & \orth\psi\geq c'\sqrt{n},\\
 & \sign\psi(t)=(-1)^{t}, &  & t=0,1,2,\ldots,N,\\
 & |\psi(t)|\in\left[\frac{c'}{(t+1)^{2}\,2^{c''t/\sqrt{n}}},\;\frac{1}{c'(t+1)^{2}\,2^{c''t/\sqrt{n}}}\right], &  & t=0,1,2,\ldots,N.
\end{align*}
\end{thm}

\noindent A self-contained proof of Theorem~\ref{thm:dual-OR} is
available in Appendix~\ref{sec:dual-OR}.

\subsection{Dominant components}

We now recall a lemma due to Bun and Thaler~\cite{bun-thaler17adeg-ac0}
that serves to identify the dominant components of a vector. Its primary
use~\cite{bun-thaler17adeg-ac0,BKT17poly-strikes-back} is to prove
concentration-of-measure results for product distributions on $\NN^{n}.$
\begin{lem}[Bun and Thaler]
\label{lem:bun-thaler-combinatorial} Let $v\in\Re^{n}$ be given,
$v\ne0^{n}$. Then there is $S\subseteq\{1,2,\ldots,n\}$ such that
\begin{align*}
 & |S|\geq\frac{\|v\|_{1}}{2\|v\|_{\infty}},\\
 & |S|\min_{i\in S}|v_{i}|\geq\frac{\|v\|_{1}}{2(1+\ln n)}.
\end{align*}
\end{lem}

\begin{proof}[Proof \emph{(adapted from~\cite{bun-thaler17adeg-ac0}).} ]
By renumbering the indices if necessary, we may assume that $|v_{1}|\geq|v_{2}|\geq\cdots\geq|v_{n}|\geq0.$
For the sake of contradiction, suppose that no such set $S$ exists.
Then
\begin{align*}
|v_{i}| & <\frac{1}{i}\cdot\frac{\|v\|_{1}}{2(1+\ln n)}
\end{align*}
for every index $i\geq\frac{\|v\|_{1}}{2\|v\|_{\infty}}.$ As a result,
\begin{align*}
\|v\|_{1} & =\sum_{i<\frac{\|v\|_{1}}{2\|v\|_{\infty}}}|v_{i}|+\sum_{i=\left\lceil \frac{\|v\|_{1}}{2\|v\|_{\infty}}\right\rceil }^{n}|v_{i}|\\
 & \leq\sum_{i<\frac{\|v\|_{1}}{2\|v\|_{\infty}}}\|v\|_{\infty}+\sum_{i=\left\lceil \frac{\|v\|_{1}}{2\|v\|_{\infty}}\right\rceil }^{n}\frac{1}{i}\cdot\frac{\|v\|_{1}}{2(1+\ln n)}\\
 & <\frac{\|v\|_{1}}{2}+\frac{\|v\|_{1}}{2(1+\ln n)}\sum_{i=1}^{n}\frac{1}{i}\\
 & \leq\|v\|_{1},
\end{align*}
where the final step uses
\[
\sum_{i=1}^{n}\frac{1}{i}=1+\sum_{i=2}^{n}\frac{1}{i}\leq1+\int_{1}^{n}\frac{di}{i}=1+\ln n.
\]
We have arrived at $\|v\|_{1}<\|v\|_{1}$, a contradiction.
\end{proof}
\noindent We will need a slightly more general statement, which can
be thought of as an extremal analogue of Lemma~\ref{lem:bun-thaler-combinatorial}.
\begin{lem}
\label{lem:combinatorial}Fix $\theta>0$ and let $v\in\Re^{n}$ be
an arbitrary vector with $\|v\|_{1}\geq\theta.$ Then there is $S\subseteq\{1,2,\ldots,n\}$
such that
\begin{align}
 & |S|\geq\frac{\|v\|_{1}}{2\|v\|_{\infty}},\label{eq:S-large}\\
 & \min_{i\in S}|v_{i}|\geq\frac{1}{|S|}\cdot\frac{\theta}{2(1+\ln n)},\label{eq:S-area}\\
 & \sum_{i\notin S}|v_{i}|<\theta.\label{eq:S-outside}
\end{align}
\end{lem}

\begin{proof}
Fix $n,$ $v,$ and $\theta$ for the remainder of the proof. We will
refer to a subset $S\subseteq\{1,2,\ldots,n\}$ as \emph{regular}
if $S$ satisfies~(\ref{eq:S-large}) and~(\ref{eq:S-area}). Lemma~\ref{lem:bun-thaler-combinatorial}
along with $\|v\|_{1}\geq\theta$ ensures the existence of at least
one regular set. Now, let $S$ be a \emph{maximal }regular set. For
the sake of contradiction, suppose that~(\ref{eq:S-outside}) fails.
Applying Lemma~\ref{lem:bun-thaler-combinatorial} to $v|_{\overline{S}}$
produces a nonempty set $T\subseteq\overline{S}$ with
\[
\min_{i\in T}|v_{i}|\geq\frac{1}{|T|}\cdot\frac{\theta}{2(1+\ln n)}.
\]
But then $S\cup T$ is regular, contradicting the maximality of $S$.
\end{proof}
\noindent Lemma~\ref{lem:combinatorial} implies the following concentration-of-measure
result for product distributions on $\NN^{n}$; cf.~Bun and Thaler~\cite{bun-thaler17adeg-ac0}\emph{.}
\begin{lem}[cf.~Bun and Thaler]
 \label{lem:concentration-of-measure} Let $\lambda_{1},\lambda_{2},\ldots,\lambda_{n}\in\Distribution(\NN)$
be given with
\begin{align}
\lambda_{i}(t) & \leq\frac{C\alpha^{t}}{(t+1)^{2}}, &  & t\in\NN,\label{eq:lambda-i-upper}
\end{align}
where $C\geq0$ and $0\leq\alpha\leq1$. Then for all $\theta\geq8C\e n(1+\ln n),$
\[
\Prob_{v\sim\lambda_{1}\times\lambda_{2}\times\cdots\times\lambda_{n}}[\|v\|_{1}\geq\theta]\leq\alpha^{\theta/2}.
\]
\end{lem}

\begin{proof}[Proof \emph{(adapted from~\cite{bun-thaler17adeg-ac0}).}]
For any vector $v\in\NN^{n}$ with $\|v\|_{1}\geq\theta/2,$ Lemma~\ref{lem:combinatorial}
guarantees the existence of a nonempty set $S\subseteq\{1,2,\ldots,n\}$
such that
\begin{align}
 & \min_{i\in S}|v_{i}|\geq\frac{1}{|S|}\cdot\frac{\theta}{4(1+\ln n)},\label{eq:S-area-1}\\
 & \sum_{i\notin S}|v_{i}|<\frac{\theta}{2}.
\end{align}
If in addition $\|v\|_{1}\geq\theta,$ the second property implies
\begin{equation}
\sum_{i\in S}|v_{i}|>\frac{\theta}{2}.\label{eq:S-outside-1}
\end{equation}
In what follows, we refer to the combination of properties~(\ref{eq:S-area-1})
and~(\ref{eq:S-outside-1}) by saying that $v$ is \emph{$S$-heavy.}
In this terminology, every vector $v\in\NN^{n}$ with $\|v\|_{1}\geq\theta$
is $S$-heavy for some nonempty set $S.$ 

Now, consider a random vector $v\in\NN^{n}$ distributed according
to $\lambda_{1}\times\lambda_{2}\times\cdots\times\lambda_{n}$. We
have
\begin{align*}
\Prob_{v}[\|v\|_{1}\geq\theta] & \leq\Prob_{v}[\text{\ensuremath{v} is \ensuremath{S}-heavy for some nonempty \ensuremath{S}]}\\
 & \leq\sum_{\substack{S\subseteq\{1,2,\ldots,n\}\\
S\ne\varnothing
}
}\Prob_{v}[\text{\ensuremath{v} is \ensuremath{S}-heavy}]\\
 & \leq\sum_{\substack{S\subseteq\{1,2,\ldots,n\}\\
S\ne\varnothing
}
}\alpha^{\theta/2}\left(\sum_{t\geq\frac{1}{|S|}\cdot\frac{\theta}{4(1+\ln n)}}\frac{C}{(t+1)^{2}}\right)^{|S|}\\
 & \leq\sum_{\substack{S\subseteq\{1,2,\ldots,n\}\\
S\ne\varnothing
}
}\alpha^{\theta/2}\left(C\int_{\frac{1}{|S|}\cdot\frac{\theta}{4(1+\ln n)}}^{\infty}\frac{dt}{t^{2}}\right)^{|S|}\allowdisplaybreaks\\
 & =\sum_{\substack{S\subseteq\{1,2,\ldots,n\}\\
S\ne\varnothing
}
}\alpha^{\theta/2}\left(\frac{C|S|\cdot4(1+\ln n)}{\theta}\right)^{|S|}\\
 & =\sum_{s=1}^{n}\binom{n}{s}\cdot\alpha^{\theta/2}\left(\frac{Cs\cdot4(1+\ln n)}{\theta}\right)^{s}\\
 & \leq\sum_{s=1}^{n}\alpha^{\theta/2}\left(\frac{\e n}{s}\cdot\frac{Cs\cdot4(1+\ln n)}{\theta}\right)^{s}\\
 & \leq\alpha^{\theta/2},
\end{align*}
where the first inequality holds by the opening paragraph of the proof;
the second step applies the union bound; the third step uses $0\leq\alpha\leq1$
and the upper bound~(\ref{eq:lambda-i-upper}) for the $\lambda_{i}$;
and the last two steps use~(\ref{eq:entropy-bound-binomial}) and
the hypothesis that $\theta\geq8C\e n(1+\ln n),$ respectively.
\end{proof}

\subsection{Input transformation}

We work almost exclusively with Boolean functions on $\NN^{n}|_{\leq\theta}$,
for appropriate settings of the dimension parameter $n$ and weight
parameter $\theta.$ This choice of domain is admittedly unusual but
greatly simplifies the analysis. Fortunately, approximation-theoretic
results obtained in this setting carry over in a black-box manner
to the hypercube. In more detail, we will now prove that every function
on $\NN^{n}|_{\leq\theta}$ can be transformed into a function in
$O(\theta\log n)$ Boolean variables with similar approximation-theoretic
properties. Analogous input transformations, with similar proofs,
have been used in previous work to translate results from $\zoo^{n}|_{\theta}$
or $\zoo^{n}|_{\leq\theta}$ to the hypercube setting~\cite{bun-thaler17adeg-ac0,BKT17poly-strikes-back}.
The presentation below seems more economical than previous treatments.

Recall that $e_{1},e_{2},\ldots,e_{n}$ denote the standard basis
for $\Re^{n}.$ The following encoding lemma was proved in~\cite[Lemma~3.1]{sherstov15asymmetry}. 
\begin{lem}[Sherstov]
\label{lem:reencoding}Let $n\geq1$ be a given integer. Then there
is a surjection $g\colon\zoo^{6\lceil\log(n+1)\rceil}\to\{0^{n},e_{1},e_{2},\ldots,e_{n}\}$
such that 
\begin{align*}
\Exp_{g^{-1}(0^{n})}\;p & =\Exp_{g^{-1}(e_{1})}\;p=\Exp_{g^{-1}(e_{2})}\;p=\cdots=\Exp_{g^{-1}(e_{n})}\;p
\end{align*}
for every polynomial $p$ of degree at most $\lceil\log(n+1)\rceil$.
Moreover, $g$ can be constructed deterministically in time polynomial
in $n.$
\end{lem}

\noindent Observe that the points $0^{n},e_{1},e_{2},\ldots,e_{n}$
in this lemma act simply as labels and can be replaced with any other
tuple of $n+1$ distinct points. Indeed, this result was originally
stated in~\cite{sherstov15asymmetry} for a different choice of points.
A tensor version of Lemma~\ref{lem:reencoding} is as follows.
\begin{lem}
\label{lem:gadget-g-tensor}Let $g\colon\zoo^{6\lceil\log(n+1)\rceil}\to\{0^{n},e_{1},e_{2},\ldots,e_{n}\}$
be as constructed in Lemma~\emph{\ref{lem:reencoding}}. Then for
any integer $\theta\geq1$ and any polynomial $p\colon(\Re^{6\lceil\log(n+1)\rceil})^{\theta}\to\Re,$
the mapping
\[
(y_{1},y_{2},\ldots,y_{\theta})\mapsto\Exp_{g^{-1}(y_{1})\times g^{-1}(y_{2})\times\cdots\times g^{-1}(y_{\theta})}p
\]
is a polynomial in $y\in\{0^{n},e_{1},e_{2},\ldots,e_{n}\}^{\theta}$
of degree at most $(\deg p)/\lceil\log(n+1)+1\rceil.$
\end{lem}

\begin{proof}
By linearity, it suffices to prove the lemma for factored polynomials
of the form $p(x_{1},x_{2},\dots,x_{\theta})=p_{1}(x_{1})p_{2}(x_{2})\cdots p_{\theta}(x_{\theta})$,
where $p_{1},p_{2},\dots,p_{\theta}$ are real polynomials on $\Re^{6\lceil\log(n+1)\rceil}.$
For such a polynomial $p$, the defining equation simplifies to 
\begin{equation}
\Exp_{g^{-1}(y_{1})\times g^{-1}(y_{2})\times\cdots\times g^{-1}(y_{\theta})}p\;=\prod_{i=1}^{n}\;\Exp_{g^{-1}(y_{i})}p_{i}.\label{eq:encoding-expectation-p-star}
\end{equation}
We now examine the individual contributions of $p_{1},p_{2},\dots,p_{\theta}$
to the degree of the right-hand side as a real polynomial in $y$.
For any polynomial $p_{i}$ of degree at most $\lceil\log(n+1)\rceil,$
Lemma~\ref{lem:reencoding} ensures that the corresponding expectation
$\Exp_{g^{-1}(y_{i})}p_{i}$ is a constant independent of the input
$y_{i}$. Thus, polynomials $p_{i}$ of degree at most $\lceil\log(n+1)\rceil$
do not contribute to the degree of the right-hand side of~(\ref{eq:encoding-expectation-p-star}).
For the other polynomials $p_{i}$, the expectation $\Exp_{g^{-1}(y_{i})}p_{i}$
is a linear polynomial in $y_{i}$, namely, 
\begin{multline*}
\Exp_{g^{-1}(y_{i})}p_{i}=y_{i,1}\Exp_{g^{-1}(e_{1})}p_{i}+y_{i,2}\Exp_{g^{-1}(e_{2})}p_{i}+\cdots+y_{i,n}\Exp_{g^{-1}(e_{n})}p_{i}\\
+\left(1-\sum_{j=1}^{n}y_{i,j}\right)\Exp_{g^{-1}(0^{n})}p_{i},
\end{multline*}
where we are crucially exploiting the fact that $y_{i}\in\{0^{n},e_{1},e_{2},\dots,e_{n}\}$.
Thus, polynomials $p_{i}$ of degree greater than $\lceil\log(n+1)\rceil$
contribute at most $1$ each to the degree. Summarizing, the right-hand
side of~(\ref{eq:encoding-expectation-p-star}) is a real polynomial
in $y_{1},y_{2},\ldots,y_{\theta}$ of degree at most 
\[
|\{i:\deg p_{i}\geq\lceil\log(n+1)\rceil+1\}|\leq\frac{\deg p}{\lceil\log(n+1)\rceil+1},
\]
as was to be shown.
\end{proof}
We have reached the claimed result on input transformation.
\begin{thm}
\label{thm:input-compression}Let $n,\theta\geq1$ be given integers.
Set $N=6\lceil\log(n+1)\rceil\theta.$ Then there is a surjection
$G\colon\zoo^{N}\to\NN^{n}|_{\leq\theta}$ such that:
\begin{enumerate}
\item \label{item:encoding-degree}for every polynomial $p\colon\Re^{N}\to\Re,$
the mapping $v\mapsto\Exp_{G^{-1}(v)}p$ is a polynomial on $\NN^{n}|_{\leq\theta}$
of degree at most $(\deg p)/\lceil\log(n+1)+1\rceil;$
\item \label{item:encoding-circuit-complexity}for every coordinate $i=1,2,\ldots,n,$
the mapping $x\mapsto\OR_{\theta}^{*}(G(x)_{i})$ is computable by
an explicitly given DNF formula with $O(\theta n^{6})$ terms, each
with at most $6\lceil\log(n+1)\rceil$ variables.
\end{enumerate}
\end{thm}

\noindent Applying Theorem~\ref{thm:input-compression} to a function
$f\colon\NN^{n}|_{\leq\theta}\to\zoo$ produces a composed function
$f\circ G\colon\zoo^{6\lceil\log(n+1)\rceil\theta}\to\zoo$ in the
hypercube setting. The theorem ensures that lower bounds for the pointwise
approximation, or sign-representation, of $f$ apply to $f\circ G$
as well. Moreover, the circuit complexity of $f\circ G$ is only slightly
higher than that of $f.$ This way, Theorem~\ref{thm:input-compression}
efficiently transfers approximation-theoretic results from $\NN^{n}|_{\leq\theta}$
(or any subset thereof, such as $\zoon|_{\leq\theta}$ or $\NN^{n}|_{\theta}$)
to the traditional setting of the hypercube.
\begin{proof}[Proof of Theorem~\emph{\ref{thm:input-compression}}.]
Define $G\colon(\zoo^{6\lceil\log(n+1)\rceil})^{\theta}\to\NN^{n}|_{\leq\theta}$
by
\[
G(x_{1},x_{2},\ldots,x_{\theta})=g(x_{1})+g(x_{2})+\cdots+g(x_{\theta}),
\]
where $g\colon\zoo^{6\lceil\log(n+1)\rceil}\to\{0^{n},e_{1},e_{2},\ldots,e_{n}\}$
is as constructed in Lemma~\ref{lem:reencoding}. The surjectivity
of $G$ follows trivially from that of $g.$ We proceed to verify
the additional properties required of $G.$

\ref{item:encoding-degree}~For $v\in\NN^{n}|_{\leq\theta},$ we
have the partition
\begin{equation}
G^{-1}(v)=\bigcup_{\substack{y\in\{0^{n},e_{1},e_{2},\ldots,e_{n}\}^{\theta}:\\
y_{1}+y_{2}+\cdots+y_{\theta}=v
}
}g^{-1}(y_{1})\times g^{-1}(y_{2})\times\cdots\times g^{-1}(y_{\theta}).\label{eq:encoding-tau-v-decomposition}
\end{equation}
All representations $v=y_{1}+y_{2}+\cdots+y_{\theta}$ with $y_{1},y_{2},\ldots,y_{\theta}\in\{0^{n},e_{1},e_{2},\ldots,e_{n}\}$
are the same up to the order of the summands. As a result, each part
$g^{-1}(y_{1})\times g^{-1}(y_{2})\times\cdots\times g^{-1}(y_{\theta})$
in the partition on the right-hand side of~(\ref{eq:encoding-tau-v-decomposition})
has the same cardinality. We conclude that for any given polynomial
$p,$
\begin{align}
\Exp_{G^{-1}(v)} & p=\Exp_{\substack{y\in\{0^{n},e_{1},e_{2},\ldots,e_{n}\}^{\theta}:\\
y_{1}+y_{2}+\cdots+y_{\theta}=v
}
}\quad\Exp_{g^{-1}(y_{1})\times g^{-1}(y_{2})\times\cdots\times g^{-1}(y_{\theta})}\;p.\label{eq:encoding-tau-exp-p}
\end{align}
Recall from Lemma~\ref{lem:gadget-g-tensor} that the rightmost expectation
in this equation is a polynomial in $y_{1},y_{2},\ldots,y_{\theta}\in\{0^{n},e_{1},e_{2},\ldots,e_{n}\}$
of degree at most $(\deg p)/\lceil\log(n+1)+1\rceil.$ As a result,
Corollary~\ref{cor:ambainis-symmetrization} implies that the right-hand
side of~(\ref{eq:encoding-tau-exp-p}) is a polynomial in $v$ of
degree at most $(\deg p)/\lceil\log(n+1)+1\rceil$.

\ref{item:encoding-circuit-complexity} Fix an index $i$. Then
\[
\OR_{\theta}^{*}(G(x)_{i})=\bigvee_{j=1}^{\theta}\I[g(x_{j})=e_{i}].
\]
Each of the disjuncts on the right-hand side is a function of $6\lceil\log(n+1)\rceil$
Boolean variables. Therefore, $\OR_{\theta}^{*}(G(x)_{i})$ is representable
by a DNF formula with $O(\theta n^{6})$ terms, each with at most
$6\lceil\log(n+1)\rceil$ variables.
\end{proof}

\section{\label{sec:Threshold-degree-of-AC0}The threshold degree of AC\protect\textsuperscript{0}}

This section is devoted to our results on threshold degree. While
we are mainly interested in the threshold degree of $\classAC^{0},$
the techniques developed here apply to a much broader class of functions.
Specifically, we prove an \emph{amplification theorem} that takes
an arbitrary function $f$ and builds from it a function $F$ with
higher threshold degree. We give analogous amplification theorems
for various other approximation-theoretic quantities. The transformation
$f\mapsto F$ is efficient with regard to circuit depth and size and
in particular preserves membership in $\classAC^{0}$. To deduce our
main results for $\classAC^{0}$, we start with a single-gate circuit
and iteratively apply the amplification theorem to produce constant-depth
circuits of higher and higher threshold degree. We develop this general
machinery in Sections~\ref{subsec:Shifting-probability-mass}\textendash \ref{subsec:Hardness-amplification-for},
followed by the applications to $\classAC^{0}$ in Sections~\ref{subsec:Results-for-AC}
and~\ref{subsec:surjectivity}. 

\subsection{\label{subsec:Shifting-probability-mass}Shifting probability mass
in product distributions}

Consider a product distribution $\Lambda$ on $\NN^{n}$ whereby every
component is concentrated near $0$. The centerpiece of our threshold
degree analysis, presented here, is the construction of an associated
probability distribution $\tilde{\Lambda}$ that is supported entirely
on inputs of low weight and cannot be distinguished from $\Lambda$
by a low-degree polynomial. More formally, define $\WeaklyBounded(r,c,\alpha)$
to be the family of probability distributions $\lambda$ on $\NN$
such that 
\[
\supp\lambda=\{0,1,2,\ldots,r'\}
\]
 for some nonnegative integer $r'\leq r,$ and in addition
\begin{align}
 & \frac{c^{t+1}}{(t+1)^{2}\,2^{\alpha t}}\leq\lambda(t)\leq\frac{1}{c(t+1)^{2}\,2^{\alpha t}}, &  & t\in\supp\lambda.\label{eq:B-tilde}
\end{align}
Distributions in this family are subject to pointwise constraints,
hence the symbol $\WeaklyBounded$ for ``bounded.'' Our bounding
functions are motivated mainly by the metric properties of the dual
polynomial for $\OR_{n}$, constructed in Theorem~\ref{thm:dual-OR}. 

In this notation, our analysis handles any distribution $\Lambda\in\WeaklyBounded(r,c,\alpha)^{\otimes n}.$
It would be possible to generalize our work further, but the lower
and upper bounds in~(\ref{eq:B-tilde}) are already exponentially
far apart and capture a much larger class of probability distributions
than what we need for the applications to $\classAC^{0}$. The precise
statement of our result is as follows.
\begin{thm}
\label{thm:weight-reduction}Let $\Lambda\in\WeaklyBounded(r,c,\alpha)^{\otimes n}$
be given, for some integer $r\geq0$ and reals $c>0$ and $\alpha\geq0.$
Let $d$ and $\theta$ be positive integers with
\begin{align}
\theta & \geq2d,\label{eq:tilde-Lambda-d-theta}\\
\theta & \geq\frac{4\e n(1+\ln n)}{c^{2}}.\label{eq:tilde-Lambda-theta}
\end{align}
Then there is a function $\tilde{\Lambda}\colon\NN^{n}\to\Re$ such
that
\begin{align}
 & \supp\tilde{\Lambda}\subseteq(\supp\Lambda)|_{<2\theta},\label{eq:tilde-Lambda-support}\\
 & \orth(\Lambda-\tilde{\Lambda})>d,\label{eq:tilde-Lambda-orthog}\\
 & |\Lambda-\tilde{\Lambda}|\leq\left(\frac{8nr}{c}\right)^{d}2^{-\lceil\theta/r\rceil-\alpha\lceil\theta/2\rceil+2}\,\Lambda\qquad\qquad\text{\emph{on} }\NN^{n}|_{<2\theta}.\label{eq:tilde-Lambda-statistical}
\end{align}
\end{thm}

\noindent In general, the function $\tilde{\Lambda}$ constructed
in Theorem~\ref{thm:weight-reduction} may not be a probability distribution.
However, when $\theta$ is large enough relative to the other parameters,
the pointwise property~(\ref{eq:tilde-Lambda-statistical}) forces
$|\Lambda-\tilde{\Lambda}|\leq\Lambda$ on the support of $\tilde{\Lambda}$,
and in particular $\tilde{\Lambda}\geq0.$ Since $\orth(\Lambda-\tilde{\Lambda})>0$
by construction, Proposition~\ref{prop:distribution-criterion} guarantees
that $\tilde{\Lambda}$ \emph{is} a probability distribution in that
case.
\begin{proof}[Proof of Theorem~\emph{\ref{thm:weight-reduction}}.]
 For $c>1,$ we have $\WeaklyBounded(r,c,\alpha)=\varnothing$ and
the theorem holds vacuously. Another degenerate possibility is $r=0,$
in which case $\Lambda$ is the single-point distribution on $0^{n}$,
and therefore it suffices to take $\tilde{\Lambda}=\Lambda.$ In what
follows, we treat the general case when 
\begin{align*}
 & c\in(0,1],\\
 & r\geq1.
\end{align*}

For every vector $v\in\NN^{n}$ with $\|v\|_{1}\geq\theta,$ let $S(v)\subseteq\{1,2,\ldots,n\}$
denote the corresponding subset identified by Lemma~\ref{lem:combinatorial}.
To restate the lemma's guarantees,
\begin{align}
 & |S(v)|\geq\frac{\theta}{r}, &  & v\in(\supp\Lambda)|_{\geq2\theta},\label{eq:Sv-large}\\
 & \min_{i\in S(v)}v_{i}\geq\frac{\theta}{2|S(v)|(1+\ln n)}, &  & v\in(\supp\Lambda)|_{\geq2\theta},\label{eq:Sv-heavy}\\
\rule{0mm}{5mm} & \|v|_{\overline{S(v)}}\|_{1}<\theta. &  & v\in(\supp\Lambda)|_{\geq2\theta}.\label{eq:Sv-outside}
\end{align}
Property~(\ref{eq:Sv-outside}) implies that
\begin{align}
\|v|_{S(v)}\|_{1} & >\theta, &  & v\in(\supp\Lambda)|_{\geq2\theta},\label{eq:v-Sv-heavy-theta}
\end{align}
and in particular
\begin{align}
\|v|_{S(v)}\|_{1} & >d, &  & v\in(\supp\Lambda)|_{\geq2\theta}.\label{eq:v-Sv-heavy-d}
\end{align}
For each $i=1,2,\ldots,n$ and each $u\in\NN^{i}|_{>d},$ Lemma~\ref{lem:mass-transfer}
gives a function $\zeta_{u}\colon\NN^{i}\to\Re$ such that
\begin{align}
 & \supp\zeta_{u}\subseteq\{u\}\cup\{v\in\NN^{i}:v\leq u\text{ and }|v|\leq d\},\label{eq:zeta-u-support}\\
 & \zeta_{u}(u)=1,\label{eq:zeta-u-at-u}\\
 & \|\zeta_{u}\|_{1}\leq1+2^{d}\binom{\|u\|_{1}}{d},\label{eq:zeta-u-ell1}\\
 & \orth\zeta_{u}>d,\label{eq:zeta-u-orth}
\end{align}
and in particular
\begin{align}
\|\zeta_{u}\|_{\infty} & \leq\max\{|\zeta_{u}(u)|,\|\zeta_{u}\|_{1}-|\zeta_{u}(u)|\}\nonumber \\
 & \leq2^{d}\binom{\|u\|_{1}}{d}\nonumber \\
 & \leq2\|u\|_{1}{}^{d}.\label{eq:zeta-u-infty-norm}
\end{align}

The central object of study in our proof is the following function
$\zeta\colon\NN^{n}\to\Re$, built from the auxiliary objects $S(v)$
and $\zeta_{u}$ just introduced:
\begin{equation}
\zeta(x)=\sum_{v\in(\supp\Lambda)|_{\geq2\theta}}\Lambda(v)\,\zeta_{v|_{S(v)}}(x|_{S(v)})\,\I[x|_{\overline{S(v)}}=v|_{\overline{S(v)}}].\label{eq:zeta-defined}
\end{equation}
The expression on the right-hand side is well-formed because, to restate~(\ref{eq:v-Sv-heavy-d}),
each string $v|_{S(v)}$ has weight greater than $d$ and can therefore
be used as a subscript in $\zeta_{v|_{S(v)}}.$ Specializing~(\ref{eq:zeta-u-orth})
and~(\ref{eq:zeta-u-infty-norm}),
\begin{align}
 & \orth\zeta_{v|_{S(v)}}>d, &  & v\in(\supp\Lambda)|_{\geq2\theta},\label{eq:zeta-vS-orth}\\
 & \|\zeta_{v|_{S(v)}}\|_{\infty}\leq2(nr)^{d}, &  & v\in(\supp\Lambda)|_{\geq2\theta}.\label{eq:zeta-vS-infty-norm}
\end{align}
Property~(\ref{eq:zeta-u-support}) ensures that $\zeta_{v|_{S(v)}}(x|_{S(v)})\,\I[x|_{\overline{S(v)}}=v|_{\overline{S(v)}}]\ne0$
only when $x\leq v.$ It follows that
\begin{align}
\supp\zeta & \subseteq\bigcup_{v\in\supp\Lambda}\{x\in\NN^{n}:x\leq v\}\nonumber \\
\rule{0mm}{5mm} & =\supp\Lambda,\label{eq:supp-Zeta-supp-Lambda}
\end{align}
where second step is valid because $\Lambda\in\WeaklyBounded(r,c,\alpha)^{\otimes n}.$

Before carrying on with the proof, we take a moment to simplify the
defining expression for $\zeta.$ For any $v\in\NN^{n}|_{\geq2\theta},$
we have
\begin{align*}
 & \zeta_{v|_{S(v)}}(x|_{S(v)})\,\I[x|_{\overline{S(v)}}=v|_{\overline{S(v)}}]\\
 & \qquad=\zeta_{v|_{S(v)}}(x|_{S(v)})\,\I[x|_{S(v)}=v|_{S(v)}\text{ or }\|x|_{S(v)}\|_{1}\leq d]\,\I[x|_{\overline{S(v)}}=v|_{\overline{S(v)}}]\\
 & \qquad=\zeta_{v|_{S(v)}}(x|_{S(v)})(\I[x|_{S(v)}=v|_{S(v)}]+\I[\|x|_{S(v)}\|_{1}\leq d])\I[x|_{\overline{S(v)}}=v|_{\overline{S(v)}}]\\
 & \qquad=\zeta_{v|_{S(v)}}(x|_{S(v)})\I[x=v]\\
 & \qquad\qquad\qquad\qquad+\zeta_{v|_{S(v)}}(x|_{S(v)})\I[\|x|_{S(v)}\|_{1}\leq d]\,\I[x|_{\overline{S(v)}}=v|_{\overline{S(v)}}]\\
 & \qquad=\I[x=v]+\zeta_{v|_{S(v)}}(x|_{S(v)})\I[\|x|_{S(v)}\|_{1}\leq d]\,\I[x|_{\overline{S(v)}}=v|_{\overline{S(v)}}],
\end{align*}
where the first, second, and fourth steps are valid by~(\ref{eq:zeta-u-support}),
(\ref{eq:v-Sv-heavy-d}), and~(\ref{eq:zeta-u-at-u}), respectively.
Making this substitution in the defining equation for $\zeta,$
\begin{multline}
\zeta(x)=\sum_{v\in(\supp\Lambda)|_{\geq2\theta}}\Lambda(v)\zeta_{v|_{S(v)}}(x|_{S(v)})\I[\|x|_{S(v)}\|_{1}\leq d]\,\I[x|_{\overline{S(v)}}=v|_{\overline{S(v)}}]\\
+\sum_{v\in(\supp\Lambda)|_{\geq2\theta}}\Lambda(v)\I[x=v].\qquad\label{eq:zeta-simplified}
\end{multline}
We proceed to establish key properties of $\zeta$.\\
~

\textsc{Step 1: Orthogonality.} By Proposition~\ref{prop:orth}\ref{item:orth-tensor},
each term in the summation on the right-hand side of~(\ref{eq:zeta-defined})
is a function orthogonal to polynomials of degree less than $\orth\zeta_{v|_{S(v)}}.$
Therefore,
\begin{align}
\orth\zeta & \geq\min_{v\in(\supp\Lambda)|_{\geq2\theta}}\orth\zeta_{v|_{S(v)}}\nonumber \\
 & >d,\label{eq:zeta-orthog}
\end{align}
where the first step uses Proposition~\ref{prop:orth}\ref{item:orth-sum},
and the second step applies~(\ref{eq:zeta-vS-orth}).

~

\textsc{Step 2: Heavy inputs.} We now examine the behavior of $\zeta$
on inputs of weight at least $2\theta,$ which we think of as ``heavy.''
For any string $v\in(\supp\Lambda)|_{\geq2\theta},$ we have
\begin{align*}
x\in\NN^{n}|_{\geq2\theta} & \implies\|x\|_{1}>d+\theta\\
 & \implies\|x|_{S(v)}\|_{1}>d\quad\vee\quad\|x|_{\overline{S(v)}}\|_{1}>\theta\\
 & \implies\|x|_{S(v)}\|_{1}>d\quad\vee\quad x|_{\overline{S(v)}}\ne v|_{\overline{S(v)}},
\end{align*}
where the final implication uses~(\ref{eq:Sv-outside}). We conclude
that the first summation in~(\ref{eq:zeta-simplified}) vanishes
on $\NN^{n}|_{\geq2\theta},$ so that
\begin{align}
 & \zeta(x)=\Lambda(x), &  & x\in\NN^{n}|_{\geq2\theta}.\label{eq:zeta-heavy}
\end{align}
This completes the analysis of heavy inputs.

~

\textsc{Step 3: Light inputs.} We now turn to inputs of weight less
than $2\theta$, the most technical part of the proof. Fix an arbitrary
string $x\in(\supp\Lambda)|_{<2\theta}$. Then 
\begin{align}
\!\!\frac{|\zeta(x)|}{\Lambda(x)} & =\left|\sum_{v\in(\supp\Lambda)|_{\geq2\theta}}\frac{\Lambda(v)}{\Lambda(x)}\zeta_{v|_{S(v)}}(x|_{S(v)})\,\I[\|x|_{S(v)}\|_{1}\leq d]\,\I[x|_{\overline{S(v)}}=v|_{\overline{S(v)}}]\right|\nonumber \\
 & \leq\sum_{v\in(\supp\Lambda)|_{\geq2\theta}}\frac{\Lambda(v)}{\Lambda(x)}|\zeta_{v|_{S(v)}}(x|_{S(v)})|\,\I[\|x|_{S(v)}\|_{1}\leq d]\,\I[x|_{\overline{S(v)}}=v|_{\overline{S(v)}}]\nonumber \\
 & \leq2(nr)^{d}\sum_{v\in(\supp\Lambda)|_{\geq2\theta}}\frac{\Lambda(v)}{\Lambda(x)}\,\I[\|x|_{S(v)}\|_{1}\leq d]\,\I[x|_{\overline{S(v)}}=v|_{\overline{S(v)}}]\nonumber \\
 & =2(nr)^{d}\sum_{\substack{S\subseteq\{1,\ldots,n\}:\\
|S|\geq\theta/r
}
}\I[\|x|_{S}\|_{1}\leq d]\sum_{\substack{v\in(\supp\Lambda)|_{\geq2\theta}:\\
S(v)=S
}
}\frac{\Lambda(v)}{\Lambda(x)}\,\I[x|_{\overline{S}}=v|_{\overline{S}}]\nonumber \\
 & \leq2(nr)^{d}\sum_{\substack{S\subseteq\{1,\ldots,n\}:\\
|S|\geq\theta/r
}
}\I[\|x|_{S}\|_{1}\leq d]\!\!\!\!\sum_{\substack{v\in\NN^{n}:\\
\sum_{i\in S}v_{i}\ge\theta,\\
\min_{i\in S}v_{i}\geq\frac{\theta}{2|S|(1+\ln n)}
}
}\!\!\!\!\frac{\Lambda(v)}{\Lambda(x)}\,\I[x|_{\overline{S}}=v|_{\overline{S}}],\label{eq:z-Lambda-intermediate}
\end{align}
where the first step uses~(\ref{eq:zeta-simplified}); the second
step applies the triangle inequality; the third step is valid by~(\ref{eq:zeta-vS-infty-norm});
the fourth step amounts to collecting terms according to $S(v),$
which by~(\ref{eq:Sv-large}) has cardinality at least $\theta/r$;
and the fifth step uses~(\ref{eq:Sv-heavy}) and~(\ref{eq:v-Sv-heavy-theta}). 

Bounding~(\ref{eq:z-Lambda-intermediate}) requires a bit of work.
To start with, write $\Lambda=\bigotimes_{i=1}^{n}\lambda_{i}$ for
some $\lambda_{1},\lambda_{2},\ldots,\lambda_{n}\in\WeaklyBounded(r,c,\alpha).$
Then for every nonempty set $S\subseteq\{1,2,\ldots,n\},$
\begin{align}
\I[\|x|_{S}\|_{1}\leq d]\prod_{i\in S}\frac{1}{\lambda_{i}(x_{i})} & \leq\I[\|x|_{S}\|_{1}\leq d]\prod_{i\in S}\frac{(x_{i}+1)^{2}\,2^{\alpha x_{i}}}{c^{x_{i}+1}}\nonumber \\
 & =\I[\|x|_{S}\|_{1}\leq d]\,c^{-|S|}\left(\frac{2^{\alpha}}{c}\right)^{\sum_{i\in S}x_{i}}\prod_{i\in S}(x_{i}+1)^{2}\nonumber \\
 & \leq\I[\|x|_{S}\|_{1}\leq d]\,c^{-|S|}\left(\frac{2^{\alpha}\e^{2}}{c}\right)^{\sum_{i\in S}x_{i}}\nonumber \\
 & \leq c^{-|S|}\left(\frac{2^{\alpha}\e^{2}}{c}\right)^{d},\label{eq:mu-i-lower}
\end{align}
where the first step applies the definition of $\WeaklyBounded(r,c,\alpha)$,
and the third step uses the bound $1+t\leq\e^{t}$ for real $t.$
Continuing,
\begin{align}
\sum_{\substack{v\in\NN^{n}:\\
\sum_{i\in S}v_{i}\ge\theta,\\
\min_{i\in S}v_{i}\geq\frac{\theta}{2|S|(1+\ln n)}
}
} & \frac{\Lambda(v)}{\Lambda(x)}\,\I[x|_{\overline{S}}=v|_{\overline{S}}]\nonumber \\
 & =\sum_{\substack{v\in\NN^{n}:\\
\sum_{i\in S}v_{i}\ge\theta,\\
\min_{i\in S}v_{i}\geq\frac{\theta}{2|S|(1+\ln n)},\\
v_{i}=x_{i}\text{ for }i\notin S
}
}\prod_{i\in S}\frac{\lambda_{i}(v_{i})}{\lambda_{i}(x_{i})}\nonumber \\
 & \leq\sum_{\substack{v\in\NN^{n}:\\
\sum_{i\in S}v_{i}\ge\theta,\\
\min_{i\in S}v_{i}\geq\frac{\theta}{2|S|(1+\ln n)},\\
v_{i}=x_{i}\text{ for }i\notin S
}
}2^{-\alpha\sum_{i\in S}v_{i}}\prod_{i\in S}\frac{1}{c(v_{i}+1)^{2}\lambda_{i}(x_{i})}\allowdisplaybreaks\nonumber \\
 & \leq\sum_{\substack{v\in\NN^{n}:\\
\min_{i\in S}v_{i}\geq\frac{\theta}{2|S|(1+\ln n)},\\
v_{i}=x_{i}\text{ for }i\notin S
}
}2^{-\alpha\theta}\prod_{i\in S}\frac{1}{c(v_{i}+1)^{2}\lambda_{i}(x_{i})}\allowdisplaybreaks\nonumber \\
 & =2^{-\alpha\theta}\left(\sum_{t=\left\lceil \frac{\theta}{2|S|(1+\ln n)}\right\rceil }^{\infty}\frac{1}{c(t+1)^{2}}\right)^{|S|}\prod_{i\in S}\frac{1}{\lambda_{i}(x_{i})}\nonumber \\
 & \leq2^{-\alpha\theta}\left(\int_{\left\lceil \frac{\theta}{2|S|(1+\ln n)}\right\rceil }^{\infty}\frac{dt}{ct^{2}}\right)^{|S|}\prod_{i\in S}\frac{1}{\lambda_{i}(x_{i})}\nonumber \\
 & \leq2^{-\alpha\theta}\left(\frac{2|S|(1+\ln n)}{c\theta}\right)^{|S|}\prod_{i\in S}\frac{1}{\lambda_{i}(x_{i})},\label{eq:Lambda-inner-sum}
\end{align}
where the first step uses $\Lambda=\bigotimes_{i=1}^{n}\lambda_{i},$
and the second step applies the definition of $\WeaklyBounded(r,c,\alpha)$.

It remains to put together the bounds obtained so far. We have:
\begin{align*}
\frac{|\zeta(x)|}{\Lambda(x)} & \leq2(nr)^{d}\sum_{\substack{S\subseteq\{1,\ldots,n\}:\\
|S|\geq\theta/r
}
}\I[\|x|_{S}\|_{1}\leq d]\cdot2^{-\alpha\theta}\left(\frac{2|S|(1+\ln n)}{c\theta}\right)^{|S|}\prod_{i\in S}\frac{1}{\lambda_{i}(x_{i})}\\
 & \leq2(nr)^{d}\sum_{\substack{S\subseteq\{1,\ldots,n\}:\\
|S|\geq\theta/r
}
}2^{-\alpha\theta}\left(\frac{2|S|(1+\ln n)}{c^{2}\theta}\right)^{|S|}\cdot\left(\frac{2^{\alpha}\e^{2}}{c}\right)^{d}\\
 & \leq2\cdot\frac{(\e^{2}nr/c)^{d}}{2^{\alpha\lceil\theta/2\rceil}}\sum_{\substack{S\subseteq\{1,\ldots,n\}:\\
|S|\geq\theta/r
}
}\left(\frac{2|S|(1+\ln n)}{c^{2}\theta}\right)^{|S|}\allowdisplaybreaks\\
 & =2\cdot\frac{(\e^{2}nr/c)^{d}}{2^{\alpha\lceil\theta/2\rceil}}\sum_{s=\lceil\theta/r\rceil}^{\infty}\binom{n}{s}\left(\frac{2s(1+\ln n)}{c^{2}\theta}\right)^{s}\\
 & \leq2\cdot\frac{(\e^{2}nr/c)^{d}}{2^{\alpha\lceil\theta/2\rceil}}\sum_{s=\lceil\theta/r\rceil}^{\infty}\left(\frac{\e n}{s}\cdot\frac{2s(1+\ln n)}{c^{2}\theta}\right)^{s}\\
 & \leq2\cdot\frac{(\e^{2}nr/c)^{d}}{2^{\alpha\lceil\theta/2\rceil}}\sum_{s=\lceil\theta/r\rceil}^{\infty}2^{-s}\\
 & =4\cdot\frac{(\e^{2}nr/c)^{d}}{2^{\alpha\lceil\theta/2\rceil+\lceil\theta/r\rceil}},
\end{align*}
where the first step follows from~(\ref{eq:z-Lambda-intermediate})
and (\ref{eq:Lambda-inner-sum}); the second step substitutes the
bound from~(\ref{eq:mu-i-lower}); the third step uses~(\ref{eq:tilde-Lambda-d-theta});
and the next-to-last step uses~(\ref{eq:tilde-Lambda-theta}). In
summary, we have shown that
\begin{align}
 & |\zeta(x)|\leq4\cdot\frac{(\e^{2}nr/c)^{d}}{2^{\alpha\lceil\theta/2\rceil+\lceil\theta/r\rceil}}\,\Lambda(x), & x\in(\supp\Lambda)|_{<2\theta}.\label{eq:zeta-lambda-light-inputs}
\end{align}

\textsc{~}

\textsc{Step 4: Finishing the proof. }Define $\tilde{\Lambda}=\Lambda-\zeta.$
Then the support property~(\ref{eq:tilde-Lambda-support}) follows
from~(\ref{eq:supp-Zeta-supp-Lambda}) and~(\ref{eq:zeta-heavy});
the analytic indistinguishability property~(\ref{eq:tilde-Lambda-orthog})
follows from~(\ref{eq:zeta-orthog}); and the pointwise property~(\ref{eq:tilde-Lambda-statistical})
follows from~(\ref{eq:supp-Zeta-supp-Lambda}) and~(\ref{eq:zeta-lambda-light-inputs}).
\end{proof}
We record a generalization of Theorem~\ref{thm:weight-reduction}
to translates of probability distributions in $\WeaklyBounded(r,c,\alpha)^{\otimes n},$
and further to convex combinations of such distributions. Formally,
define $\WeaklyBounded(r,c,\alpha,\Delta)$ for $\Delta\geq0$ to
be the family of probability distributions $\lambda$ on $\NN$ such
that $\lambda(t)\equiv\lambda'(t-a)$ for some $\lambda'\in\WeaklyBounded(r,c,\alpha)$
and $a\in[0,\Delta]$. We have: 
\begin{cor}
\label{cor:weight-reduction}Let $\Lambda\in\conv(\WeaklyBounded(r,c,\alpha,\Delta)^{\otimes n})$
be given, for some integers $r,\Delta\geq0$ and reals $c>0$ and
$\alpha\geq0$. Let $d$ and $\theta$ be positive integers with
\begin{align}
 & \theta\geq2d,\label{eq:tilde-Lambda-d-theta-1}\\
 & \theta\geq\frac{4\e n(1+\ln n)}{c^{2}},\label{eq:tilde-Lambda-theta-1}\\
 & 2^{\lceil\theta/r\rceil+\alpha\lceil\theta/2\rceil}\geq4\left(\frac{8nr}{c}\right)^{d}.
\end{align}
Then there is a probability distribution $\tilde{\Lambda}$ such that
\begin{align}
 & \supp\tilde{\Lambda}\subseteq(\supp\Lambda)|_{<2\theta+n\Delta},\label{eq:tilde-Lambda-support-restated}\\
 & \orth(\Lambda-\tilde{\Lambda})>d.\label{eq:tilde-Lambda-orthog-restated}
\end{align}
\end{cor}

\begin{proof}
We first consider the special case when $\Lambda\in\WeaklyBounded(r,c,\alpha,\Delta)^{\otimes n}.$
Then by definition, $\Lambda(t_{1},\ldots,t_{n})=\Lambda'(t_{1}-a_{1},\ldots,t_{n}-a_{n})$
for some probability distribution $\Lambda'\in\WeaklyBounded(r,c,\alpha)^{\otimes n}$
and integers $a_{1},\ldots,a_{n}\in[0,\Delta]$. Applying Theorem~\ref{thm:weight-reduction}
to $\Lambda'$ yields a function $\tilde{\Lambda}'\colon\NN^{n}\to\Re$
with
\begin{align}
 & \supp\tilde{\Lambda}'\subseteq(\supp\Lambda')|_{<2\theta},\label{eq:support-reduced-restated}\\
 & \orth(\Lambda'-\tilde{\Lambda}')>d,\label{eq:orthog-preserved-restated}\\
 & |\Lambda'-\tilde{\Lambda}'|\leq\Lambda'\qquad\qquad\text{on }\supp\tilde{\Lambda}'.\label{eq:pointwise}
\end{align}
The last property implies in particular that $\tilde{\Lambda}'$ is
a nonnegative function. As a result,~(\ref{eq:orthog-preserved-restated})
and Proposition~\ref{prop:distribution-criterion} guarantee that
$\tilde{\Lambda}'$ is a probability distribution. Now the sought
properties~(\ref{eq:tilde-Lambda-support-restated}) and (\ref{eq:tilde-Lambda-orthog-restated})
follow from~(\ref{eq:support-reduced-restated}) and~(\ref{eq:orthog-preserved-restated}),
respectively, for the probability distribution  $\tilde{\Lambda}(t_{1},\ldots,t_{n})=\tilde{\Lambda}'(t_{1}-a_{1},\ldots,t_{n}-a_{n}).$

In the general case of a convex combination $\Lambda=\lambda_{1}\Lambda_{1}+\cdots+\lambda_{k}\Lambda_{k}$
of probability distributions $\Lambda_{1},\ldots,\Lambda_{k}\in\WeaklyBounded(r,c,\alpha,\Delta)^{\otimes n},$
one uses the technique of the previous paragraph to transform $\Lambda_{1},\ldots,\Lambda_{k}$
individually into corresponding probability distributions $\tilde{\Lambda}_{1},\ldots,\tilde{\Lambda}_{k}$,
and takes $\tilde{\Lambda}=\lambda_{1}\tilde{\Lambda}_{1}+\cdots+\lambda_{k}\tilde{\Lambda}_{k}$.
\end{proof}

\subsection{\label{subsec:A-bounded-dual}A bounded dual polynomial for MP}

We now turn to the construction of a gadget for our amplification
theorem. Let $\StronglyBounded(r,c,\alpha)$ denote the family of
probability distributions $\lambda$ on $\NN$ such that 
\[
\supp\lambda=\{0,1,2,\ldots,r'\}
\]
for some nonnegative integer $r'\leq r,$ and moreover 
\begin{align*}
 & \frac{c}{(t+1)^{2}\,2^{\alpha t}}\leq\lambda(t)\leq\frac{1}{c(t+1)^{2}\,2^{\alpha t}}, &  & t\in\supp\lambda.
\end{align*}
In this family, a distribution's weight at any given point is prescribed
up to the multiplicative constant $c$, in contrast to the exponentially
large range allowed in the definition of $\WeaklyBounded(r,c,\alpha).$
For all parameter settings, we have
\begin{equation}
\mathfrak{\StronglyBounded}(r,c,\alpha)\subseteq\WeaklyBounded(r,c,\alpha).\label{eq:strong-weak-bounded}
\end{equation}
Indeed, the containment holds trivially for $c\leq1,$ and remains
valid for $c>1$ because the left-hand side and right-hand side are
both empty in that case. As before, it will be helpful to have shorthand
notation for \emph{translates} of distributions in $\WeaklyBounded(r,c,\alpha)$:
we define $\StronglyBounded(r,c,\alpha,\Delta)$ for $\Delta\geq0$
to be the family of probability distributions $\lambda$ on $\NN$
such that $\lambda(t)=\lambda'(t-a)$ for some $\lambda'\in\StronglyBounded(r,c,\alpha)$
and $a\in[0,\Delta]$. 

As our next step toward analyzing the threshold degree of $\classAC^{0},$
we will construct a dual object that witnesses the high threshold
degree of $\MP_{m,r}^{*}$ and possesses additional metric properties
in the sense of $\StronglyBounded$. To simplify the exposition, we
start with an auxiliary construction.
\begin{lem}
\label{lem:dual-OR-distributions}Let $0<\epsilon<1$ be given. Then
for some constants $c_{1},c_{2}\in(0,1)$ and all integers $R\geq r\geq1,$
there are $($explicitly given$)$ probability distributions $\lambda_{0},\lambda_{1},\lambda_{2}$
such that:
\begin{align}
 & \supp\lambda_{0}=\{0\},\label{eq:OR-distributions0-supp}\\
 & \supp\lambda_{i}=\{1,2,\ldots,R\}, &  & i=1,2,\label{eq:OR-distributions-i-supp}\\
 & \lambda_{i}\in\StronglyBounded\left(R,c_{1},\frac{c_{2}}{\sqrt{r}},1\right), &  & i=0,1,2,\label{eq:OR-distributions-i-pointwise}\\
 & \orth((1-\epsilon)\lambda_{0}+\epsilon\lambda_{2}-\lambda_{1})\geq c_{1}\sqrt{r}.\label{eq:OR-distributions-orth}
\end{align}
 
\end{lem}

\noindent Our analysis of the threshold degree of $\classAC^{0}$
only uses the special case $R=r$ of Lemma~\ref{lem:dual-OR-distributions}.
The more general formulation with $R\geq r$ will be needed much later,
in the analysis of the sign-rank of $\classAC^{0}.$ 
\begin{proof}
Theorem~\ref{thm:dual-OR} constructs a function $\psi\colon\{0,1,2,\ldots,R\}\to\Re$
such that 
\begin{align}
 & \psi(0)>\frac{1-\frac{\epsilon}{2}}{2},\label{eq:psi-at-0-restated}\\
 & \|\psi\|_{1}=1,\label{eq:psi-bounded-restated}\\
 & \orth\psi\geq c'\sqrt{r},\label{eq:psi-orthog-restated}\\
 & |\psi(t)|\in\left[\frac{c'}{(t+1)^{2}\,2^{c''t/\sqrt{r}}},\;\frac{1}{c'(t+1)^{2}\,2^{c''t/\sqrt{r}}}\right], &  & t=0,1,\ldots,R,\label{eq:psi-metric-restated}
\end{align}
for some constants $c',c''\in(0,1)$ independent of $R,r$. Property~(\ref{eq:psi-bounded-restated})
makes it possible to view $|\psi|$ as a probability distribution
on $\{0,1,2,\ldots,R\}.$ Let $\mu_{0},\mu_{1},\mu_{2}$ be the probability
distributions induced by $|\psi|$ on $\{0\},\{t\ne0:\psi(t)<0\},$
and $\{t\ne0:\psi(t)>0\},$ respectively. It is clear from~(\ref{eq:psi-at-0-restated})
that the negative part of $\psi$ is a multiple of $\mu_{1},$ whereas
the positive part of $\psi$ is a nonnegative linear combination of
$\mu_{0}$ and $\mu_{2}.$ Moreover, it follows from $\langle\psi,1\rangle=0$
and $\|\psi\|_{1}=1$ that the positive and negative parts of $\psi$
both have $\ell_{1}$-norm $1/2.$ Summarizing,
\begin{equation}
\psi=\frac{1-\delta}{2}\mu_{0}-\frac{1}{2}\mu_{1}+\frac{\delta}{2}\mu_{2}\label{eq:bounded-dual-psi}
\end{equation}
for some $0\leq\delta\leq1.$ In view of~(\ref{eq:psi-at-0-restated}),
we infer the more precise bound 
\begin{equation}
0\leq\delta<\frac{\epsilon}{2}.\label{eq:eps-range}
\end{equation}

We define
\begin{align}
\lambda_{0} & =\mu_{0},\label{eq:lambda0-def}\\
\lambda_{1} & =\frac{1-\epsilon\delta}{1-\delta^{2}}\mu_{1}+\delta\cdot\frac{\epsilon-\delta}{1-\delta^{2}}\mu_{2},\\
\lambda_{2} & =\frac{\epsilon-\delta}{\epsilon(1-\delta^{2})}\mu_{1}+\delta\cdot\frac{1-\epsilon\delta}{\epsilon(1-\delta^{2})}\mu_{2}.\label{eq:lambda2-def}
\end{align}
It follows from $0\leq\delta\leq\epsilon$ that $\lambda_{1}$ and
$\lambda_{2}$ are convex combinations of $\mu_{1}$ and $\mu_{2}$
and are therefore probability distributions with support
\begin{align}
\supp\lambda_{i} & \subseteq\{1,2,\ldots,R\}, &  & i=1,2.\label{eq:mu-i-bar-support}
\end{align}
Recall from~(\ref{eq:bounded-dual-psi}) that $|\psi|=\frac{1}{2}\mu_{1}+\frac{\delta}{2}\mu_{2}$
on $\{1,2,\ldots,R\}$. Comparing the coefficients in $|\psi|=\frac{1}{2}\mu_{1}+\frac{\delta}{2}\mu_{2}$
with the corresponding coefficients in the defining equations for
$\lambda_{1}$ and $\lambda_{2}$, where $0\leq\delta\leq\epsilon/2$
by~(\ref{eq:eps-range}), we conclude that $\lambda_{1},\lambda_{2}\in[c'''|\psi|,|\psi|/c''']$
on $\{1,2,\ldots,R\}$ for some constant $c'''=c'''(\epsilon)\in(0,1).$
In view of~(\ref{eq:psi-metric-restated}), we arrive at
\begin{multline}
\lambda_{i}(t)\in\left[\frac{c'c'''}{(t+1)^{2}\,2^{c''t/\sqrt{r}}},\;\frac{1}{c'c'''(t+1)^{2}\,2^{c''t/\sqrt{r}}}\right],\\
i=1,2;\quad t=1,2,\ldots,R.\qquad\label{eq:lambda-i-pointwise-in-proof}
\end{multline}
Continuing,
\begin{align}
\orth((1-\epsilon)\lambda_{0}+\epsilon\lambda_{2}-\lambda_{1}) & =\orth\left(2\cdot\frac{1-\epsilon}{1-\delta}\left(\frac{1-\delta}{2}\mu_{0}-\frac{1}{2}\mu_{1}+\frac{\delta}{2}\mu_{2}\right)\right)\nonumber \\
 & =\orth\left(2\cdot\frac{1-\epsilon}{1-\delta}\,\psi\right)\nonumber \\
 & \geq c'\sqrt{r},\label{eq:lambdas-orth-inproof-gadget}
\end{align}
where the first step follows from the defining equations~(\ref{eq:lambda0-def})\textendash (\ref{eq:lambda2-def}),
the second step uses~(\ref{eq:bounded-dual-psi}), and the final
step is a restatement of~(\ref{eq:psi-orthog-restated}). 

We are now in a position to verify the claimed properties of $\lambda_{0},\lambda_{1},\lambda_{2}$
in the theorem statement. Property~(\ref{eq:OR-distributions0-supp})
follows from~(\ref{eq:lambda0-def}), whereas property~(\ref{eq:OR-distributions-i-supp})
is immediate from~(\ref{eq:mu-i-bar-support}) and~(\ref{eq:lambda-i-pointwise-in-proof}).
The remaining properties~(\ref{eq:OR-distributions-i-pointwise})
and~(\ref{eq:OR-distributions-orth}) for $c_{2}=c''$ and a small
enough constant $c_{1}>0$ now follow from~(\ref{eq:lambda-i-pointwise-in-proof})
and~(\ref{eq:lambdas-orth-inproof-gadget}), respectively.
\end{proof}
We are now in a position to construct our desired dual polynomial
for the Minsky\textendash Papert function.
\begin{thm}
\label{thm:dual-MP}For some absolute constants $c_{1},c_{2}\in(0,1)$
and all positive integers $m$ and $r,$ there are probability distributions
$\Lambda_{0},\Lambda_{1}$ such that
\begin{align}
 & \Lambda_{i}\in\conv\,\left(\StronglyBounded\!\left(r,c_{1},\frac{c_{2}}{\sqrt{r}},1\right)^{\otimes m}\right), &  & i=0,1,\label{eq:Lambda-conv}\\
 & \supp\Lambda_{i}\subseteq(\MP_{m,r}^{*})^{-1}(i), &  & i=0,1,\label{eq:Lambda-support}\\
 & \orth(\Lambda_{1}-\Lambda_{0})\geq\min\{m,c_{1}\sqrt{r}\}.\label{eq:Lambda-orthog}
\end{align}
\end{thm}

\noindent The last two properties in the theorem statement are equivalent,
in the sense of linear programming duality, to the lower bound $\degthr(\MP_{m,r}^{*})\geq\min\{m,c_{1}\sqrt{r}\}$
and can be recovered in a black-box manner from many previous papers,
e.g.,~\cite{minsky88perceptrons,sherstov07ac-majmaj,sherstov14sign-deg-ac0}.
The key additional property that we prove is~(\ref{eq:Lambda-conv}),
which is where the newly established Lemma~\ref{lem:dual-OR-distributions}
plays an essential role. 
\begin{proof}[Proof of Theorem~\emph{\ref{thm:dual-MP}}.]
 Take $\epsilon=1/2$ and $R=r$ in Lemma~\ref{lem:dual-OR-distributions},
and let $\lambda_{0},\lambda_{1},\lambda_{2}$ be the resulting probability
distributions. Let
\begin{align*}
\Lambda_{0} & =\Exp_{\substack{S\subseteq\{1,2,\ldots,m\}\\
\text{\ensuremath{|S|} odd}
}
}\;\lambda_{0}^{\otimes S}\cdot\lambda_{2}^{\otimes\overline{S}},\\
\Lambda_{1} & =\lambda_{1}^{\otimes m}.
\end{align*}
Then (\ref{eq:Lambda-conv}) is immediate from~(\ref{eq:OR-distributions-i-pointwise}),
whereas~(\ref{eq:Lambda-support}) follows from~(\ref{eq:OR-distributions0-supp})
and~(\ref{eq:OR-distributions-i-supp}). To verify the remaining
property~(\ref{eq:Lambda-orthog}), rewrite
\begin{align*}
\Lambda_{0} & =2^{-m+1}\sum_{\substack{S\subseteq\{1,2,\ldots,m\}\\
\text{\ensuremath{|S|} odd}
}
}\lambda_{0}^{\otimes S}\cdot\lambda_{2}^{\,\otimes\overline{S}}\\
 & =\left(\frac{1}{2}\lambda_{0}+\frac{1}{2}\lambda_{2}\right)^{\otimes m}-\left(-\frac{1}{2}\lambda_{0}+\frac{1}{2}\lambda_{2}\right)^{\otimes m}.
\end{align*}
Observe that
\begin{align}
\orth(\lambda_{i}-\lambda_{j}) & \geq1, &  & i,j=0,1,2,\label{eq:orth-lambdai-lambdaj}
\end{align}
which can be seen from $\langle\lambda_{i}-\lambda_{j},1\rangle=\langle\lambda_{i},1\rangle-\langle\lambda_{j},1\rangle=1-1=0.$
Now
\begin{align*}
\!\!\!\!\!\! & \orth(\Lambda_{1}-\Lambda_{0})\\
 & \quad=\orth\left(\lambda_{1}^{\otimes m}-\left(\frac{1}{2}\lambda_{0}+\frac{1}{2}\lambda_{2}\right)^{\otimes m}+\left(-\frac{1}{2}\lambda_{0}+\frac{1}{2}\lambda_{2}\right)^{\otimes m}\right)\\
 & \quad\geq\min\left\{ \orth\left(\lambda_{1}^{\otimes m}-\left(\frac{1}{2}\lambda_{0}+\frac{1}{2}\lambda_{2}\right)^{\otimes m}\right),\orth\left(\left(-\frac{1}{2}\lambda_{0}+\frac{1}{2}\lambda_{2}\right)^{\otimes m}\right)\right\} \\
 & \quad\geq\min\left\{ \orth\left(\lambda_{1}-\frac{1}{2}\lambda_{0}-\frac{1}{2}\lambda_{2}\right),\orth\left(\left(-\frac{1}{2}\lambda_{0}+\frac{1}{2}\lambda_{2}\right)^{\otimes m}\right)\right\} \\
 & \quad=\min\left\{ \orth\left(\lambda_{1}-\frac{1}{2}\lambda_{0}-\frac{1}{2}\lambda_{2}\right),m\orth\left(-\frac{1}{2}\lambda_{0}+\frac{1}{2}\lambda_{2}\right)\right\} \\
 & \quad\geq\min\left\{ \orth\left(\lambda_{1}-\frac{1}{2}\lambda_{0}-\frac{1}{2}\lambda_{2}\right),m\right\} \\
 & \quad\geq\min\{c_{1}\sqrt{r},m\},
\end{align*}
where the last five steps are valid by Proposition~\ref{prop:orth}\ref{item:orth-sum},
Proposition~\ref{prop:orth}\ref{item:orth-difference-of-tensors},
Proposition~\ref{prop:orth}\ref{item:orth-tensor}, equation~(\ref{eq:orth-lambdai-lambdaj}),
and equation~(\ref{eq:OR-distributions-orth}), respectively.
\end{proof}

\subsection{\label{subsec:Hardness-amplification-for}Hardness amplification
for threshold degree and beyond}

We now present a black-box transformation that takes any given circuit
in $n$ variables with threshold degree $n^{1-\epsilon}$ into a circuit
with polynomially larger threshold degree, $\Omega(n^{1-\frac{\epsilon}{1+\epsilon}})$.
This hardness amplification procedure increases the circuit size additively
by $n^{O(1)}$ and the circuit depth by $2$, preserving membership
in $\classAC^{0}$. We obtain analogous hardness amplification results
for a host of other approximation-theoretic complexity measures. For
this reason, we adopt the following abstract view of polynomial approximation.
Let $I_{0},I_{1},I_{*}$ be nonempty convex subsets of the real line,
i.e., any kind of nonempty intervals (closed, open, or half-open;
bounded or unbounded). Let $f\colon X\to\{0,1,*\}$ be a (possibly
partial) Boolean function on a finite subset $X$ of Euclidean space.
We define an \emph{$(I_{0},I_{1},I_{*})$-approximant for $f$} to
be any real polynomial that maps $f^{-1}(0),f^{-1}(1),f^{-1}(*)$
into $I_{0},I_{1},I_{*},$ respectively. The \emph{$(I_{0},I_{1},I_{*})$-approximate
degree of $f,$} denoted $\deg_{I_{0},I_{1},I_{*}}(f)$, is the least
degree of an $(I_{0},I_{1},I_{*})$-approximant for $f.$ Threshold
degree corresponds to the special case
\begin{equation}
\degthr=\deg_{(0,\infty),(-\infty,0),(-\infty,\infty)}.\label{eq:degthr-iii}
\end{equation}
Other notable cases include \emph{$\epsilon$-approximate degree}
and \emph{one-sided $\epsilon$-approximate degree}, given by
\begin{align}
\deg_{\epsilon} & =\deg_{[-\epsilon,\epsilon],[1-\epsilon,1+\epsilon],[-\epsilon,1+\epsilon]},\label{eq:degeps-iii}\\
\deg_{\epsilon}^{+} & =\deg_{[-\epsilon,\epsilon],[1-\epsilon,\infty),(-\infty,\infty)},\label{eq:degeps-onesided-iii}
\end{align}
respectively. Our hardness amplification result applies to $(I_{0},I_{1},I_{*})$-approximate
degree for any nonempty convex $I_{0},I_{1},I_{*}\subseteq\Re$, with
threshold degree being a special case. The centerpiece of our argument
is the following lemma.
\begin{lem}
\label{lem:Booleanize}Let $c_{1},c_{2}>0$ be the absolute constants
from Theorem~\emph{\ref{thm:dual-MP}}. Let $n,m,r,d,\theta$ be
positive integers such that
\begin{align}
\theta & \geq2d,\label{eq:tilde-Lambda-d-theta-2}\\
\theta & \geq\frac{4\e nm(1+\ln(nm))}{c_{1}^{2}},\label{eq:tilde-Lambda-theta-2}\\
\theta & \geq\frac{2\sqrt{r}}{c_{2}}\left(d\log\left(\frac{8nmr}{c_{1}}\right)+2\right).\label{eq:tilde-Lambda-theta-new}
\end{align}
Then for each $z\in\zoon,$ there is a probability distribution $\tilde{\Lambda}_{z}$
on $\NN^{nm}$ such that:
\begin{enumerate}[topsep=2mm]
\item the support of $\tilde{\Lambda}_{z}$ is contained in $(\prod_{i=1}^{n}(\MP_{m,r}^{*})^{-1}(z_{i}))|_{<2\theta+nm};$
\item for every polynomial $p\colon\Re^{nm}\to\Re$ of degree at most $d,$
the mapping $z\mapsto\Exp_{\tilde{\Lambda}_{z}}p$ is a polynomial
on $\zoon$ of degree at most $\frac{1}{\min\{m,c_{1}\sqrt{r}\}}\cdot\deg p.$
\end{enumerate}
\end{lem}

\begin{proof}
Theorem~\ref{thm:dual-MP} constructs probability distributions $\Lambda_{0}$
and $\Lambda_{1}$ such that
\begin{align}
 & \Lambda_{i}\in\conv\,\left(\StronglyBounded\!\left(r,c_{1},\frac{c_{2}}{\sqrt{r}},1\right)^{\otimes m}\right), &  & i=0,1,\label{eq:Lambda-conv-1}\\
 & \supp\Lambda_{i}\subseteq(\MP_{m,r}^{*})^{-1}(i), &  & i=0,1,\label{eq:Lambda-support-1}\\
 & \orth(\Lambda_{1}-\Lambda_{0})\geq\min\{m,c_{1}\sqrt{r}\}.\label{eq:Lambda-orthog-1}
\end{align}
As a result, the probability distributions $\Lambda_{z}=\bigotimes_{i=1}^{n}\Lambda_{z_{i}}$
for $z\in\zoon$ obey
\begin{align}
\Lambda_{z} & \in\left(\conv\left(\StronglyBounded\!\left(r,c_{1},\frac{c_{2}}{\sqrt{r}},1\right)^{\otimes m}\right)\right)^{\otimes n}\nonumber \\
 & \subseteq\conv\left(\StronglyBounded\!\left(r,c_{1},\frac{c_{2}}{\sqrt{r}},1\right)^{\otimes nm}\right)\nonumber \\
 & \subseteq\conv\left(\WeaklyBounded\!\left(r,c_{1},\frac{c_{2}}{\sqrt{r}},1\right)^{\otimes nm}\right),\label{eq:Lambda-z-bounded}
\end{align}
where the last two steps are valid by~(\ref{eq:tensor-conv-conv-tensor})
and~(\ref{eq:strong-weak-bounded}), respectively. By~(\ref{eq:tilde-Lambda-d-theta-2})\textendash (\ref{eq:tilde-Lambda-theta-new}),
(\ref{eq:Lambda-z-bounded}), and Corollary~\ref{cor:weight-reduction},
there are probability distributions $\tilde{\Lambda}_{z}$ for $z\in\zoon$
such that
\begin{align}
 & \supp\tilde{\Lambda}_{z}\subseteq(\supp\Lambda_{z})|_{<2\theta+nm},\label{eq:support-leadup}\\
 & \orth(\Lambda_{z}-\tilde{\Lambda}_{z})>d.\label{eq:orth-leadup}
\end{align}
We proceed to verify the properties required of $\tilde{\Lambda}_{z}$.
For~(i), it follows from~(\ref{eq:Lambda-support-1}) and~(\ref{eq:support-leadup})
that each $\tilde{\Lambda}_{z}$ has support contained in $(\prod_{i=1}^{n}(\MP_{m,r}^{*})^{-1}(z_{i}))|_{<2\theta+nm}.$
For~(ii), let $p$ be any polynomial of degree at most $d.$ Then~(\ref{eq:orth-leadup})
forces $\Exp_{\tilde{\Lambda}_{z}}p=\Exp_{\Lambda_{z}}p$, where the
right-hand side is by~(\ref{eq:Lambda-orthog-1}) and Proposition~\ref{prop:expect-out}
a polynomial in $z\in\zoon$ of degree at most $\deg p/\orth(\Lambda_{1}-\Lambda_{0})\leq\deg p/\min\{m,c_{1}\sqrt{r}\}.$
\end{proof}
At its core, a hardness amplification result is a lower bound on the
complexity of a composed function in terms of the complexities of
its constituent parts. We now prove such a composition theorem for
$(I_{0},I_{1},I_{*})$-approximate degree.
\begin{thm}
\label{thm:degthr-composition}There is an absolute constant $0<c<1$
such that 
\begin{align*}
\deg_{I_{0},I_{1},I_{*}}((f\circ\MP_{m}^{*})|_{\leq\theta}) & \geq\min\left\{ cm\deg_{I_{0},I_{1},I_{*}}(f),\frac{c\theta}{m\log(n+m)}-n\right\} ,\\
\deg_{I_{0},I_{1},I_{*}}((f\circ\neg\MP_{m}^{*})|_{\leq\theta}) & \geq\min\left\{ cm\deg_{I_{0},I_{1},I_{*}}(f),\frac{c\theta}{m\log(n+m)}-n\right\} 
\end{align*}
for all positive integers $n,m,\theta,$ all functions $f\colon\zoon\to\{0,1,*\},$
and all nonempty convex sets $I_{0},I_{1},I_{*}\subseteq\Re.$
\end{thm}

\begin{proof}
Negating a function's inputs has no effect on the $(I_{0},I_{1},I_{*})$-approximate
degree, so that $f(x_{1},x_{2},\ldots,x_{n})$ and $f(\neg x_{1},\neg x_{2},\ldots,\neg x_{n})$
both have $(I_{0},I_{1},I_{*})$-approximate degree $\deg_{I_{0},I_{1},I_{*}}(f).$
Therefore, it suffices to prove the lower bound on $\deg_{I_{0},I_{1},I_{*}}((f\circ\MP_{m}^{*})|_{\leq\theta})$
for all $f$.

Let $c\in(0,1)$ be an absolute constant that is sufficiently small
relative to the constants in Lemma~\ref{lem:Booleanize}. For $\theta\leq\frac{1}{c}\cdot nm\log(n+m),$
the lower bounds in the statement of the theorem are nonpositive and
therefore trivially true. In the complementary case $\theta>\frac{1}{c}\cdot nm\log(n+m),$
Lemma~\ref{lem:Booleanize} applies to the positive integers $n',m',r',d',\theta'$,
where
\begin{align*}
n' & =n,\\
m' & =m,\\
r' & =m^{2},\\
\theta' & =\left\lfloor \frac{\theta-nm}{2}\right\rfloor ,\\
d' & =\left\lfloor \frac{c\theta}{m\log(n+m)}\right\rfloor .
\end{align*}
We thus obtain, for each $z\in\zoon$, a probability distribution
$\tilde{\Lambda}_{z}$ on $\NN^{nm}$ such that:
\begin{enumerate}
\item the support of $\tilde{\Lambda}_{z}$ is contained in $(\prod_{i=1}^{n}(\MP_{m}^{*})^{-1}(z_{i}))|_{\leq\theta};$
\item for every polynomial $p\colon\Re^{nm}\to\Re$ of degree at most $d',$
the mapping $z\mapsto\Exp_{\tilde{\Lambda}_{z}}p$ is a polynomial
on $\zoon$ of degree at most $\frac{1}{cm}\cdot\deg p.$
\end{enumerate}
Now, let $p\colon\Re^{nm}\to\Re$ be an $(I_{0},I_{1},I_{*})$-approximant
for $(f\circ\MP_{m}^{*})|_{\leq\theta}$ of degree at most $d'.$
Consider the mapping $p^{*}\colon z\mapsto\Exp_{\tilde{\Lambda}_{z}}p$,
which we view as a polynomial in $z\in\zoon$. Then~(i) along with
the convexity of $I_{0},I_{1},I_{*}$ ensures that $p^{*}$ is an
$(I_{0},I_{1},I_{*})$-approximant for $f$, whence $\deg p^{*}\geq\deg_{I_{0},I_{1},I_{*}}(f).$
At the same time,~(ii) guarantees that $\deg p^{*}\leq\frac{1}{cm}\cdot\deg p.$
This pair of lower and upper bounds force 
\[
\deg p\geq cm\deg_{I_{0},I_{1},I_{*}}(f).
\]
Since $p$ was chosen arbitrarily from among $(I_{0},I_{1},I_{*})$-approximants
of $(f\circ\MP_{m}^{*})|_{\leq\theta}$ that have degree at most $d',$
we conclude that
\begin{align*}
\deg_{I_{0},I_{1},I_{*}}((f\circ\MP_{m}^{*})|_{\leq\theta}) & \geq\min\{cm\deg_{I_{0},I_{1},I_{*}}(f),\;d'+1\}\\
 & \geq\min\left\{ cm\deg_{I_{0},I_{1},I_{*}}(f),\;\frac{c\theta}{m\log(n+m)}\right\} . &  & \qedhere
\end{align*}
\end{proof}
\noindent The previous composition theorem has the following analogue
for Boolean inputs.
\begin{thm}
\label{thm:degthr-composition-Boolean-input}Let $0<c<1$ be the absolute
constant from Theorem~\emph{\ref{thm:degthr-composition}.} Let $n,m,N$
be positive integers. Then there is an $($explicitly given$)$ transformation
$H\colon\zoo^{N}\to\zoon,$ computable by an AND-OR-AND circuit of
size $(Nnm)^{O(1)}$ with bottom fan-in $O(\log(nm)),$ such that
for all functions $f\colon\zoon\to\{0,1,*\}$ and all nonempty convex
sets $I_{0},I_{1},I_{*}\subseteq\Re,$ 
\begin{align*}
 & \!\!\!\!\!\!\!\!\!\!\!\!\!\!\!\phantom{\neg}\deg_{I_{0},I_{1},I_{*}}(f\circ H)\geq\min\left\{ cm\deg_{I_{0},I_{1},I_{*}}(f),\frac{cN}{50m\log^{2}(n+m)}-n\right\} \log(n+m),\\
 & \!\!\!\!\!\!\!\!\!\!\!\!\!\!\!\deg_{I_{0},I_{1},I_{*}}(f\circ\neg H)\geq\min\left\{ cm\deg_{I_{0},I_{1},I_{*}}(f),\frac{cN}{50m\log^{2}(n+m)}-n\right\} \log(n+m).
\end{align*}
\end{thm}

\noindent For the function $H$ with multibit output, the notation
$\neg H$ above refers to the function obtained by negating \emph{each}
of $H$'s outputs.
\begin{proof}[Proof of Theorem~\emph{\ref{thm:degthr-composition-Boolean-input}}.]
 As in the previous proof, settling the first lower bound for all
$f$ will automatically settle the second lower bound, due to the
invariance of $(I_{0},I_{1},I_{*})$-approximate degree under negation
of the input bits. In what follows, we focus on~$f\circ H$. 

We may assume that $N\geq50mn\log^{2}(n+m)$ since otherwise the lower
bounds in the theorem statement are nonpositive and hence trivially
true. Define
\[
\theta=\left\lceil \frac{N}{50\log(n+m)}\right\rceil .
\]
Theorem~\ref{thm:input-compression} gives a surjection $G\colon\zoo^{6\theta\lceil\log(nm+1)\rceil}\to\NN^{nm}|_{\leq\theta}$
with the following two properties: 
\begin{enumerate}[topsep=3mm]
\item for every coordinate $i=1,2,\ldots,nm,$ the mapping $x\mapsto\OR_{\theta}^{*}(G(x)_{i})$
is computable by an explicit DNF formula of size $(nm\theta)^{O(1)}=N^{O(1)}$
with bottom fan-in $O(\log(nm))$;
\item for any polynomial $p,$ the map $v\mapsto\Exp_{G^{-1}(v)}p$ is a
polynomial on $\NN^{nm}|_{\leq\theta}$ of degree at most $(\deg p)/\lceil\log(nm+1)+1\rceil\leq(\deg p)/\log(n+m)$.
\end{enumerate}
Consider the composition $F=(f\circ\MP_{m,\theta}^{*})\circ G.$ Then
\begin{align*}
F & =(f\circ(\AND_{m}\circ\OR_{\theta}^{*}))\circ G\\
 & =f\circ((\underbrace{\AND_{m}\circ\OR_{\theta}^{*},\ldots,\AND_{m}\circ\OR_{\theta}^{*}}_{n})\circ G),
\end{align*}
which by property~(i) of $G$ means that $F$ is the composition
of $f$ and an AND-OR-AND circuit $H$ on $6\theta\lceil\log(nm+1)\rceil\leq N$
variables  of size $(nmN)^{O(1)}=N^{O(1)}$ with bottom fan-in $O(\log(nm)).$
Hence, the proof will be complete once we show that
\begin{equation}
\deg_{I_{0},I_{1},I_{*}}(F)\geq\min\left\{ cm\deg_{I_{0},I_{1},I_{*}}(f),\frac{cN}{50m\log^{2}(n+m)}-n\right\} \log(n+m).\label{eq:degthr-composition-needed-bound}
\end{equation}

For this, fix an $(I_{0},I_{1},I_{*})$-approximant $p$ for $F$
of degree $\deg_{I_{0},I_{1},I_{*}}(F)$. Consider the polynomial
$p^{*}\colon\NN^{nm}|_{\leq\theta}\to\Re$ given by $p^{*}(v)=\Exp_{G^{-1}(v)}p.$
Since $I_{0},I_{1},I_{*}$ are convex and $p$ is an $(I_{0},I_{1},I_{*})$-approximant
for $F=(f\circ\MP_{m,\theta}^{*})\circ G$, it follows that $p^{*}$
is an $(I_{0},I_{1},I_{*})$-approximant for $(f\circ\MP_{m,\theta}^{*})|_{\leq\theta}$.
Therefore,
\begin{align*}
\deg p^{*} & \geq\deg_{I_{0},I_{1},I_{*}}((f\circ\MP_{m,\theta}^{*})|_{\leq\theta})\\
 & \geq\deg_{I_{0},I_{1},I_{*}}((f\circ\MP_{m}^{*})|_{\leq\theta})\\
 & \geq\min\left\{ cm\deg_{I_{0},I_{1},I_{*}}(f),\frac{c\theta}{m\log(n+m)}-n\right\} \\
 & \geq\min\left\{ cm\deg_{I_{0},I_{1},I_{*}}(f),\frac{cN}{50m\log^{2}(n+m)}-n\right\} ,
\end{align*}
where the second step is valid because $(f\circ\MP_{m,\theta}^{*})|_{\leq\theta}$
either contains the function $(f\circ\MP_{m}^{*})|_{\leq\theta}=(f\circ\MP_{m,m^{2}}^{*})|_{\leq\theta}$
as a subfunction (case $\theta>m^{2}$), or is equal to it (case $\theta\leq m^{2}$);
and the third step applies Theorem~\ref{thm:degthr-composition}.
However, property~(ii) of $G$ states that
\begin{align*}
\deg p^{*} & \leq\frac{\deg p}{\log(n+m)}\\
 & =\frac{\deg_{I_{0},I_{1},I_{*}}(F)}{\log(n+m)}.
\end{align*}
Comparing these lower and upper bounds on the degree of $p^{*}$ settles~(\ref{eq:degthr-composition-needed-bound}). 
\end{proof}
At last, we illustrate the use of the previous two composition results
to amplify hardness for polynomial approximation.
\begin{thm}[Hardness amplification]
\label{thm:degthr-hardness-amplification}Let $I_{0},I_{1},I_{*}\subseteq\Re$
be any nonempty convex subsets. Let $f\colon\zoon\to\zoo$ be a given
function with 
\[
\deg_{I_{0},I_{1},I_{*}}(f)\geq n^{1-\frac{1}{k}},
\]
for some real number $k\geq1.$ Suppose further that $f$ is computable
by a Boolean circuit of size $s$ and depth $d,$ where $d\geq1.$
Then there is a function $F\colon\zoo^{N}\to\zoo$ on $N=\Theta(n^{1+\frac{1}{k}}\log^{2}n)$
variables with
\[
\deg_{I_{0},I_{1},I_{*}}(F)\geq\Omega\left(\frac{N^{1-\frac{1}{k+1}}}{\log^{1-\frac{2}{k+1}}N}\right).
\]
Moreover, $F$ is computable by a Boolean circuit of size $s+n^{O(1)},$
bottom fan-in $O(\log n),$ depth $d+2$ if the circuit for $f$ is
monotone, and depth $d+3$ otherwise.
\end{thm}

\begin{proof}
Take
\begin{align*}
m & =\lceil n^{1/k}\rceil,\\
N & =\left\lceil \frac{100}{c}\,mn\log^{2}(n+m)\right\rceil ,
\end{align*}
where $0<c<1$ is the absolute constant from Theorem~\ref{thm:degthr-composition}.
Then Theorem~\ref{thm:degthr-composition-Boolean-input} gives an
explicit transformation $H\colon\zoo^{N}\to\zoon,$ computable by
an AND-OR-AND circuit of size $n^{O(1)}$ with bottom fan-in $O(\log n)$,
such that
\begin{align*}
 & \min\{\deg_{I_{0},I_{1},I_{*}}(f\circ H),\deg_{I_{0},I_{1},I_{*}}(f\circ\neg H)\}\\
 & \qquad\qquad\geq\min\left\{ cm\deg_{I_{0},I_{1},I_{*}}(f),\frac{cN}{50m\log^{2}(n+m)}-n\right\} \log(n+m)\\
 & \qquad\qquad\geq cn\log n\\
 & \qquad\qquad=\Theta\left(\frac{N^{1-\frac{1}{k+1}}}{\log^{1-\frac{2}{k+1}}N}\right).
\end{align*}
Now, fix a circuit for $f$ of size $s$ and depth $d\geq1$. Composing
the circuits for $f$ and $H$ results in circuits for $f\circ H$
and $f\circ\neg H$ of size $s+n^{O(1)}$, bottom fan-in $O(\log n)$,
and depth at most $d+3$. Thus, $F$ can be taken to be either of
$f\circ H$ and $f\circ\neg H$.

When the circuit for $f$ is monotone, the depth of $F$ can be reduced
to $d+2$ as follows. After merging like gates if necessary, the circuit
for $f$ can be viewed as composed of $d$ layers of alternating gates
($\wedge$ and $\vee$). The bottom layer of $f$ can therefore be
merged with the top layer of either $H$ or $\neg H,$ resulting in
a circuit of depth at most $(d+3)-1=2$.
\end{proof}

We emphasize that in view of~(\ref{eq:degthr-iii}), the symbol $\deg_{I_{0},I_{1},I_{*}}$
in Theorems~\ref{thm:degthr-composition}\textendash \ref{thm:degthr-hardness-amplification}
can be replaced with the threshold degree symbol $\degthr$. The same
goes for any other special case of $(I_{0},I_{1},I_{*})$-approximate
degree.

\subsection{\label{subsec:Results-for-AC}Threshold degree and discrepancy of
AC\protect\textsuperscript{0}}

We have reached our main result on the sign-representation of constant-depth
circuits. For any $\epsilon>0,$ the next theorem constructs a circuit
family in $\classAC^{0}$ with threshold degree $\Omega(n^{1-\epsilon}).$
The proof amounts to a recursive application of the hardness amplification
procedure of Section~\ref{subsec:Hardness-amplification-for}.
\begin{thm}
\label{thm:degthr-ac0}Let $k\geq1$ be a fixed integer. Then there
is an $($explicitly given$)$ family of functions $\{f_{k,n}\}_{n=1}^{\infty},$
where $f_{k,n}\colon\zoon\to\zoo$ has threshold degree 
\begin{equation}
\degthr(f_{k,n})=\Omega\left(n^{\frac{k-1}{k+1}}\cdot(\log n)^{-\frac{1}{k+1}\lceil\frac{k-2}{2}\rceil\lfloor\frac{k-2}{2}\rfloor}\right)\label{eq:degthr-ac0}
\end{equation}
and is computable by a monotone Boolean circuit of size $n^{O(1)}$
and depth $k.$ In addition, the circuit for $f_{k,n}$ has bottom
fan-in $O(\log n)$ for all $k\ne2.$
\end{thm}

\begin{proof}
The proof is by induction on $k.$ The base cases $k=1$ and $k=2$
correspond to the families
\begin{align*}
f_{1,n}(x) & =x_{1}, &  & n=1,2,3,\ldots,\\
f_{2,n}(x) & =\MP_{\lfloor n^{1/3}\rfloor}, &  & n=1,2,3,\ldots.
\end{align*}
For the former, the threshold degree lower bound~(\ref{eq:degthr-ac0})
is trivial. For the latter, it follows from Theorem~\ref{thm:MP-thrdeg}.

For the inductive step, fix $k\geq3$. Due to the asymptotic nature
of~(\ref{eq:degthr-ac0}), it is enough to construct the functions
in $\{f_{k,n}\}_{n=1}^{\infty}$ for $n$ larger than a certain constant
of our choosing. As a starting point, the inductive hypothesis gives
an explicit family $\{f_{k-2,n}\}_{n=1}^{\infty}$ in which $f_{k-2,n}\colon\zoon\to\zoo$
has threshold degree
\begin{equation}
\degthr(f_{k-2,n})=\Omega\left(n^{\frac{k-3}{k-1}}\cdot(\log n)^{-\frac{1}{k-1}\lceil\frac{k-4}{2}\rceil\lfloor\frac{k-4}{2}\rfloor}\right)\label{eq:degthr-inductive}
\end{equation}
and is computable by a monotone Boolean circuit of size $n^{O(1)}$
and depth $k-2.$ We view the circuit for $f_{k-2,n}$ as composed
of $k-2$ layers of alternating gates, where without loss of generality
the bottom layer consists of AND gates. This last property can be
forced by using $\neg f_{k-2,n}(\neg x_{1},\neg x_{2},\ldots,\neg x_{n})$
instead of $f_{k-2,n}(x_{1},x_{2},\ldots,x_{n})$, which interchanges
the circuit's AND and OR gates without affecting the threshold degree,
circuit depth, or circuit size.

Now, let $c>0$ be the absolute constant from Theorem~\ref{thm:degthr-composition}.
For every $N$ larger than a certain constant, we apply Theorem~\ref{thm:degthr-composition-Boolean-input}
with 
\begin{align}
n & =\left\lceil N^{\frac{k-1}{k+1}}(\log N)^{-\frac{1}{k+1}\lceil\frac{k-4}{2}\rceil\lfloor\frac{k-4}{2}\rfloor-\frac{2(k-1)}{k+1}}\cdot\frac{c}{100}\right\rceil ,\label{eq:param-setting-n}\\
m & =\left\lceil N^{\frac{2}{k+1}}(\log N)^{\frac{1}{k+1}\lceil\frac{k-4}{2}\rceil\lfloor\frac{k-4}{2}\rfloor-\frac{4}{k+1}}\right\rceil ,\label{eq:param-setting-m}\\
f & =f_{k-2,n},\\
I_{0} & =(0,\infty),\\
I_{1} & =(-\infty,0),\\
I_{*} & =(-\infty,\infty)
\end{align}
to obtain a function $H_{N}\colon\zoo^{N}\to\zoo^{n}$ such that the
composition $F_{N}=f_{k-2,n}\circ H_{N}$ has threshold degree
\begin{align}
\degthr(F_{N}) & \geq\min\left\{ cm\degthr(f_{k-2,n}),\frac{cN}{50m\log^{2}(n+m)}-n\right\} \log(n+m)\nonumber \\
 & =\Theta\left(N^{\frac{k-1}{k+1}}\,(\log N)^{-\frac{1}{k+1}\lceil\frac{k-4}{2}\rceil\lfloor\frac{k-4}{2}\rfloor-\frac{k-3}{k+1}}\right)\nonumber \\
 & =\Theta\left(N^{\frac{k-1}{k+1}}\,(\log N)^{-\frac{1}{k+1}\lceil\frac{k-2}{2}\rceil\lfloor\frac{k-2}{2}\rfloor}\right),\label{eq:F_N}
\end{align}
where the second step uses~(\ref{eq:degthr-inductive})\textendash (\ref{eq:param-setting-m}).
Moreover, Theorem~\ref{thm:degthr-composition-Boolean-input} ensures
that $H_{N}$ is computable by an AND-OR-AND circuit of polynomial
size and bottom fan-in $O(\log N)$. The bottom layer of $f_{k-2,n}$
consists of AND gates, which can be merged with the top layer of $H_{N}$
to produce a circuit for $F_{N}=f_{k-2,n}\circ H_{N}$ of depth~$(k-2)+3-1=k.$

We have thus constructed, for some constant $N_{0}$, a family of
functions $\{F_{N}\}_{N=N_{0}}^{\infty}$ in which each $F_{N}\colon\zoo^{N}\to\zoo$
has threshold degree~(\ref{eq:F_N}) and is computable by a Boolean
circuit of polynomial size, depth $k,$ and bottom fan-in $O(\log N).$
Now, take the circuit for $F_{N}$ and replace the negated inputs
in it with $N$ new, unnegated inputs. The resulting \emph{monotone}
circuit on $2N$ variables clearly has threshold degree at least that
of $F_{N}$. This completes the inductive step.
\end{proof}

\noindent Theorem~\ref{thm:degthr-ac0} settles Theorem~\ref{thm:MAIN-degthr-ac0}
from the introduction. Using the pattern matrix method, we now ``lift''
this result to communication complexity.
\begin{thm}
\label{thm:disc-ac0}Let $k\geq3$ be a fixed integer. Let $\ell\colon\NN\to\{2,3,4,\ldots\}$
be given. Then there is an $($explicitly given$)$ family $\{F_{n}\}_{n=1}^{\infty},$
where $F_{n}\colon(\zoon)^{\ell(n)}\to\zoo$ is an $\ell(n)$-party
communication problem with discrepancy
\begin{equation}
\disc(F_{n})\leq2\exp\left(-\Omega\left(\frac{\left\lceil \frac{n}{4^{\ell(n)}\ell(n)^{2}}\right\rceil ^{\frac{k-1}{k+1}}}{\left(1+\log\left\lceil \frac{n}{4^{\ell(n)}\ell(n)^{2}}\right\rceil \right)^{\frac{1}{k+1}\lceil\frac{k-2}{2}\rceil\lfloor\frac{k-2}{2}\rfloor}}\right)\right)\label{eq:disc-ac0-Fn-disc}
\end{equation}
and communication complexity
\begin{equation}
\pp(F_{n})=\Omega\left(\frac{\left\lceil \frac{n}{4^{\ell(n)}\ell(n)^{2}}\right\rceil ^{\frac{k-1}{k+1}}}{\left(1+\log\left\lceil \frac{n}{4^{\ell(n)}\ell(n)^{2}}\right\rceil \right)^{\frac{1}{k+1}\lceil\frac{k-2}{2}\rceil\lfloor\frac{k-2}{2}\rfloor}}\right).\label{eq:disc-ac0-Fn-PP}
\end{equation}
Moreover, $F_{n}$ is computable by a Boolean circuit of size $n^{O(1)}$
and depth $k+1$ in which the bottom two layers have fan-in $O(4^{\ell(n)}\ell(n)^{2}\log n)$
and $\ell(n),$ in that order. In particular, if $\ell(n)=O(1),$
then $F_{n}$ is computable by a Boolean circuit of polynomial size,
depth $k,$ and bottom fan-in $O(\log n).$
\end{thm}

\begin{proof}
Theorem~\ref{thm:degthr-ac0} constructs a family of functions $\{f_{n}\}_{n=1}^{\infty}$,
where $f_{n}\colon\zoon\to\zoo$ has threshold degree 
\begin{equation}
\degthr(f_{n})=\Omega\left(n^{\frac{k-1}{k+1}}\cdot(\log n)^{-\frac{1}{k+1}\lceil\frac{k-2}{2}\rceil\lfloor\frac{k-2}{2}\rfloor}\right)\label{eq:degthr-ac0-restated}
\end{equation}
and is computable by a \emph{monotone} Boolean circuit of polynomial
size, depth~$k,$ and bottom fan-in $O(\log n).$ We view the circuit
for $f_{n}$ as composed of $k$ layers of alternating gates, where
without loss of generality the bottom layer consists of AND gates.
This last property can be forced by using $\neg f_{n}(\neg x_{1},\neg x_{2},\ldots,\neg x_{n})$
instead of $f_{n}(x_{1},x_{2},\ldots,x_{n})$, which interchanges
the circuit's AND and OR gates without affecting the threshold degree,
circuit depth, circuit size, or bottom fan-in. 

Now, let $c>0$ be the absolute constant from Theorem~\ref{thm:pm-discrepancy-k-party}.
For any given $n,$ define
\[
F_{n}=\begin{cases}
\AND_{\ell(n)} & \text{if }n\leq2m,\\
f_{\lfloor n/m\rfloor}\circ\NOR_{m}\circ\AND_{\ell(n)} & \text{otherwise},
\end{cases}
\]
where $m=2\lceil c2^{\ell(n)}\ell(n)\rceil^{2}$. Then the discrepancy
bound~(\ref{eq:disc-ac0-Fn-disc}) is trivial for $n\leq2m$, and
follows from~(\ref{eq:degthr-ac0-restated}) and Theorem~\ref{thm:pm-discrepancy-k-party}
for $n>2m$. The lower bound~(\ref{eq:disc-ac0-Fn-PP}) on the communication
complexity of $F_{n}$ with weakly unbounded error is now immediate
by the discrepancy method (Corollary~\ref{cor:dm}).

It remains to examine the circuit complexity of $F_{n}.$ Since $f_{n}$
is computable by a monotone circuit of size $n^{O(1)}$ and depth
$k$, with the bottom layer composed of AND gates of fan-in $O(\log n)$,
it follows that $F_{n}$ is computable by a circuit of size $n^{O(1)}$
and depth at most $k+1$ in which the bottom two levels have fan-in
$O(\log n)\cdot m=O(4^{\ell(n)}\ell(n)^{2}\log n)$ and $\ell(n),$
in that order. This means that for $\ell(n)=O(1),$ the bottom three
levels of $F_{n}$ can be computed by a circuit of polynomial size,
depth~$2,$ and bottom fan-in $O(\log n)$, which in turn gives  a
circuit for $F_{n}$ of polynomial size, depth~$(k+1)-3+2=k,$ and
bottom fan-in~$O(\log n)$.
\end{proof}
\noindent Taking $\ell(n)=2$ in Theorem~\ref{thm:disc-ac0} settles
Theorem~\ref{thm:MAIN-disc-ac0} from the introduction.

\subsection{\label{subsec:surjectivity}Threshold degree of surjectivity}

We close this section with another application of our amplification
theorem, in which we take the outer function $f$ to be the identity
map $f\colon\zoo\to\zoo$ on a single bit. 
\begin{thm}
\label{thm:degthr-surj}For any integer $m\geq1,$
\[
\degthr(\MP_{m}^{*}|_{\leq m^{2}\log m})=\Omega(m).
\]
\end{thm}

\begin{proof}
Let $f\colon\zoo\to\zoo$ be the identity function, so that $\degthr(f)=1.$
Invoking Theorem~\ref{thm:degthr-composition} with $n=1$ and $\theta=\lfloor m^{2}\log m\rfloor$,
one obtains the claimed lower bound. 
\end{proof}
\noindent Theorem~\ref{thm:degthr-surj} has a useful interpretation.
For positive integers $n$ and $r$, the \emph{surjectivity problem}
is the problem of determining whether a given mapping $\{1,2,\ldots,n\}\to\{1,2,\ldots,r\}$
is surjective. This problem is trivial for $r>n,$ and the standard
regime studied in previous work is $r\leq cn$ for some constant $0<c<1.$
The input to the surjectivity problem is represented by a Boolean
matrix $x\in\zoo^{r\times n}$ with precisely one nonzero entry in
every column. More formally, let $e_{1},e_{2},\ldots,e_{r}$ be the
standard basis for $\Re^{r}.$ The surjectivity function $\SURJ_{n,r}\colon\{e_{1},e_{2},\ldots,e_{r}\}^{n}\to\zoo$
is given by
\[
\SURJ_{n,r}(x_{1},x_{2},\ldots,x_{n})=\bigwedge_{j=1}^{r}\bigvee_{i=1}^{n}x_{i,j}.
\]
It is clear that $\SURJ_{n,r}(x_{1},x_{2},\ldots,x_{n})$ is uniquely
determined by the vector sum $x_{1}+x_{2}+\cdots+x_{n}\in\NN^{r}|_{n}.$
It is therefore natural to consider a symmetric counterpart of the
surjectivity function, with domain $\NN^{r}|_{n}$ instead of $\{e_{1},e_{2},\ldots,e_{r}\}^{n}.$
This symmetric version is $(\AND_{r}\circ\OR_{n}^{*})|_{n}=\MP_{r,n}^{*}|_{n}$,
and Proposition~\ref{prop:ambainis-symmetrization} ensures that
\begin{equation}
\degthr(\SURJ_{n,r})=\degthr(\MP_{r,n}^{*}|_{n}).\label{eq:surj-symm-nonsymm}
\end{equation}
The surjectivity problem has seen much work recently~\cite{beame-machmouchi12quantum-query-ac0,sherstov17algopoly,BKT17poly-strikes-back,BT18ac0-large-error}.
In particular, Bun and Thaler~\cite{BT18ac0-large-error} have obtained
an essentially tight lower bound of $\tilde{\Omega}(\min\{r,\sqrt{n}\})$
on the threshold degree of $\SURJ_{n,r}$ in the standard regime $r\leq(1-\Omega(1))n.$
As a corollary to Theorem~\ref{thm:degthr-surj}, we give a new proof
of Bun and Thaler's result, sharpening their bound by a polylogarithmic
factor.
\begin{cor}
For any integers $n>r>1,$
\begin{equation}
\degthr(\SURJ_{n,r})\geq\Omega\left(\min\left\{ r,\sqrt{\frac{n-r}{1+\log(n-r)}}\right\} \right).\label{eq:surj-n-r-degthr}
\end{equation}
\end{cor}

\begin{proof}
Define
\begin{equation}
r'=\min\left\{ r-1,\left\lfloor \sqrt{\frac{n-r}{1+\log(n-r)}}\right\rfloor \right\} .\label{eq:r-prime-defined}
\end{equation}
We may assume that $r'\geq1$ since~(\ref{eq:surj-n-r-degthr}) holds
trivially otherwise. The identity
\begin{multline*}
\MP_{r',n}^{*}(x_{1},x_{2},\ldots,x_{r'})\\
=\MP_{r,n}^{*}\left(x_{1},x_{2},\ldots,x_{r'},\underbrace{1,1,\ldots,1}_{r-r'-1},1+n-(r-r')-\sum_{i=1}^{r'}x_{i}\right)
\end{multline*}
holds for all $(x_{1},x_{2},\ldots,x_{r'})\in\NN^{r'}|_{\leq n-(r-r')},$
whence
\begin{equation}
\degthr(\MP_{r',n}^{*}|_{\leq n-(r-r')})\leq\degthr(\MP_{r,n}^{*}|_{n}).\label{eq:MP-restrictions}
\end{equation}
Now
\begin{align*}
\degthr(\SURJ_{n,r}) & =\degthr(\MP_{r,n}^{*}|_{n})\\
 & \geq\degthr(\MP_{r',n}^{*}|_{\leq n-(r-r')})\\
 & \geq\degthr(\MP_{r',r'^{2}}^{*}|_{\leq r'^{2}\log r'})\\
 & \geq\Omega(r'),
\end{align*}
where the four steps use~(\ref{eq:surj-symm-nonsymm}), (\ref{eq:MP-restrictions}),
(\ref{eq:r-prime-defined}), and Theorem~\ref{thm:degthr-surj},
respectively.
\end{proof}

\section{\label{sec:Sign-rank-of-AC0}The sign-rank of AC\protect\textsuperscript{0}}

We now turn to the second main result of this paper, a near-optimal
lower bound on the sign-rank of constant-depth circuits. To start
with, we show that our smoothing technique from Theorem~\ref{thm:dual-MP}
already gives an exponential lower bound on the sign-rank of $\classAC^{0}$.
Specifically, we prove in Section~\ref{subsec:RS-simplified} that
the Minsky\textendash Papert function $\MP_{n^{1/3}}$ has $\exp(-O(n^{1/3}))$-smooth
threshold degree $\Omega(n^{1/3}),$ which by Theorem~\ref{thm:thrdeg-to-sign-rank}
immediately implies an $\exp(\Omega(n^{1/3}))$ lower bound on the
sign-rank of an $\classAC^{0}$ circuit of depth~$3$. This result
was originally obtained, with a longer and more demanding proof, by
Razborov and Sherstov~\cite{RS07dc-dnf}.

To obtain the near-optimal lower bound of $\exp(\Omega(n^{1-\epsilon}))$
for every $\epsilon>0$, we use a completely different approach. It
is based on the notion of \emph{local smoothness} and is unrelated
to the threshold degree analysis. In Section~\ref{subsec:Local-smoothness},
we define local smoothness and record basic properties of locally
smooth functions. In Sections~\ref{subsec:Metric-properties-of-locally-smooth}
and~\ref{subsec:Weight-transfer}, we develop techniques for manipulating
locally smooth functions to achieve desired global behavior, without
the manipulations being detectable by low-degree polynomials. To apply
this machinery to constant-depth circuits, we design in Section~\ref{subsec:A-locally-smooth}
a locally smooth dual polynomial for the Minsky\textendash Papert
function. We use this dual object in Section~\ref{subsec:An-amplification-theorem-for-smooth-thrdeg}
to prove an amplification theorem for \emph{smooth} threshold degree.
We apply the amplification theorem iteratively in Section~\ref{subsec:The-smooth-threshold-degree-AC0}
to construct, for any $\epsilon>0,$ a constant-depth circuit with
$\exp(-O(n^{1-\epsilon}))$-smooth threshold degree $\Omega(n^{1-\epsilon}).$
Finally, we present our main result on the sign-rank of $\classAC^{0}$
in Section~\ref{subsec:The-sign-rank-of-AC0}.

In the remainder of this section, we adopt the following additional
notation. For an arbitrary subset $X$ of Euclidean space, we write
$\diam X=\sup_{x,x'\in X}|x-x'|,$ with the convention that $\diam\varnothing=0.$
For a vector $x\in\ZZ^{n}$ and a natural number $d,$ we let $B_{d}(x)=\{v\in\ZZ^{n}:|x-v|\leq d\}$
denote the set of \emph{integer-valued} vectors within distance $d$
of $x$. For all $x,$
\begin{equation}
|B_{d}(x)|=|B_{d}(0)|\leq2^{d}\binom{n+d}{d},\label{eq:ball-around-x}
\end{equation}
where the binomial coefficient corresponds to the number of \emph{nonnegative}
integer vectors of weight at most $d.$ Finally, for vectors $u,v\in\NN^{n},$
we define $\cube(u,v)$ to be the smallest Cartesian product of integer
intervals that contains both $u$ and $v.$ Specifically,
\begin{align*}
\cube(u,v) & =\{w\in\NN^{n}:\min\{u_{i},v_{i}\}\leq w_{i}\leq\max\{u_{i},v_{i}\}\text{ for all }i\}\\
 & =\prod_{i=1}^{n}\{\min\{u_{i},v_{i}\},\min\{u_{i},v_{i}\}+1,\ldots,\max\{u_{i},v_{i}\}\}.
\end{align*}

\subsection{\label{subsec:RS-simplified}A simple lower bound for depth~3}

We start by presenting a new proof of Razborov and Sherstov's exponential
lower bound~\cite{RS07dc-dnf} on the sign-rank of $\classAC^{0}.$
More precisely, we prove the following stronger result that was not
known before. 
\begin{thm}
\label{thm:smooth-MP}There is a constant $0<c<1$ such that for all
positive integers $m$ and $r,$
\[
\degthr(\MP_{m,r},12^{-m-1})\geq\min\{m,c\sqrt{r}\}.
\]
\end{thm}

\noindent Theorem~\ref{thm:smooth-MP} is asymptotically optimal,
and it is the first lower bound on the smooth threshold degree of
the Minsky\textendash Papert function. As we will discuss shortly,
this theorem implies an $\exp(\Omega(n^{1/3}))$ lower bound on the
sign-rank of $\classAC^{0}$. In addition, we will use Theorem~\ref{thm:smooth-MP}
as the base case in the inductive proof of Theorem~\ref{thm:MAIN-sign-rank-3k-plus-1}.
\begin{proof}[Proof of Theorem~\emph{\ref{thm:smooth-MP}}.]
 It is well-known~\cite{nisan-szegedy94degree,paturi92approx,spalek08dual-or}
that for some constant $c>0$ and all $r,$ any real polynomial $p\colon\zoo^{r}\to\Re$
with $\|p-\NOR_{r}\|_{\infty}\leq0.49$ has degree at least $c\sqrt{r}$.
By linear programming duality~\cite[Theorem~2.5]{sherstov14sign-deg-ac0},
this approximation-theoretic fact is equivalent to the existence of
a function $\psi\colon\zoo^{r}\to\Re$ with
\begin{align}
 & \psi(0)>0.49,\label{eq:psi-at-0-restated-1}\\
 & \|\psi\|_{1}=1,\label{eq:psi-bounded-restated-1}\\
 & \orth\psi\geq c\sqrt{r}.\label{eq:psi-orthog-restated-1}
\end{align}

The rest of the proof is a reprise of Section~\ref{subsec:A-bounded-dual}.
To begin with, property~(\ref{eq:psi-bounded-restated-1}) makes
it possible to view $|\psi|$ as a probability distribution on $\zoo^{r}.$
Let $\mu_{0},\mu_{1},\mu_{2}$ be the probability distributions induced
by $|\psi|$ on the sets $\{0^{r}\},\{x\ne0^{r}:\psi(x)<0\},$ and
$\{x\ne0^{r}:\psi(x)>0\},$ respectively. It is clear from~(\ref{eq:psi-at-0-restated-1})
that the negative part of $\psi$ is a multiple of $\mu_{1},$ whereas
the positive part of $\psi$ is a nonnegative linear combination of
$\mu_{0}$ and $\mu_{2}.$ Moreover, it follows from $\langle\psi,1\rangle=0$
and $\|\psi\|_{1}=1$ that the positive and negative parts of $\psi$
both have $\ell_{1}$-norm $1/2.$ Summarizing,
\begin{equation}
\psi=\frac{1-\delta}{2}\mu_{0}-\frac{1}{2}\mu_{1}+\frac{\delta}{2}\mu_{2}\label{eq:bounded-dual-psi-1}
\end{equation}
for some $0\leq\delta\leq1.$ In view of~(\ref{eq:psi-at-0-restated-1}),
we infer the more precise bound 
\begin{equation}
0\leq\delta<\frac{1}{50}.\label{eq:eps-range-1}
\end{equation}

Let $\upsilon$ be the uniform probability distribution on $\zoo^{r}\setminus\{0^{r}\}.$
We define
\begin{align}
\lambda_{0} & =\mu_{0},\label{eq:lambda0-def-1}\\
\lambda_{1} & =\frac{2}{3(1-\delta)}\mu_{1}+\left(1-\frac{2}{3(1-\delta)}\right)\upsilon,\\
\lambda_{2} & =\frac{2\delta}{1-\delta}\mu_{2}+\left(1-\frac{2\delta}{1-\delta}\right)\upsilon.\label{eq:lambda2-def-1}
\end{align}
It is clear from~(\ref{eq:eps-range-1}) that $\lambda_{1}$ and
$\lambda_{2}$ are convex combinations of $\upsilon,\mu_{1},\mu_{2}$
and therefore are probability distributions with support
\begin{align}
\supp\lambda_{i} & \subseteq\zoo^{r}\setminus\{0^{r}\}, &  & i=1,2,\label{eq:mu-i-bar-support-1}
\end{align}
whereas
\begin{equation}
\supp\lambda_{0}=\{0^{r}\}\label{eq:RS-lambda0-support}
\end{equation}
by definition. Moreover,~(\ref{eq:eps-range-1}) implies that
\begin{align}
\lambda_{i} & \geq\frac{1}{4}\upsilon, &  & i=1,2.\label{eq:lambda-i-lower-bound}
\end{align}
The defining equations~(\ref{eq:lambda0-def-1})\textendash (\ref{eq:lambda2-def-1})
further imply that
\[
\frac{2}{3}\lambda_{0}+\frac{1}{3}\lambda_{2}-\lambda_{1}=\frac{4}{3(1-\delta)}\psi,
\]
which along with~(\ref{eq:psi-orthog-restated-1}) gives
\begin{align}
\orth\left(\frac{2}{3}\lambda_{0}+\frac{1}{3}\lambda_{2}-\lambda_{1}\right) & \geq c\sqrt{r}.\label{eq:lambdas-orth-inproof-gadget-1}
\end{align}

With this work behind us, define
\[
\Lambda=\frac{1}{2}\left(\frac{2}{3}\lambda_{0}+\frac{1}{3}\lambda_{2}\right)^{\otimes m}-\frac{1}{2}\left(-\frac{1}{3}\lambda_{0}+\frac{1}{3}\lambda_{2}\right)^{\otimes m}+\frac{1}{2}\lambda_{1}^{\otimes m}.
\]
 Multiplying out the tensor products in the definition of $\Lambda$
and collecting like terms, we obtain
\begin{align}
\Lambda & =\frac{1}{2}\sum_{\substack{S\subseteq\{1,2,\ldots,m\}\\
S\ne\varnothing
}
}\frac{2^{|S|}-(-1)^{|S|}}{3^{m}}\lambda_{0}^{\otimes S}\cdot\lambda_{2}^{\otimes\overline{S}}+\frac{1}{2}\lambda_{1}^{\otimes m}\label{eq:RS-Lambda-expanded}\\
 & \geq\frac{1}{4}\sum_{\substack{S\subseteq\{1,2,\ldots,m\}\\
S\ne\varnothing
}
}\frac{2^{|S|}}{3^{m}}\lambda_{0}^{\otimes S}\cdot\lambda_{2}^{\otimes\overline{S}}+\frac{1}{2}\lambda_{1}^{\otimes m}\nonumber \\
 & \geq\frac{1}{4}\sum_{\substack{S\subseteq\{1,2,\ldots,m\}\\
S\ne\varnothing
}
}\frac{2^{|S|}}{3^{m}}\lambda_{0}^{\otimes S}\cdot\left(\frac{1}{4}\upsilon\right)^{\otimes\overline{S}}+\frac{1}{2}\left(\frac{1}{4}\upsilon\right)^{\otimes m}\nonumber \\
 & \geq\frac{1}{4}\sum_{\substack{S\subseteq\{1,2,\ldots,m\}}
}\frac{2^{|S|}}{3^{m}}\lambda_{0}^{\otimes S}\cdot\left(\frac{1}{4}\upsilon\right)^{\otimes\overline{S}}\nonumber \\
 & =\frac{1}{4}\left(\frac{2}{3}\lambda_{0}+\frac{1}{3}\cdot\frac{1}{4}\upsilon\right)^{\otimes m}\nonumber \\
 & \geq\frac{1}{4}\left(\frac{1}{12\cdot2^{r}}\right)^{m}\1_{(\zoo^{r})^{m}},\label{eq:RS-smoothness}
\end{align}
where the third step uses~(\ref{eq:lambda-i-lower-bound}). In particular,
$\Lambda$ is a nonnegative function. We further calculate
\begin{align}
\langle\Lambda,1\rangle & =\frac{1}{2}\left\langle \frac{2}{3}\lambda_{0}+\frac{1}{3}\lambda_{2},1\right\rangle ^{m}-\frac{1}{2}\left\langle -\frac{1}{3}\lambda_{0}+\frac{1}{3}\lambda_{2},1\right\rangle ^{m}+\frac{1}{2}\langle\lambda_{1},1\rangle^{m}\nonumber \\
 & =\frac{1}{2}\left\langle \frac{2}{3}\lambda_{0}+\frac{1}{3}\lambda_{2},1\right\rangle ^{m}+\frac{1}{2}\langle\lambda_{1},1\rangle^{m}\nonumber \\
 & =\frac{1}{2}+\frac{1}{2}\nonumber \\
 & =1,\label{eq:RS-Lambda-norm}
\end{align}
which makes $\Lambda$ a probability distribution on $(\zoo^{r})^{m}.$

It remains to examine the orthogonal content of $\Lambda\cdot(-1)^{\MP_{m,r}}.$
We have
\begin{align*}
\Lambda\cdot(-1)^{\MP_{m,r}} & =\frac{1}{2}\sum_{\substack{S\subseteq\{1,2,\ldots,m\}\\
S\ne\varnothing
}
}\frac{2^{|S|}-(-1)^{|S|}}{3^{m}}\lambda_{0}^{\otimes S}\cdot\lambda_{2}^{\otimes\overline{S}}\cdot(-1)^{\MP_{m,r}}\\
 & \qquad\qquad\qquad\qquad\qquad\qquad\qquad\qquad+\frac{1}{2}\lambda_{1}^{\otimes m}\cdot(-1)^{\MP_{m,r}}\\
 & =\frac{1}{2}\sum_{\substack{S\subseteq\{1,2,\ldots,m\}\\
S\ne\varnothing
}
}\frac{2^{|S|}-(-1)^{|S|}}{3^{m}}\lambda_{0}^{\otimes S}\cdot\lambda_{2}^{\otimes\overline{S}}-\frac{1}{2}\lambda_{1}^{\otimes m}\\
 & =\frac{1}{2}\left(\frac{2}{3}\lambda_{0}+\frac{1}{3}\lambda_{2}\right)^{\otimes m}-\frac{1}{2}\left(-\frac{1}{3}\lambda_{0}+\frac{1}{3}\lambda_{2}\right)^{\otimes m}-\frac{1}{2}\lambda_{1}^{\otimes m},
\end{align*}
where the first step uses~(\ref{eq:RS-Lambda-expanded}); the second
step uses~(\ref{eq:mu-i-bar-support-1}) and~(\ref{eq:RS-lambda0-support});
and the final equality can be verified by multiplying out the tensor
powers and collecting like terms. Now
\begin{align*}
 & \orth(\Lambda\cdot(-1)^{\MP_{m,r}})\\
 & \qquad\qquad=\min\left\{ \orth\left(\frac{1}{2}\left(\frac{2}{3}\lambda_{0}+\frac{1}{3}\lambda_{2}\right)^{\otimes m}-\frac{1}{2}\lambda_{1}^{\otimes m}\right),\right.\\
 & \qquad\qquad\qquad\qquad\qquad\qquad\left.\orth\left(-\frac{1}{2}\left(-\frac{1}{3}\lambda_{0}+\frac{1}{3}\lambda_{2}\right)^{\otimes m}\right)\right\} \\
 & \qquad\qquad\geq\min\left\{ \orth\left(\frac{2}{3}\lambda_{0}+\frac{1}{3}\lambda_{2}-\lambda_{1}\right),m\orth\left(-\frac{1}{3}\lambda_{0}+\frac{1}{3}\lambda_{2}\right)\right\} \\
 & \qquad\qquad\geq\min\left\{ c\sqrt{r},m\orth\left(-\frac{1}{3}\lambda_{0}+\frac{1}{3}\lambda_{2}\right)\right\} \\
 & \qquad\qquad\geq\min\{c\sqrt{r},m\},
\end{align*}
where the first step applies Proposition~\ref{prop:orth}\ref{item:orth-sum};
the second step applies Proposition~\ref{prop:orth}\ref{item:orth-tensor},~\ref{item:orth-difference-of-tensors};
the third step substitutes the lower bound from~(\ref{eq:lambdas-orth-inproof-gadget-1});
and the last step uses $\langle-\lambda_{0}+\lambda_{2},1\rangle=-\langle\lambda_{0},1\rangle+\langle\lambda_{2},1\rangle=-1+1=0.$
Combining this conclusion with~(\ref{eq:RS-smoothness}) and~(\ref{eq:RS-Lambda-norm})
completes the proof.
\end{proof}
We now ``lift'' the approximation-theoretic result just obtained
to a sign-rank lower bound, reproving a result of Razborov and Sherstov~\cite{RS07dc-dnf}.
\begin{thm}[Razborov and Sherstov]
 \label{thm:RS}Define $F_{n}\colon\zoon\times\zoon\to\zoo$ by 
\[
F_{n}=\AND_{n^{1/3}}\circ\OR_{n^{2/3}}\circ\AND_{2}.
\]
Then
\[
\srank(F_{n})\geq2^{\Omega(n^{1/3})}.
\]
\end{thm}

\begin{proof}
Theorem~\ref{thm:smooth-MP} gives 
\[
\degthr(\AND_{n^{1/3}}\circ\OR_{n^{2/3}},\exp(-c'n^{1/3}))\geq c''n^{1/3}
\]
for some absolute constants $c',c''>0$ and all $n.$ This lower bound
along with Theorem~\ref{thm:thrdeg-to-sign-rank} implies that the
composition
\[
H_{n}=\AND_{n^{1/3}}\circ\OR_{n^{2/3}}\circ\OR_{2\left\lceil \exp\left(\frac{4c'}{c''}\right)\right\rceil }\circ\AND_{2}
\]
has sign-rank $\srank(H_{n})=\exp(\Omega(n^{1/3})).$ This completes
the proof because for some integer constant $c\geq1,$ each $H_{n}$
is a subfunction of $F_{cn}$.
\end{proof}

\subsection{\label{subsec:Local-smoothness}Local smoothness}

The remainder of this paper focuses on our $\exp(\Omega(n^{1-\epsilon}))$
lower bound on the sign-rank of $\classAC^{0}$, whose proof is unrelated
to the work in Section~\ref{sec:Threshold-degree-of-AC0} and Section~\ref{subsec:RS-simplified}.
Central to our approach is an analytic notion that we call \emph{local
smoothness}. Formally, let $\Phi\colon\NN^{n}\to\Re$ be a function
of interest. For a subset $X\subseteq\NN^{n}$ and a real number $K\geq1,$
we say that $\Phi$ is \emph{$K$-smooth on $X$} if
\begin{align*}
|\Phi(x)| & \leq K^{|x-x'|}|\Phi(x')|\text{\quad for all }x,x'\in X.
\end{align*}
Put another way, for any two points of $X$ at distance $d,$ the
corresponding values of $\Phi$ differ in magnitude by a factor of
at most $K^{d}.$ For any set $X,$ we let $\SmoothFunction(K,X)$
denote the family of functions that are $K$-smooth on $X.$ The following
proposition collects basic properties of local smoothness, to which
we refer as the restriction property, scaling property, tensor property,
and conical property.
\begin{prop}
\label{prop:smooth}Let $K\geq1$ be given.
\begin{enumerate}
\item \label{enu:smooth-restriction}If $\Phi\in\SmoothFunction(K,X)$ and
$X'\subseteq X,$ then $\Phi\in\SmoothFunction(K,X')$.
\item \label{enu:smooth-scaling}If $\Phi\in\SmoothFunction(K,X)$ and $a\in\Re,$
then $a\Phi\in\SmoothFunction(K,X).$
\item \label{enu:smooth-tensor}$\SmoothFunction(K,X)\otimes\SmoothFunction(K,Y)\subseteq\SmoothFunction(K,X\times Y).$
\item \label{enu:smooth-conical}If $\Phi,\Psi\in\SmoothFunction(K,X)$
and $\Phi,\Psi$ are nonnegative on $X,$ then $\cone\{\Phi,\Psi\}\subseteq\SmoothFunction(K,X)$.
\end{enumerate}
\end{prop}

\begin{proof}
Properties~\ref{enu:smooth-restriction} and~\ref{enu:smooth-scaling}
are immediate from the definition of $K$-smoothness. For~\ref{enu:smooth-tensor},
fix $\Phi\in\SmoothFunction(K,X)$ and $\Psi\in\SmoothFunction(K,Y).$
Then for all $(x,y),(x',y')\in X\times Y,$ we have
\begin{align*}
|\Phi(x)\Psi(y)| & \leq K^{|x-x'|}|\Phi(x')|\,K^{|y-y'|}|\Psi(y')|\\
 & =K^{|(x,y)-(x',y')|}|\Phi(x')\Psi(y')|,
\end{align*}
where the first step uses the $K$-smoothness of $\Phi$ and $\Psi.$
Finally, for~\ref{enu:smooth-conical}, let $a$ and $b$ be nonnegative
reals. Then
\begin{align*}
|a\Phi(x)+b\Psi(x)| & =a|\Phi(x)|+b|\Psi(x)|\\
 & \leq aK^{|x-x'|}|\Phi(x')|+bK^{|x-x'|}|\Psi(x')|\\
 & =K^{|x-x'|}|a\Phi(x')+b\Psi(x')|
\end{align*}
for all $x,x'\in X,$ where the second step uses the $K$-smoothness
of $\Phi$ and $\Psi.$ 
\end{proof}
We will take a special interest in locally smooth functions that are
probability distributions. For our purposes, it will be sufficient
to consider locally smooth distributions whose support is the Cartesian
product of integer intervals. For an integer $n\geq1$ and a real
number $K\geq1,$ we let $\Smooth(n,K)$ denote the set of probability
distributions $\Lambda$ such that:
\begin{enumerate}
\item $\Lambda$ is supported on $\prod_{i=1}^{n}\{0,1,2,\ldots,r_{i}\}$,
for some $r_{1},r_{2},\ldots,r_{n}\in\NN$;
\item $\Lambda$ is $K$-smooth on its support.
\end{enumerate}
Analogous to the development in Section~\ref{sec:Threshold-degree-of-AC0},
we will need a notation for translates of distributions in $\Smooth(n,K)$.
For $\Delta\geq0,$ we let $\Smooth(n,K,\Delta)$ denote the set of
probability distributions $\Lambda\in\Distribution(\NN^{n})$ such
that $\Lambda(t_{1},\ldots,t_{n})\equiv\Lambda'(t_{1}-a_{1},\ldots,t_{n}-a_{n})$
for some fixed $\Lambda'\in\Smooth(n,K)$ and $a\in\NN^{n}|_{\leq\Delta}$.
As a special case, $\Smooth(n,K,0)=\Smooth(n,K).$ Specializing Proposition~\ref{prop:smooth}\ref{enu:smooth-tensor}
to this context, we obtain:
\begin{prop}
\label{prop:smooth-tensor}For any $n',n'',\Delta',\Delta'',K,$ one
has
\[
\Smooth(n',K,\Delta')\otimes\Smooth(n'',K,\Delta'')\subseteq\Smooth(n'+n'',K,\Delta'+\Delta'').
\]
\end{prop}

\begin{proof}
The only nontrivial property to verify is $K$-smoothness, which follows
from Proposition~\ref{prop:smooth}\ref{enu:smooth-tensor}.
\end{proof}

\subsection{\label{subsec:Metric-properties-of-locally-smooth}Metric properties
of locally smooth distributions}

If $\Lambda$ is a locally smooth distribution on $X=\prod_{i=1}^{n}\{0,1,2,\ldots,r_{i}\},$
then a moment's thought reveals that $\Lambda(x)>0$ at every point
$x\in X.$ In general, local smoothness gives considerable control
over metric properties of $\Lambda,$ making it possible to prove
nontrivial upper and lower bounds on $\Lambda(S)$ for various sets
$S\subseteq X.$ We now record two such results, as regards our work
on the sign-rank on $\classAC^{0}$.
\begin{prop}
\label{prop:remove-cap}Let $\Lambda$ be a probability distribution
on $X=\prod_{i=1}^{n}\{0,1,2,\ldots,r_{i}\}.$ Let $\theta$ and $d$
be nonnegative integers with $\theta\geq d$. If $\Lambda$ is $K$-smooth
on $X|_{\leq\theta},$ then
\[
\Lambda(X|_{\leq\theta})\leq K^{d}\binom{n+d}{d}\Lambda(X|_{\leq\theta-d}).
\]
\end{prop}

\begin{proof}
Consider an arbitrary vector $x\in X|_{\leq\theta}$. By definition,
the components of $x$ are nonnegative integers that sum to at most
$\theta.$ By decreasing the components of $x$ as needed, one can
obtain a vector $x'$ with
\begin{align*}
 & x'\in X|_{\leq\theta-d},\\
 & x'\leq x,\\
 & |x'-x|\leq d.
\end{align*}
In particular, the $K$-smoothness of $\Lambda$ implies that
\[
\Lambda(x)\leq K^{d}\Lambda(x').
\]
Summing on both sides over $x\in X|_{\leq\theta},$ we obtain
\begin{align*}
\Lambda(X|_{\leq\theta}) & \leq K^{d}\Lambda(X|_{\leq\theta-d})\max_{x'\in X|_{\leq\theta-d}}|\{x\in X|_{\leq\theta}:x\geq x'\text{ and }|x-x'|\leq d\}|\\
 & \leq K^{d}\Lambda(X|_{\leq\theta-d})\,\max_{x'\in\NN^{n}}|\{x\in\NN^{n}:x\geq x'\text{ and }|x-x'|\leq d\}|\\
 & =K^{d}\Lambda(X|_{\leq\theta-d})\binom{n+d}{d}.\tag*{\qedhere}
\end{align*}
\end{proof}
\begin{prop}
\label{prop:remove-ball}Let $\Lambda$ be a probability distribution
on $X=\prod_{i=1}^{n}\{0,1,2,\ldots,r_{i}\}.$ Let $\theta$ and $d$
be nonnegative integers with
\begin{equation}
d<\frac{1}{2}\min\left\{ \theta,\sum_{i=1}^{n}r_{i}\right\} .\label{eq:d-theta-M}
\end{equation}
If $\Lambda$ is $K$-smooth on $X|_{\leq\theta},$ then
\[
\Lambda(X|_{\leq\theta})\leq2^{d+1}K^{2d+1}\binom{n+d}{d}\Lambda(X|_{\leq\theta}\setminus B_{d}(u))
\]
for every $u\in X.$
\end{prop}

\begin{proof}
Fix $u\in X$ for the rest of the proof. If $|u|>\theta+d,$ then
$X|_{\leq\theta}\setminus B_{d}(u)=X|_{\leq\theta}$ and the statement
holds trivially. In what follows, we treat the complementary case
$|u|\leq\theta+d.$ Here, the key is to find a vector $u'$ with
\begin{align}
 & |u-u'|=d+1,\label{eq:u-distance}\\
 & u'\in X|_{\leq\theta}.\label{eq:u-good}
\end{align}
The algorithm for finding $u'$ depends on $|u|,$ as follows.
\begin{enumerate}
\item If $|u|>d,$ decrease one or more of the components of $u$ as needed
to obtain a vector $u'$ whose components are nonnegative integers
that sum to exactly $|u|-d-1.$ Then~(\ref{eq:u-distance}) is immediate,
whereas~(\ref{eq:u-good}) follows in view of $|u|\leq\theta+d.$
\item If $|u|\leq d,$ the analysis is more subtle. Recall that $u\in\prod_{i=1}^{n}\{0,1,2,\ldots,r_{i}\}$
and therefore $|(r_{1},\ldots,r_{n})-u|=\sum r_{i}-|u|\geq\sum r_{i}-d>d,$
where the last step uses~(\ref{eq:d-theta-M}). As a result, by increasing
the components of $u$ as necessary, one can obtain a vector $u'\in\prod_{i=1}^{n}\{0,1,2,\ldots,r_{i}\}$
with $|u'|=|u|+d+1.$ Then property~(\ref{eq:u-distance}) is immediate.
Property~(\ref{eq:u-good}) follows from $|u'|=|u|+d+1\leq2d+1<\theta+1,$
where the last step uses~(\ref{eq:d-theta-M}).
\end{enumerate}
Now that $u'$ has been constructed, apply the $K$-smoothness of
$\Lambda$ to conclude that for every $x\in X|_{\leq\theta}\cap B_{d}(u)$,
\begin{align}
\Lambda(x) & \leq K^{|x-u'|}\Lambda(u')\nonumber \\
 & \leq K^{|x-u|+|u-u'|}\Lambda(u')\nonumber \\
 & \leq K^{2d+1}\Lambda(u'),\label{eq:single-point-ineq}
\end{align}
 where the last step uses~(\ref{eq:u-distance}). As a result,
\begin{align}
\Lambda(X|_{\leq\theta}\cap B_{d}(u)) & \leq\big|\,X|_{\le\theta}\,\cap\,B_{d}(u)\big|\,K^{2d+1}\Lambda(u')\nonumber \\
 & \leq|B_{d}(u)|\,K^{2d+1}\Lambda(u')\nonumber \\
 & \leq|B_{d}(u)|\,K^{2d+1}\Lambda(X|_{\leq\theta}\setminus B_{d}(u))\nonumber \\
 & \leq2^{d}\binom{n+d}{d}\,K^{2d+1}\Lambda(X|_{\leq\theta}\setminus B_{d}(u)),\label{eq:ball-excluded-bound}
\end{align}
where the first inequality is the result of summing~(\ref{eq:single-point-ineq})
over $x\in X|_{\leq\theta}\cap B_{d}(u)$; the third step uses~(\ref{eq:u-distance})
and~(\ref{eq:u-good}); and the last step applies~(\ref{eq:ball-around-x}).
To complete the proof, add $\Lambda(X|_{\leq\theta}\setminus B_{d}(u))$
to both sides of~(\ref{eq:ball-excluded-bound}).
\end{proof}

\subsection{\label{subsec:Weight-transfer}Weight transfer in locally smooth
distributions}

Locally smooth functions exhibit great plasticity. In what follows,
we will show that a locally smooth function on $\prod_{i=1}^{n}\{0,1,2,\ldots,r_{i}\}$
can be modified to achieve a broad range of global metric behaviors\textemdash without
the modification being detectable by low-degree polynomials. Among
other things, we will be able to take any locally smooth distribution
and make it globally min-smooth. Our starting point is a generalization
of Lemma~\ref{lem:mass-transfer}, which corresponds to taking $v=0^{n}$
in the new result.
\begin{lem}
\label{lem:mass-transfer-v}Fix points $u,v\in\NN^{n}$ and a natural
number $d<|u-v|.$ Then there is a function $\zeta_{u,v}\colon\cube(u,v)\to\Re$
such that
\begin{align}
 & \supp\zeta_{u,v}\subseteq\{u\}\cup\{x\in\cube(u,v):|x-v|\leq d\},\label{eq:zeta-support-1}\\
 & \zeta_{u,v}(u)=1,\label{eq:zeta-at-u-1}\\
 & \|\zeta_{u,v}\|_{1}\leq1+2^{d}\binom{|u-v|}{d},\label{eq:zeta-norm-1}\\
 & \orth\zeta_{u,v}>d.\label{eq:zeta-orth-1}
\end{align}
\end{lem}

\begin{proof}
Abbreviate $u^{*}=(|u_{1}-v_{1}|,|u_{2}-v_{2}|,\ldots,|u_{n}-v_{n}|).$
Lemma~\ref{lem:mass-transfer} constructs a function $\zeta_{u^{*}}\colon\NN^{n}\to\Re$
such that
\begin{align}
 & \supp\zeta_{u^{*}}\subseteq\{u^{*}\}\cup\{x\in\NN^{n}:x\leq u^{*}\text{ and }|x|\leq d\},\label{eq:zeta-support-2}\\
 & \zeta_{u^{*}}(u^{*})=1,\label{eq:zeta-at-u-2}\\
 & \|\zeta_{u^{*}}\|_{1}\leq1+2^{d}\binom{|u^{*}|}{d},\label{eq:zeta-norm-2}\\
 & \orth\zeta_{u^{*}}>d.\label{eq:zeta-orth-2}
\end{align}
Define $\zeta_{u,v}\colon\cube(u,v)\to\Re$ by
\[
\zeta_{u,v}(x)=\zeta_{u^{*}}(|x_{1}-v_{1}|,|x_{2}-v_{2}|,\ldots,|x_{n}-v_{n}|).
\]
Then~(\ref{eq:zeta-support-1}) and~(\ref{eq:zeta-at-u-1}) are
immediate from~(\ref{eq:zeta-support-2}) and~(\ref{eq:zeta-at-u-2}),
respectively. Property~(\ref{eq:zeta-norm-1}) can be verified as
follows:
\begin{align*}
\|\zeta_{u,v}\|_{1} & =\sum_{x\in\cube(u,v)}|\zeta_{u^{*}}(|x_{1}-v_{1}|,|x_{2}-v_{2}|,\ldots,|x_{n}-v_{n}|)|\\
 & =\sum_{\substack{w\in\NN^{n}:\\
w\leq u^{*}
}
}|\zeta_{u^{*}}(w)|\\
 & \leq1+2^{d}\binom{|u^{*}|}{d},
\end{align*}
where the last step uses~(\ref{eq:zeta-norm-2}). For~(\ref{eq:zeta-orth-1}),
fix an arbitrary polynomial $p$ of degree at most $d.$ Then at every
point $x\in\cube(u,v),$ we have
\begin{align}
p(x) & =p((x_{1}-v_{1})+v_{1},\ldots,(x_{n}-v_{n})+v_{n})\nonumber \\
 & =p(\sign(u_{1}-v_{1})|x_{1}-v_{1}|+v_{1},\ldots,\sign(u_{n}-v_{n})|x_{n}-v_{n}|+v_{n})\nonumber \\
 & =q(|x_{1}-v_{1}|,\ldots,|x_{n}-v_{n}|),\label{eq:p-transformed}
\end{align}
where $q$ is some polynomial of degree at most $d.$ As a result,
\begin{align*}
\langle\zeta_{u,v},p\rangle & =\sum_{x\in\cube(u,v)}\zeta_{u^{*}}(|x_{1}-v_{1}|,\ldots,|x_{n}-v_{n}|)\,p(x)\\
 & =\sum_{x\in\cube(u,v)}\zeta_{u^{*}}(|x_{1}-v_{1}|,\ldots,|x_{n}-v_{n}|)\,q(|x_{1}-v_{1}|,\ldots,|x_{n}-v_{n}|)\\
 & =\sum_{\substack{w\in\NN^{n}:\\
w\leq u^{*}
}
}\zeta_{u^{*}}(w)\,q(w)\\
 & =\langle\zeta_{u^{*}},q\rangle\\
 & =0,
\end{align*}
where the second, fourth, and fifth steps are valid by~(\ref{eq:p-transformed}),
(\ref{eq:zeta-support-2}), and~(\ref{eq:zeta-orth-2}), respectively.
\end{proof}
\noindent Our next result is a smooth analogue of Lemma~\ref{lem:mass-transfer-v}.
The smoothness offers a great deal of flexibility when using the lemma
to transfer $\ell_{1}$ mass from one region of $\NN^{n}$ to another.
\begin{lem}
\label{lem:Zeroizer-for-distributions}Let $X=\prod_{i=1}^{n}\{0,1,2,\ldots,r_{i}\},$
where each $r_{i}\geq0$ is an integer. Let $\theta$ and $d$ be
nonnegative integers with
\[
d<\frac{1}{3}\min\left\{ \theta,\sum_{i=1}^{n}r_{i}\right\} .
\]
Let $\Lambda$ be a probability distribution on $X|_{\leq\theta}.$
Suppose further that $\Lambda$ is $K$-smooth on $X|_{\leq\theta}.$
Then for every $u\in X,$ there is a function $Z_{u}\colon X\to\Re$
with
\begin{align}
 & \supp Z_{u}\subseteq X|_{\leq\theta}\cup\{u\},\label{eq:Zeta-u-supp}\\
 & Z_{u}(u)=1,\label{eq:Zeta-u-at-u}\\
 & \orth Z_{u}>d,\label{eq:Zeta-u-orth}\\
 & \|Z_{u}\|_{1}\leq2^{d}\binom{\diam(\{u\}\cup\supp\Lambda)}{d}+1,\label{eq:Zeta-u-norm}\\
 & |Z_{u}(x)|\leq2^{3d+1}K^{4d+1}\binom{n+d}{d}^{3}\binom{\diam(\{u\}\cup\supp\Lambda)}{d}\Lambda(x), &  & x\ne u.\label{eq:Zeta-u-outside-u}
\end{align}
\end{lem}

\begin{proof}
We have
\begin{align}
1 & =\Lambda(X|_{\leq\theta})\nonumber \\
 & \leq K^{d}\binom{n+d}{d}\Lambda(X|_{\leq\theta-d})\nonumber \\
 & \leq2^{d+1}K^{3d+1}\binom{n+d}{d}^{2}\Lambda(X|_{\leq\theta-d}\setminus B_{d}(u)),\label{eq:Lambda-V}
\end{align}
where the last two steps apply Propositions~\ref{prop:remove-cap}
and~\ref{prop:remove-ball}, respectively. 

We now move on to the construction of $Z_{u}.$ For any $v\in X|_{\leq\theta-d}\setminus B_{d}(u),$
Lemma~\ref{lem:mass-transfer-v} gives a function $\zeta_{u,v}\colon X\to\Re$
with
\begin{align}
 & \supp\zeta_{u,v}\subseteq\{x\in\cube(u,v):|x-v|\leq d\}\cup\{u\}\label{eq:zeta-support-1-1a}\\
 & \phantom{\supp\zeta_{u,v}}\,\subseteq X|_{\leq\theta}\cup\{u\},\label{eq:zeta-support-1-1}\\
 & \zeta_{u,v}(u)=1,\label{eq:zeta-at-u-1-1}\\
 & \orth\zeta_{u,v}>d,\label{eq:zeta-orth-1-1}\\
 & \|\zeta_{u,v}\|_{1}\leq2^{d}\binom{|u-v|}{d}+1.\label{eq:zeta-norm-1-1}
\end{align}
The last inequality can be simplified as follows:
\begin{align}
\|\zeta_{u,v}\|_{1} & \leq2^{d}\binom{\diam(X|_{\leq\theta}\cup\{u\})}{d}+1\nonumber \\
 & \leq2^{d}\binom{\diam(\{u\}\cup\supp\Lambda)}{d}+1,\label{eq:zeta-norm-weakened}
\end{align}
where the first step uses $v\in X|_{\leq\theta}$, and the second
step is legitimate because $\Lambda$ is a $K$-smooth probability
distribution on $X|_{\leq\theta}$ and therefore $\Lambda\ne0$ at
every point of $X|_{\leq\theta}.$ Combining~(\ref{eq:zeta-at-u-1-1})
and~(\ref{eq:zeta-norm-weakened}),
\begin{align}
\|\zeta_{u,v}\|_{\infty} & \leq2^{d}\binom{\diam(\{u\}\cup\supp\Lambda)}{d}.\label{eq:zeta-u-v-infty}
\end{align}
We define $Z_{u}\colon X\to\Re$ by 
\[
Z_{u}(x)=\frac{1}{\Lambda(X|_{\leq\theta-d}\setminus B_{d}(u))}\;\sum_{v\in X|_{\leq\theta-d}\setminus B_{d}(u)}\Lambda(v)\,\zeta_{u,v}(x),
\]
which is legitimate since $\Lambda(X|_{\leq\theta-d}\setminus B_{d}(u))>0$
by~(\ref{eq:Lambda-V}). Then properties~(\ref{eq:Zeta-u-supp}),
(\ref{eq:Zeta-u-at-u}), (\ref{eq:Zeta-u-orth}), and~(\ref{eq:Zeta-u-norm})
for $Z_{u}$ are immediate from the corresponding properties~(\ref{eq:zeta-support-1-1}),
(\ref{eq:zeta-at-u-1-1}), (\ref{eq:zeta-orth-1-1}), and~(\ref{eq:zeta-norm-weakened})
of $\zeta_{u,v}.$

It remains to verify~(\ref{eq:Zeta-u-outside-u}). Fix $x\ne u.$
If $x\notin X|_{\leq\theta},$ then~(\ref{eq:zeta-support-1-1})
implies that $Z_{u}(x)=0$ and therefore (\ref{eq:Zeta-u-outside-u})
holds in that case. In the complementary case when $x\in X|_{\leq\theta},$
we have
\begin{align*}
|Z_{u}(x)| & \leq\sum_{v\in X|_{\leq\theta-d}\setminus B_{d}(u)}\frac{\Lambda(v)}{\Lambda(X|_{\leq\theta-d}\setminus B_{d}(u))}\cdot|\zeta_{u,v}(x)|\\
 & =\sum_{\substack{v\in X|_{\leq\theta-d}\setminus B_{d}(u):\\
|v-x|\leq d
}
}\frac{\Lambda(v)}{\Lambda(X|_{\leq\theta-d}\setminus B_{d}(u))}\cdot|\zeta_{u,v}(x)|\\
 & \leq\sum_{\substack{v\in X|_{\leq\theta-d}\setminus B_{d}(u):\\
|v-x|\leq d
}
}\frac{K^{d}\Lambda(x)}{\Lambda(X|_{\leq\theta-d}\setminus B_{d}(u))}\cdot2^{d}\binom{\diam(\{u\}\cup\supp\Lambda)}{d}\\
 & \leq2^{d}\binom{n+d}{d}\cdot\frac{K^{d}\Lambda(x)}{\Lambda(X|_{\leq\theta-d}\setminus B_{d}(u))}\cdot2^{d}\binom{\diam(\{u\}\cup\supp\Lambda)}{d},
\end{align*}
where the first step applies the triangle inequality to the definition
of $Z_{u}$; the second step uses~(\ref{eq:zeta-support-1-1a}) and
$x\ne u$; the third step applies the $K$-smoothness of $\Lambda$
and substitutes the bound from~(\ref{eq:zeta-u-v-infty}); and the
final step uses~(\ref{eq:ball-around-x}). In view of~(\ref{eq:Lambda-V}),
this completes the proof of~(\ref{eq:Zeta-u-outside-u}).
\end{proof}
We now show how to efficiently zero out a locally smooth function
on points of large Hamming weight. The modified function is pointwise
close to the original and cannot be distinguished from it by any low-degree
polynomial.
\begin{lem}
\label{lem:SMOOTH-remove-large-Hamming}Define $X=\prod_{i=1}^{n}\{0,1,2,\ldots,r_{i}\},$
where each $r_{i}\geq0$ is an integer. Let $\theta$ and $d$ be
nonnegative integers with
\begin{equation}
d<\frac{\theta}{3}.\label{eq:smooth-remove-large-hm-d-theta}
\end{equation}
Let $\Phi\colon X\to\Re$ be a function that is $K$-smooth on $X|_{\leq\theta},$
with $\Phi|_{\leq\theta}\not\equiv0.$ Then there is $\tilde{\Phi}\colon X\to\Re$
such that
\begin{align}
 & \orth(\Phi-\tilde{\Phi})>d,\label{eq:zeroed-distribution-indistinguishable}\\
 & \supp\tilde{\Phi}\subseteq X|_{\leq\theta},\label{eq:zeroed-distribution-support}\\
 & |\Phi-\tilde{\Phi}|\leq2^{3d+1}K^{4d+1}\binom{n+d}{d}^{3}\binom{\diam(\supp\Phi)}{d}\frac{\|\Phi|_{>\theta}\|_{1}}{\|\Phi|_{\leq\theta}\|_{1}}\cdot|\Phi|\nonumber \\
 & \qquad\qquad\qquad\qquad\qquad\qquad\qquad\qquad\qquad\qquad\qquad\qquad\text{on }X|_{\leq\theta.}\label{eq:zeroed-distribution-distance}
\end{align}
\end{lem}

\begin{proof}
If $\theta>\sum_{i=1}^{n}r_{i},$ the lemma holds trivially for $\tilde{\Phi}=\Phi.$
In what follows, we treat the complementary case $\theta\leq\sum_{i=1}^{n}r_{i}.$
By~(\ref{eq:smooth-remove-large-hm-d-theta}),
\[
d<\frac{1}{3}\min\left\{ \theta,\sum_{i=1}^{n}r_{i}\right\} .
\]
Since $\Phi$ is $K$-smooth on $X|_{\leq\theta}$, the probability
distribution $\Lambda$ on $X|_{\leq\theta}$ given by $\Lambda(x)=|\Phi(x)|/\|\Phi|_{\leq\theta}\|_{1}$
is also $K$-smooth. As a result, Lemma~\ref{lem:Zeroizer-for-distributions}
gives for every $u\in X$ a function $Z_{u}\colon X\to\Re$ with
\begin{align}
 & Z_{u}(u)=1,\label{eq:Zeta-u-at-u-1}\\
 & |Z_{u}(x)|\leq2^{3d+1}K^{4d+1}\binom{n+d}{d}^{3}\binom{\diam(\{u\}\cup\supp\Lambda)}{d}\frac{|\Phi(x)|}{\|\Phi|_{\leq\theta}\|_{1}}\nonumber \\
 & \qquad\qquad\qquad\qquad\qquad\qquad\qquad\qquad\qquad\qquad\qquad\qquad\text{ for }x\ne u,\label{eq:Zeta-u-outside-u-1}\\
 & \orth Z_{u}>d,\label{eq:Zeta-u-orth-1}\\
 & \supp Z_{u}\subseteq X|_{\leq\theta}\cup\{u\}.\label{eq:eq:Zeta-supp-1}
\end{align}
Now define
\[
\tilde{\Phi}=\Phi-\sum_{u\in X|_{>\theta}}\Phi(u)Z_{u}.
\]
Then~(\ref{eq:zeroed-distribution-indistinguishable}) is immediate
from~(\ref{eq:Zeta-u-orth-1}). To verify~(\ref{eq:zeroed-distribution-support}),
fix any point $x\in X|_{>\theta}.$ Then
\begin{align*}
\tilde{\Phi}(x) & =\Phi(x)-\sum_{u\in X|_{>\theta}}\Phi(u)Z_{u}(x)\\
 & =\Phi(x)-\Phi(x)Z_{x}(x)\\
 & =0,
\end{align*}
where the last two steps use~(\ref{eq:eq:Zeta-supp-1}) and~(\ref{eq:Zeta-u-at-u-1}),
respectively. 

It remains to verify~(\ref{eq:zeroed-distribution-distance}) on
$X|_{\leq\theta}$:
\begin{align*}
\!\!\!\!\!|\Phi-\tilde{\Phi}| & \leq\sum_{\substack{u\in X|_{>\theta}:\\
\Phi(u)\ne0
}
}|\Phi(u)|\,|Z_{u}|\\
 & \leq2^{3d+1}K^{4d+1}\binom{n+d}{d}^{3}\binom{\diam(\supp\Phi)}{d}\,\sum_{\substack{u\in X|_{>\theta}:\\
\Phi(u)\ne0
}
}|\Phi(u)|\cdot\frac{|\Phi|}{\|\Phi|_{\leq\theta}\|_{1}}\\
 & =2^{3d+1}K^{4d+1}\binom{n+d}{d}^{3}\binom{\diam(\supp\Phi)}{d}\frac{\|\Phi|_{>\theta}\|_{1}}{\|\Phi|_{\leq\theta}\|_{1}}\cdot|\Phi|,
\end{align*}
where the second step uses~(\ref{eq:Zeta-u-outside-u-1}).
\end{proof}
For technical reasons, we need a generalization of the previous lemma
to functions on $\prod_{i=1}^{n}\{\Delta_{i},\Delta_{i}+1,\ldots,\Delta_{i}+r_{i}\}$
for nonnegative integers $\Delta_{i}$ and $r_{i},$ and further to
convex combinations of such functions. We obtain these generalizations
in the two corollaries that follow.
\begin{cor}
\label{cor:Zeroizer-for-distributions}Define $X=\prod_{i=1}^{n}\{\Delta_{i},\Delta_{i}+1,\ldots,\Delta_{i}+r_{i}\},$
where all $\Delta_{i}$ and $r_{i}$ are nonnegative integers. Let
$\theta$ and $d$ be nonnegative integers with
\[
d<\frac{1}{3}\left(\theta-\sum_{i=1}^{n}\Delta_{i}\right).
\]
Let $\Phi\colon X\to\Re$ be a function that is $K$-smooth on $X|_{\leq\theta},$
with $\Phi|_{\leq\theta}\not\equiv0.$ Then there is a function $\tilde{\Phi}\colon X\to\Re$
such that
\begin{align}
 & \orth(\Phi-\tilde{\Phi})>d,\label{eq:zeroed-distribution-indistinguishable-1}\\
 & \supp\tilde{\Phi}\subseteq X|_{\leq\theta},\label{eq:zeroed-distribution-support-1}\\
 & |\Phi-\tilde{\Phi}|\leq2^{3d+1}K^{4d+1}\binom{n+d}{d}^{3}\binom{\diam(\supp\Phi)}{d}\frac{\|\Phi|_{>\theta}\|_{1}}{\|\Phi|_{\leq\theta}\|_{1}}\cdot|\Phi|\nonumber \\
 & \qquad\qquad\qquad\qquad\qquad\qquad\qquad\qquad\qquad\qquad\qquad\qquad\text{on }X|_{\leq\theta.}\label{eq:zeroed-distribution-distance-1}
\end{align}
\end{cor}

\begin{proof}
Abbreviate $X'=\prod_{i=1}^{n}\{0,1,2,\ldots,r_{i}\}$ and $\theta'=\theta-\sum_{i=1}^{n}\Delta_{i}.$
In this notation,
\begin{equation}
d<\frac{\theta'}{3}.\label{eq:d-prime}
\end{equation}
Consider the function $\Phi'\colon X'\to\Re$ given by $\Phi'(x)=\Phi(x+(\Delta_{1},\Delta_{2}\ldots,\Delta_{n})).$
Then any two points $u,v\in X'|_{\leq\theta'}$ obey
\begin{align*}
|\Phi'(u)| & =|\Phi(u+(\Delta_{1},\Delta_{2},\ldots,\Delta_{n}))|\\
 & \leq K^{|u-v|}|\Phi(v+(\Delta_{1},\Delta_{2},\ldots,\Delta_{n}))|\\
 & =K^{|u-v|}|\Phi'(v)|,
\end{align*}
where the second step uses the $K$-smoothness of $\Phi$ on $X|_{\leq\theta}$.
As a result, $\Phi'$ is $K$-smooth on $X'|_{\leq\theta'}.$ Moreover,
$\|\Phi'|_{\leq\theta'}\|_{1}=\|\Phi|_{\leq\theta}\|_{1}>0.$ In view
of~(\ref{eq:d-prime}), Lemma~\ref{lem:SMOOTH-remove-large-Hamming}
gives a function $\tilde{\Phi}'\colon X'\to\Re$ such that
\begin{align*}
 & \orth(\Phi'-\tilde{\Phi}')>d,\\
 & \supp\tilde{\Phi}'\subseteq X'|_{\leq\theta'},
\end{align*}
and
\begin{align*}
 & |\Phi'-\tilde{\Phi}'|\leq2^{3d+1}K^{4d+1}\binom{n+d}{d}^{3}\binom{\diam(\supp\Phi')}{d}\frac{\|\Phi'|_{>\theta'}\|_{1}}{\|\Phi'|_{\leq\theta'}\|_{1}}\cdot|\Phi'|\\
 & \phantom{|\Lambda'-\tilde{\Lambda}'|}=2^{3d+1}K^{4d+1}\binom{n+d}{d}^{3}\binom{\diam(\supp\Phi)}{d}\frac{\|\Phi|_{>\theta}\|_{1}}{\|\Phi|_{\leq\theta}\|_{1}}\cdot|\Phi'|
\end{align*}
on $X'|_{\leq\theta'}.$ As a result,~(\ref{eq:zeroed-distribution-indistinguishable-1})\textendash (\ref{eq:zeroed-distribution-distance-1})
hold for the real-valued function $\tilde{\Phi}\colon X\to\Re$ given
by $\tilde{\Phi}(x)=\tilde{\Phi}'(x-(\Delta_{1},\Delta_{2},\ldots,\Delta_{n})).$
\end{proof}
\begin{cor}
\label{cor:Zeroizer-for-distributions-convex-hull}Fix integers $\Delta,d,\theta\geq0$
and $n\geq1,$ and a real number $\delta,$ where 
\begin{align*}
 & \delta\in[0,1),\\
 & d<\frac{1}{3}(\theta-\Delta).
\end{align*}
Then for every
\[
\Lambda\in\conv(\Smooth(n,K,\Delta)\cap\{\Lambda'\in\Distribution(\NN^{n}):\Lambda'(\NN^{n}|_{>\theta})\leq\delta\}),
\]
there is a function $\tilde{\Lambda}\colon\NN^{n}\to\Re$ such that
\begin{align*}
 & \orth(\Lambda-\tilde{\Lambda})>d,\\
 & \supp\tilde{\Lambda}\subseteq\NN^{n}|_{\leq\theta}\cap\supp\Lambda,\\
 & |\Lambda-\tilde{\Lambda}|\leq2^{3d+1}K^{4d+1}\binom{n+d}{d}^{3}\binom{\diam(\supp\Lambda)}{d}\frac{\delta}{1-\delta}\cdot\Lambda\qquad\text{ on }\NN^{n}|_{\leq\theta}.
\end{align*}
\end{cor}

\begin{proof}
Write $\Lambda$ out explicitly as
\[
\Lambda=\sum_{i=1}^{N}\lambda_{i}\Lambda_{i}
\]
for some positive reals $\lambda_{1},\ldots,\lambda_{N}$ with $\sum\lambda_{i}=1,$
where $\Lambda_{i}\in\Smooth(n,K,\Delta)$ and $\Lambda_{i}(\NN^{n}|_{>\theta})\leq\delta$.
Then clearly
\begin{equation}
\supp\Lambda=\bigcup_{i=1}^{n}\supp\Lambda_{i}.\label{eq:chain-starts}
\end{equation}
For $i=1,2,\ldots,N,$ Corollary~\ref{cor:Zeroizer-for-distributions}
constructs $\tilde{\Lambda}_{i}\colon\NN^{n}\to\Re$ with
\begin{align}
 & \orth(\Lambda_{i}-\tilde{\Lambda}_{i})>d,\\
 & \supp\tilde{\Lambda}_{i}\subseteq\NN^{n}|_{\leq\theta},\\
 & |\Lambda_{i}-\tilde{\Lambda}_{i}|\leq2^{3d+1}K^{4d+1}\binom{n+d}{d}^{3}\binom{\diam(\supp\Lambda_{i})}{d}\frac{\delta}{1-\delta}\cdot\Lambda_{i}\nonumber \\
 & \qquad\qquad\qquad\qquad\qquad\qquad\qquad\qquad\qquad\qquad\qquad\qquad\text{ on }\NN^{n}|_{\leq\theta},\\
 & \supp\tilde{\Lambda}_{i}\subseteq\supp\Lambda_{i},\label{eq:chain-ends}
\end{align}
where the last property follows from the two before it. In view of~(\ref{eq:chain-starts})\textendash (\ref{eq:chain-ends}),
the proof is complete by taking $\tilde{\Lambda}=\sum_{i=1}^{N}\lambda_{i}\tilde{\Lambda}_{i}$.
\end{proof}
Our next result uses local smoothness to achieve something completely
different. Here, we show how to start with a locally smooth function
and make it globally min-smooth. The new function has the same sign
pointwise as the original, and cannot be distinguished from it by
any low-degree polynomial. Crucially for us, the global min-smoothness
can be achieved relative to any distribution on the domain.
\begin{lem}
\label{lem:SMOOTH-redistribute}Define $X=\prod_{i=1}^{n}\{0,1,2,\ldots,r_{i}\},$
where each $r_{i}\geq0$ is an integer. Let $\theta$ and $d$ be
nonnegative integers with
\[
d<\frac{1}{3}\min\left\{ \theta,\sum_{i=1}^{n}r_{i}\right\} .
\]
Let $\Phi\colon X|_{\leq\theta}\to\Re$ be a function that is $K$-smooth
on $X|_{\leq\theta}.$ Then for every probability distribution $\Lambda^{*}$
on $X|_{\leq\theta},$ there is $\Phi^{*}\colon X|_{\leq\theta}\to\Re$
such that
\begin{align}
 & \orth(\Phi-\Phi^{*})>d,\label{eq:smoothed-distribution-indistinguishable}\\
 & \|\Phi^{*}\|_{1}\leq2\|\Phi\|_{1},\label{eq:smoothed-distribution-ell1}\\
 & \Phi\cdot\Phi^{*}\geq0,\label{eq:smoothed-distribution-sign}\\
 & |\Phi^{*}|\geq\left(2^{3d+1}K^{4d+1}\binom{n+d}{d}^{3}\binom{\diam(\supp\Phi)}{d}\right)^{-1}\|\Phi\|_{1}\,\Lambda^{*}.\label{eq:smoothed-distribution-factor}
\end{align}
\end{lem}

\begin{proof}
If $\Phi\equiv0,$ the lemma holds trivially with $\Phi^{*}=\Phi.$
In the complementary case, abbreviate
\[
N=2^{3d+1}K^{4d+1}\binom{n+d}{d}^{3}\binom{\diam(\supp\Phi)}{d}.
\]
We will view $|\Phi|/\|\Phi\|_{1}$ as a probability distribution
on $X|_{\leq\theta}.$ By hypothesis, this probability distribution
is $K$-smooth on $X|_{\leq\theta}.$ In particular, $\supp|\Phi|=\supp\Phi=X|_{\leq\theta}.$
Therefore, Lemma~\ref{lem:Zeroizer-for-distributions} gives for
every $u\in X|_{\leq\theta}$ a function $Z_{u}\colon X|_{\leq\theta}\to\Re$
with
\begin{align}
 & Z_{u}(u)=1,\label{eq:Zeta-u-at-u-2}\\
 & \|Z_{u}\|_{1}\leq\frac{N}{2}+1,\label{eq:Zeta-u-norm-2}\\
 & |Z_{u}(x)|\leq N\cdot\frac{|\Phi(x)|}{\|\Phi\|_{1}}, &  & x\ne u,\label{eq:Zeta-u-outside-u-2}\\
 & \orth Z_{u}>d.\label{eq:Zeta-u-orth-2}
\end{align}
Now, define $\Phi^{*}\colon X|_{\leq\theta}\to\Re$ by
\[
\Phi^{*}=\Phi+\frac{\|\Phi\|_{1}}{N}\sum_{u\in X|_{\leq\theta}}\Sgn(\Phi(u))\Lambda^{*}(u)Z_{u}.
\]
Then~(\ref{eq:smoothed-distribution-indistinguishable}) follows
directly from~(\ref{eq:Zeta-u-orth-2}). For~(\ref{eq:smoothed-distribution-ell1}),
we have:
\begin{align}
\|\Phi^{*}\|_{1} & \leq\|\Phi\|_{1}+\frac{\|\Phi\|_{1}}{N}\sum_{u\in X|_{\leq\theta}}\Lambda^{*}(u)\,\|Z_{u}\|_{1}\nonumber \\
 & \leq\|\Phi\|_{1}+\frac{\|\Phi\|_{1}}{N}\cdot\left(\frac{N}{2}+1\right)\sum_{u\in X|_{\leq\theta}}\Lambda^{*}(u)\nonumber \\
 & =\frac{3N+2}{2N}\|\Phi\|_{1}\nonumber \\
\rule{0mm}{5mm} & \leq2\,\|\Phi\|_{1},\label{eq:redistribute-Phi-tilde-norm}
\end{align}
where the second step uses~(\ref{eq:Zeta-u-norm-2}). The remaining
properties~(\ref{eq:smoothed-distribution-sign}) and~(\ref{eq:smoothed-distribution-factor})
can be established simultaneously as follows: for every $x\in X|_{\leq\theta}$,
\begin{align}
\Sgn(\Phi(x)) & \cdot\Phi^{*}(x)\nonumber \\
 & =|\Phi(x)|+\frac{\|\Phi\|_{1}}{N}\sum_{u\in X|_{\leq\theta}}\Lambda^{*}(u)Z_{u}(x)\nonumber \\
 & \geq|\Phi(x)|+\frac{\|\Phi\|_{1}}{N}\,\Lambda^{*}(x)Z_{x}(x)-\frac{\|\Phi\|_{1}}{N}\sum_{\substack{u\in X|_{\leq\theta}:\\
u\ne x
}
}\Lambda^{*}(u)\,|Z_{u}(x)|\nonumber \\
 & =|\Phi(x)|+\frac{\|\Phi\|_{1}}{N}\,\Lambda^{*}(x)-\frac{\|\Phi\|_{1}}{N}\sum_{\substack{u\in X|_{\leq\theta}:\\
u\ne x
}
}\Lambda^{*}(u)\,|Z_{u}(x)|\nonumber \\
 & \geq|\Phi(x)|+\frac{\|\Phi\|_{1}}{N}\,\Lambda^{*}(x)-\frac{\|\Phi\|_{1}}{N}\cdot N\cdot\frac{|\Phi(x)|}{\|\Phi\|_{1}}\sum_{\substack{u\in X|_{\leq\theta}:\\
u\ne x
}
}\Lambda^{*}(u)\nonumber \\
 & =|\Phi(x)|+\frac{\|\Phi\|_{1}}{N}\,\Lambda^{*}(x)-|\Phi(x)|\,(1-\Lambda^{*}(x))\nonumber \\
 & \geq\frac{\|\Phi\|_{1}}{N}\,\Lambda^{*}(x),\label{eq:redistribute-Phi-tilde-smoothness}
\end{align}
where the third and fourth steps use~(\ref{eq:Zeta-u-at-u-2}) and~(\ref{eq:Zeta-u-outside-u-2}),
respectively.
\end{proof}

\subsection{\label{subsec:A-locally-smooth}A locally smooth dual polynomial
for MP}

As Sections~\ref{subsec:Local-smoothness}\textendash \ref{subsec:Weight-transfer}
show, local smoothness implies several useful metric and analytic
properties. To tap into this resource, we now construct a locally
smooth dual polynomial for the Minsky\textendash Papert function.
It is helpful to view this new result as a counterpart of Theorem~\ref{thm:dual-MP}
from our analysis of the threshold degree of $\classAC^{0}$. The
new proof is considerably more technical because local smoothness
is a delicate property to achieve. 
\begin{thm}
\label{thm:dual-MP-smooth}For some absolute constant $0<c<1$ and
all positive integers $m,r,R$ with $r\leq R,$ there are probability
distributions $\Lambda_{0}$ and $\Lambda_{1}$ such that
\begin{align}
 & \supp\Lambda_{0}=(\MP_{m,R}^{*})^{-1}(0),\label{eq:srank-Lambda0-supp}\\
 & \supp\Lambda_{1}=(\MP_{m,R}^{*})^{-1}(1),\label{eq:srank-Lambda1-supp}\\
 & \orth(\Lambda_{0}-\Lambda_{1})\geq\min\{m,c\sqrt{r}\},\label{eq:srank-Lambda-orth-1}\\
 & \frac{\Lambda_{0}+\Lambda_{1}}{2}\in\SmoothFunction\left(\frac{m}{c},\{0,1,2,\ldots,R\}^{m}\right),\label{eq:srank-Lambda-smooth}\\
 & \Lambda_{0},\Lambda_{1}\in\conv\left(\left\{ \lambda\in\Smooth\left(1,\frac{1}{c},1\right):\phantom{\left\{ \frac{1}{c(t+1)^{2}\,2^{-ct/\sqrt{r}}}\right\} ^{\otimes m}}\right.\right.\nonumber \\
 & \quad\qquad\qquad\qquad\qquad\left.\left.\lambda(t)\leq\frac{1}{c(t+1)^{2}\,2^{ct/\sqrt{r}}}\text{ for }t\in\NN\right\} ^{\otimes m}\right).\label{eq:srank-Lambda-z-smooth}
\end{align}
\end{thm}

\noindent Our proof of Theorem~\ref{thm:dual-MP-smooth} repeatedly
employs the following simple but useful\emph{ }criterion for $K$-smoothness:
a probability distribution $\lambda$ is $K$-smooth on an integer
interval $I=\{i,i+1,i+2,\ldots,j\}$ if and only if the probabilities
of any two \emph{consecutive} integers in $I$ are within a factor
of $K$.
\begin{proof}[Proof of Theorem~\emph{\ref{thm:dual-MP-smooth}}.]
Abbreviate $\epsilon=1/6.$ For some absolute constants $c',c''\in(0,1)$,
Lemma~\ref{lem:dual-OR-distributions} constructs probability distributions
$\lambda_{0},\lambda_{1},\lambda_{2}$ such that
\begin{align}
 & \supp\lambda_{0}=\{0\},\label{eq:Lambda01-mu0}\\
 & \supp\lambda_{i}=\{1,2,\ldots,R\}, &  & i=1,2,\label{eq:Lambda01-mui}\\
 & \lambda_{i}(t)\in\left[\frac{c'}{t^{2}\,2^{c''t/\sqrt{r}}},\;\frac{1}{c't^{2}\,2^{c''t/\sqrt{r}}}\right], &  & i=1,2;\quad t=1,2,\ldots,R,\label{eq:Lambda01-mui-pointwise}\\
 & \orth((1-\epsilon)\lambda_{0}+\epsilon\lambda_{2}-\lambda_{1})\geq c'\sqrt{r}.\label{eq:Lambda01-orth}
\end{align}
We infer that
\begin{align}
\lambda_{0} & \in\Smooth(1,K),\label{eq:mu0-smooth}\\
\lambda_{1} & \in\Smooth(1,K,1),\label{eq:mu1-smooth}\\
\lambda_{2} & \in\Smooth(1,K,1),\label{eq:mu2-smooth}\\
(1-\epsilon)\lambda_{0}+\epsilon\lambda_{2} & \in\Smooth(1,K),\label{eq:mu02-smooth}\\
\frac{1}{m+1}\lambda_{0}+\frac{m}{m+1}\lambda_{1} & \in\Smooth(1,Km)\label{eq:mu01-smooth}
\end{align}
for some large constant $K=K(c',c'')\geq1.$ Indeed, (\ref{eq:mu0-smooth})
is trivial since $\lambda_{0}$ is the single-point distribution on
the origin; (\ref{eq:mu1-smooth}) holds because by~(\ref{eq:Lambda01-mui})
and~(\ref{eq:Lambda01-mui-pointwise}), the probabilities of any
pair of consecutive integers in $\supp\lambda_{1}=\{1,2,\ldots,R\}$
are the same up to a constant factor; and~(\ref{eq:mu2-smooth})\textendash (\ref{eq:mu01-smooth})
can be seen analogously, by comparing the probabilities of any pair
of consecutive integers. Combining (\ref{eq:mu0-smooth})\textendash (\ref{eq:mu01-smooth})
with Proposition~\ref{prop:smooth-tensor}, we obtain
\begin{align}
\{\lambda_{0},\lambda_{1},\lambda_{2}\}^{\otimes m} & \subseteq\Smooth(m,K,m),\label{eq:mui-tensor}\\
((1-\epsilon)\lambda_{0}+\epsilon\lambda_{2})^{\otimes m} & \in\Smooth(m,K),\label{eq:mu02-tensor}\\
\left(\frac{1}{m+1}\lambda_{0}+\frac{m}{m+1}\lambda_{1}\right)^{\otimes m} & \in\Smooth(m,Km).\label{eq:mu01-tensor}
\end{align}

The proof centers around the dual objects $\Psi_{1},\Psi_{2}\colon\{0,1,2,\ldots,R\}^{m}\to\Re$
given by
\begin{align*}
\Psi_{1} & =\left(\frac{1}{m+1}\lambda_{0}+\frac{m}{m+1}\lambda_{1}\right)^{\otimes m}-2\lambda_{1}^{\otimes m}
\end{align*}
and
\begin{multline*}
\Psi_{2}=2((1-\epsilon)\lambda_{0}+\epsilon\lambda_{2})^{\otimes m}-2(-\epsilon\lambda_{0}+\epsilon\lambda_{2})^{\otimes m}\\
-\left(\frac{1}{m+1}\lambda_{0}+\frac{m}{m+1}((1-\epsilon)\lambda_{0}+\epsilon\lambda_{2})\right)^{\otimes m}.
\end{multline*}
The next four claims establish key properties of $\Psi_{1}$ and $\Psi_{2}.$
\begin{claim}
\label{claim:Psi1}$\Psi_{1}$ satisfies
\begin{align}
 & \pospart\Psi_{1}\in\cone(\{\lambda_{0},\lambda_{1}\}^{\otimes m}\setminus\{\lambda_{1}^{\otimes m}\}),\label{eq:Psi1-pospart}\\
 & \negpart\Psi_{1}\in\cone\{\lambda_{1}^{\otimes m}\},\label{eq:Psi1-negpart}\\
 & \frac{1}{5}|\Psi_{1}|\leq\left(\frac{1}{m+1}\lambda_{0}+\frac{m}{m+1}\lambda_{1}\right)^{\otimes m}\leq|\Psi_{1}|.\label{eq:Psi1-smooth}
\end{align}
\end{claim}

\begin{claim}
\label{claim:Psi2}$\Psi_{2}$ satisfies
\begin{align}
 & \pospart\Psi_{2}\in\cone(\{\lambda_{0},\lambda_{2}\}^{\otimes m}\setminus\{\lambda_{2}^{\otimes m}\}),\label{eq:Psi2-pospart}\\
 & \negpart\Psi_{2}\in\cone\{\lambda_{2}^{\otimes m}\},\label{eq:Psi2-negpart}\\
 & \frac{1}{3}|\Psi_{2}|\leq\left((1-\epsilon)\lambda_{0}+\epsilon\lambda_{2}\right)^{\otimes m}\leq3|\Psi_{2}|.\label{eq:Psi2-smooth}
\end{align}
\end{claim}

\begin{claim}
\label{claim:Psi1-Psi2-support}$\Psi_{1}$ and $\Psi_{2}$ satisfy
\begin{align}
\supp(\pospart\Psi_{i}) & =(\MP_{m,R}^{*})^{-1}(0), &  & i=1,2,\label{eq:Psi12-pospart}\\
\supp(\negpart\Psi_{i}) & =(\MP_{m,R}^{*})^{-1}(1), &  & i=1,2.\label{eq:Psi12-negpart}
\end{align}
\end{claim}

\begin{claim}
\label{claim:Psi1-Psi2-orthog}$\orth(\Psi_{1}+\Psi_{2})\geq\min\{m,c'\sqrt{r}\}.$
\end{claim}

We will settle Claims~\ref{claim:Psi1}\textendash \ref{claim:Psi1-Psi2-orthog}
shortly, once we complete the main proof. Define
\begin{align*}
 & \Lambda_{0}=\frac{2}{\|\Psi_{1}\|_{1}+\|\Psi_{2}\|_{1}}\pospart(\Psi_{1}+\Psi_{2}),\\
 & \Lambda_{1}=\frac{2}{\|\Psi_{1}\|_{1}+\|\Psi_{2}\|_{1}}\negpart(\Psi_{1}+\Psi_{2}),
\end{align*}
where the denominators are nonzero by~(\ref{eq:Psi1-smooth}). We
proceed to verify the properties required of $\Lambda_{0}$ and $\Lambda_{1}$
in the theorem statement.

~

\textsc{Support.} Recall from Claim~\ref{claim:Psi1-Psi2-support}
that the positive parts of $\Psi_{1}$ and $\Psi_{2}$ are supported
on $(\MP_{m,R}^{*})^{-1}(0).$ Therefore, the positive part of $\Psi_{1}+\Psi_{2}$
is supported on $(\MP_{m,R}^{*})^{-1}(0)$ as well, which in turn
implies that
\begin{equation}
\supp\Lambda_{0}=(\MP_{m,R}^{*})^{-1}(0).\label{eq:supp-Lambda0}
\end{equation}
Analogously, Claim~\ref{claim:Psi1-Psi2-support} states that the
negative parts of $\Psi_{1}$ and $\Psi_{2}$ are supported on $(\MP_{m,R}^{*})^{-1}(1).$
As a result, the negative part of $\Psi_{1}+\Psi_{2}$ is also supported
on $(\MP_{m,R}^{*})^{-1}(1)$, whence
\begin{equation}
\supp\Lambda_{1}=(\MP_{m,R}^{*})^{-1}(1).\label{eq:supp-Lambda1}
\end{equation}
~

\textsc{Orthogonality.} The defining equations for $\Lambda_{0}$
and $\Lambda_{1}$ imply that 
\[
\Lambda_{0}-\Lambda_{1}=\frac{2}{\|\Psi_{1}\|_{1}+\|\Psi_{2}\|_{1}}\,(\Psi_{1}+\Psi_{2}),
\]
 which along with Claim~\ref{claim:Psi1-Psi2-orthog} forces
\begin{equation}
\orth(\Lambda_{0}-\Lambda_{1})\geq\min\{m,c'\sqrt{r}\}.\label{eq:inproof-Lambda0-Lambda1-orth}
\end{equation}
~

\textsc{Nonnegativity and norm.} By definition, $\Lambda_{0}$ and
$\Lambda_{1}$ are nonnegative functions. We calculate
\begin{align}
\|\Lambda_{0}\|_{1}-\|\Lambda_{1}\|_{1} & =\langle\Lambda_{0},1\rangle-\langle\Lambda_{1},1\rangle\nonumber \\
 & =\langle\Lambda_{0}-\Lambda_{1},1\rangle\nonumber \\
 & =0,\label{eq:inproof-Lambdas-diff}
\end{align}
where the first step uses the nonnegativity of $\Lambda_{0}$ and
$\Lambda_{1}$, and the last step applies~(\ref{eq:inproof-Lambda0-Lambda1-orth}).
In addition,
\begin{align}
\|\Lambda_{0}\|_{1}+\|\Lambda_{1}\|_{1} & =\frac{2}{\|\Psi_{1}\|_{1}+\|\Psi_{2}\|_{1}}(\|\pospart(\Psi_{1}+\Psi_{2})\|_{1}+\|\negpart(\Psi_{1}+\Psi_{2})\|_{1})\nonumber \\
 & =\frac{2}{\|\Psi_{1}\|_{1}+\|\Psi_{2}\|_{1}}\|\Psi_{1}+\Psi_{2}\|_{1}\nonumber \\
 & =2,\label{eq:inproof-Lambdas-sum}
\end{align}
where the last step uses~Claim~\ref{claim:Psi1-Psi2-support}. A
consequence of~(\ref{eq:inproof-Lambdas-diff}) and~(\ref{eq:inproof-Lambdas-sum})
is that $\|\Lambda_{0}\|_{1}=\|\Lambda_{1}\|_{1}=1$, which makes
$\Lambda_{0}$ and $\Lambda_{1}$ probability distributions. In view
of~(\ref{eq:supp-Lambda0}) and~(\ref{eq:supp-Lambda1}), we conclude
that
\begin{align}
\Lambda_{i} & \in\Distribution((\MP_{m,R}^{*})^{-1}(i)), &  & i=0,1.\label{eq:inproof-Lambdas-probab-distributions}
\end{align}
In particular,
\begin{equation}
\frac{\Lambda_{0}+\Lambda_{1}}{2}\in\Distribution(\{0,1,2,\ldots,R\}^{m}).\label{eq:Lambda01-distribution}
\end{equation}
 ~

\textsc{Smoothness.} We have
\begin{align}
\frac{\Lambda_{0}+\Lambda_{1}}{2} & =\frac{|\Psi_{1}+\Psi_{2}|}{\|\Psi_{1}\|_{1}+\|\Psi_{2}\|_{1}}\nonumber \\
 & =\frac{1}{\|\Psi_{1}\|_{1}+\|\Psi_{2}\|_{1}}\,|\Psi_{1}|+\frac{1}{\|\Psi_{1}\|_{1}+\|\Psi_{2}\|_{1}}\,|\Psi_{2}|,\label{eq:Lambda0-plus-Lambda1}
\end{align}
where the first step follows from the defining equations for $\Lambda_{0}$
and $\Lambda_{1}$, and the second step uses Claim~\ref{claim:Psi1-Psi2-support}.
Inequality~(\ref{eq:Psi1-smooth}) shows that at every point, $|\Psi_{1}|$
is within a factor of $5$ of the tensor product~$(\frac{1}{m+1}\lambda_{0}+\frac{m}{m+1}\lambda_{1})^{\otimes m}$,
which by~(\ref{eq:mu01-tensor}) is $Km$-smooth on its support.
It follows that $|\Psi_{1}|$ is $25Km$-smooth on $\{0,1,2,\ldots,R\}^{m}.$
By an analogous argument,~(\ref{eq:Psi2-smooth}) and~(\ref{eq:mu02-tensor})
imply that $|\Psi_{2}|$ is $9K$-smooth (and hence also $25Km$-smooth)
on $\{0,1,2,\ldots,R\}^{m}.$ Now~(\ref{eq:Lambda0-plus-Lambda1})
shows that $\frac{1}{2}(\Lambda_{0}+\Lambda_{1})$ is a conical combination
of two nonnegative $25Km$-smooth functions on $\{0,1,2,\ldots,R\}^{m}.$
By Proposition~\ref{prop:smooth}\ref{enu:smooth-conical},
\begin{equation}
\frac{\Lambda_{0}+\Lambda_{1}}{2}\in\SmoothFunction(25Km,\{0,1,2,\ldots,R\}^{m}).\label{eq:Lambda0-sum-Lambda1-smooth}
\end{equation}

Having examined the convex combination $\frac{\Lambda_{0}+\Lambda_{1}}{2},$
we now turn to the individual distributions $\Lambda_{0}$ and $\Lambda_{1}$.
We have 
\begin{align*}
\Lambda_{0} & =\frac{2}{\|\Psi_{1}\|_{1}+\|\Psi_{2}\|_{1}}\pospart(\Psi_{1}+\Psi_{2})\\
 & =\frac{2}{\|\Psi_{1}\|_{1}+\|\Psi_{2}\|_{1}}(\pospart(\Psi_{1})+\pospart(\Psi_{2}))\\
\rule{0mm}{4mm} & \in\cone(\{\lambda_{0},\lambda_{1},\lambda_{2}\}^{\otimes m}),
\end{align*}
where the first equation restates the definition of $\Lambda_{0},$
the second step applies~(\ref{eq:Psi12-pospart}), and the last step
uses~(\ref{eq:Psi1-pospart}) and (\ref{eq:Psi2-pospart}). Analogously,
\begin{align*}
\Lambda_{1} & =\frac{2}{\|\Psi_{1}\|_{1}+\|\Psi_{2}\|_{1}}\negpart(\Psi_{1}+\Psi_{2})\\
 & =\frac{2}{\|\Psi_{1}\|_{1}+\|\Psi_{2}\|_{1}}(\negpart(\Psi_{1})+\negpart(\Psi_{2}))\\
\rule{0mm}{4mm} & \in\cone(\{\lambda_{1}^{\otimes m},\lambda_{2}^{\otimes m}\}),
\end{align*}
where the first equation restates the definition of $\Lambda_{1},$
the second step applies~(\ref{eq:Psi12-negpart}), and the last step
uses~(\ref{eq:Psi1-negpart}) and (\ref{eq:Psi2-negpart}). Thus,
$\Lambda_{0}$ and $\Lambda_{1}$ are\emph{ }conical\emph{ }combinations
of probability distributions in $\{\lambda_{0},\lambda_{1},\lambda_{2}\}^{\otimes m}.$
Since $\Lambda_{0}$ and $\Lambda_{1}$ are themselves probability
distributions, we conclude that
\[
\Lambda_{0},\Lambda_{1}\in\conv(\{\lambda_{0},\lambda_{1},\lambda_{2}\}^{\otimes m}).
\]
 By~(\ref{eq:Lambda01-mu0})\textendash (\ref{eq:Lambda01-mui-pointwise}),
\begin{align*}
\lambda_{i}(t) & \leq\frac{1}{c'''(t+1)^{2}\,2^{c'''t/\sqrt{r}}} &  & (t\in\NN;\;i=0,1,2)
\end{align*}
for some constant $c'''>0.$ The last two equations along with~(\ref{eq:mu0-smooth})\textendash (\ref{eq:mu2-smooth})
yield
\begin{multline}
\Lambda_{0},\Lambda_{1}\in\conv\left(\left\{ \lambda\in\Smooth(1,K,1):\phantom{\left\{ \frac{1}{c(t+1)^{2}\,2^{-ct/\sqrt{r}}}\right\} ^{\otimes m}}\right.\right.\\
\left.\left.\lambda(t)\leq\frac{1}{c'''(t+1)^{2}\,2^{c'''t/\sqrt{r}}}\text{ for }t\in\NN\right\} ^{\otimes m}\right).\label{eq:Lambda0-Lambda1-smooth}
\end{multline}
 Now~(\ref{eq:supp-Lambda0})\textendash (\ref{eq:inproof-Lambda0-Lambda1-orth}),
(\ref{eq:Lambda0-sum-Lambda1-smooth}), and (\ref{eq:Lambda0-Lambda1-smooth})
imply~(\ref{eq:srank-Lambda0-supp})\textendash (\ref{eq:srank-Lambda-z-smooth})
for a small enough constant $c>0.$
\end{proof}
We now settle the four claims made in the proof of Theorem~\ref{thm:dual-MP-smooth}.
\begin{proof}[Proof of Claim~\emph{\ref{claim:Psi1}}.]
Multiplying out the tensor product in the definition of $\Psi_{1}$
and collecting like terms, we obtain
\begin{multline}
\Psi_{1}=-\left(2-\left(\frac{m}{m+1}\right)^{m}\right)\lambda_{1}^{\otimes m}\\
+\sum_{\substack{S\subseteq\{1,2,\ldots,m\}\\
S\ne\varnothing
}
}\left(\frac{1}{m+1}\right)^{|S|}\left(\frac{m}{m+1}\right)^{m-|S|}\lambda_{0}^{\otimes S}\cdot\lambda_{1}^{\otimes\overline{S}}.\qquad\label{eq:Phi1-multiply-out}
\end{multline}
Recall from~(\ref{eq:Lambda01-mu0}) and~(\ref{eq:Lambda01-mui})
that $\lambda_{0}$ and $\lambda_{1}$ are supported on $\{0\}$ and
$\{1,2,\ldots,R\},$ respectively. Therefore, the right-hand side
of~(\ref{eq:Phi1-multiply-out}) is the sum of $2^{m}$ nonzero functions
whose supports are pairwise disjoint. Now~(\ref{eq:Psi1-pospart})
and (\ref{eq:Psi1-negpart}) follow directly from~(\ref{eq:Phi1-multiply-out}).
One further obtains that
\begin{multline*}
|\Psi_{1}|=\left(2-\left(\frac{m}{m+1}\right)^{m}\right)\lambda_{1}^{\otimes m}\\
\qquad\qquad+\sum_{\substack{S\subseteq\{1,2,\ldots,m\}\\
S\ne\varnothing
}
}\left(\frac{1}{m+1}\right)^{|S|}\left(\frac{m}{m+1}\right)^{m-|S|}\lambda_{0}^{\otimes S}\cdot\lambda_{1}^{\otimes\overline{S}}.
\end{multline*}
From first principles, 
\begin{multline*}
\left(\frac{1}{m+1}\lambda_{0}+\frac{m}{m+1}\lambda_{1}\right)^{\otimes m}=\left(\frac{m}{m+1}\right)^{m}\lambda_{1}^{\otimes m}\\
\qquad\qquad+\sum_{\substack{S\subseteq\{1,2,\ldots,m\}\\
S\ne\varnothing
}
}\left(\frac{1}{m+1}\right)^{|S|}\left(\frac{m}{m+1}\right)^{m-|S|}\lambda_{0}^{\otimes S}\cdot\lambda_{1}^{\otimes\overline{S}}.
\end{multline*}
Comparing the right-hand sides of the last two equations settles~(\ref{eq:Psi1-smooth}).
\end{proof}
\begin{proof}[Proof of Claim~\emph{\ref{claim:Psi2}}.]
Multiplying out the tensor powers in the definition of $\Psi_{2}$
and collecting like terms, we obtain
\begin{equation}
\Psi_{2}=-\left(\frac{m}{m+1}\right)^{m}\epsilon^{m}\lambda_{2}^{\otimes m}+\sum_{\substack{S\subseteq\{1,2,\ldots,m\}\\
S\ne\varnothing
}
}a_{|S|}\,\lambda_{0}^{\otimes S}\cdot\lambda_{2}^{\otimes\overline{S}},\label{eq:Psi2-multiplied-out}
\end{equation}
where the coefficients $a_{1},a_{2},\ldots,a_{m}$ are given by
\begin{align}
a_{i} & =\left(2(1-\epsilon)^{i}\epsilon^{m-i}-2(-1)^{i}\epsilon^{m}-\left(1-\frac{\epsilon m}{m+1}\right)^{i}\left(\frac{\epsilon m}{m+1}\right)^{m-i}\right)\nonumber \\
 & =(1-\epsilon)^{i}\epsilon^{m-i}\left(2-2\left(\frac{-\epsilon}{1-\epsilon}\right)^{i}-\left(1+\frac{\epsilon}{(1-\epsilon)(m+1)}\right)^{i}\left(\frac{m}{m+1}\right)^{m-i}\right)\nonumber \\
 & \in\left[\frac{1}{3}(1-\epsilon)^{i}\epsilon^{m-i},\,3(1-\epsilon)^{i}\epsilon^{m-i}\right].\label{eq:Psi2-ai}
\end{align}
As in the proof of the previous claim, recall from~(\ref{eq:Lambda01-mu0})
and~(\ref{eq:Lambda01-mui}) that $\lambda_{0}$ and $\lambda_{2}$
have disjoint support. Therefore, the right-hand side of~(\ref{eq:Psi2-multiplied-out})
is the sum of $2^{m}$ nonzero functions whose supports are pairwise
disjoint. Now~(\ref{eq:Psi2-pospart}) and (\ref{eq:Psi2-negpart})
are immediate from~(\ref{eq:Psi2-ai}). The disjointness of the supports
of the summands on the right-hand side of~(\ref{eq:Psi2-multiplied-out})
also implies that
\[
|\Psi_{2}|=\left(\frac{m}{m+1}\right)^{m}\epsilon^{m}\lambda_{2}^{\otimes m}+\sum_{\substack{S\subseteq\{1,2,\ldots,m\}\\
S\ne\varnothing
}
}|a_{|S|}|\,\lambda_{0}^{\otimes S}\cdot\lambda_{2}^{\otimes\overline{S}}.
\]
In view of~(\ref{eq:Psi2-ai}), we conclude that $|\Psi_{2}|$ coincides
up to a factor of $3$ with the function
\[
\sum_{S\subseteq\{1,2,\ldots,m\}}(1-\epsilon)^{|S|}\epsilon^{m-|S|}\lambda_{0}^{\otimes S}\cdot\lambda_{2}^{\otimes\overline{S}}=((1-\epsilon)\lambda_{0}+\epsilon\lambda_{2})^{\otimes m}.
\]
This settles~(\ref{eq:Psi2-smooth}) and completes the proof.
\end{proof}
\begin{proof}[Proof of Claim~\emph{\ref{claim:Psi1-Psi2-support}}.]
Recall from~(\ref{eq:Lambda01-mu0}) and~(\ref{eq:Lambda01-mui})
that $\supp\lambda_{0}=\{0\}$ and $\supp\lambda_{1}=\supp\lambda_{2}=\{1,2,\ldots,R\}.$
In this light, (\ref{eq:Psi1-pospart})\textendash (\ref{eq:Psi1-smooth})
imply
\begin{align*}
\supp(\pospart\Psi_{1}) & \subseteq(\MP_{m,R}^{*})^{-1}(0),\\
\supp(\negpart\Psi_{1}) & \subseteq(\MP_{m,R}^{*})^{-1}(1),\\
\supp(\Psi_{1}) & =(\MP_{m,R}^{*})^{-1}(0)\cup(\MP_{m,R}^{*})^{-1}(1),
\end{align*}
respectively. Analogously, (\ref{eq:Psi2-pospart})\textendash (\ref{eq:Psi2-smooth})
imply
\begin{align*}
\supp(\pospart\Psi_{2}) & \subseteq(\MP_{m,R}^{*})^{-1}(0),\\
\supp(\negpart\Psi_{2}) & \subseteq(\MP_{m,R}^{*})^{-1}(1),\\
\supp(\Psi_{2}) & =(\MP_{m,R}^{*})^{-1}(0)\cup(\MP_{m,R}^{*})^{-1}(1).
\end{align*}
Since the support of each $\Psi_{i}$ is the disjoint union of the
supports of its positive and negative parts, (\ref{eq:Psi12-pospart})
and~(\ref{eq:Psi12-negpart}) follow.
\end{proof}
\begin{proof}[Proof of Claim~\emph{\ref{claim:Psi1-Psi2-orthog}}.]
 Write $\Psi_{1}+\Psi_{2}=A+B+C,$ where
\begin{align*}
 & \hspace{-4mm}A=\left(\frac{1}{m+1}\lambda_{0}+\frac{m}{m+1}\lambda_{1}\right)^{\otimes m}-\left(\frac{1}{m+1}\lambda_{0}+\frac{m}{m+1}((1-\epsilon)\lambda_{0}+\epsilon\lambda_{2})\right)^{\otimes m},\\
 & \hspace{-4mm}B=2((1-\epsilon)\lambda_{0}+\epsilon\lambda_{2})^{\otimes m}-2\lambda_{1}^{\otimes m},\\
 & \hspace{-4mm}C=-2(-\epsilon\lambda_{0}+\epsilon\lambda_{2})^{\otimes m}.
\end{align*}
As a result, Proposition~\ref{prop:orth}\ref{item:orth-sum} guarantees
that
\begin{equation}
\orth(\Psi_{1}+\Psi_{2})\geq\min\{\orth A,\orth B,\orth C\}.\label{eq:orth-Psi1-Psi2-inproof}
\end{equation}
We have
\begin{align}
\orth A & \geq\orth\left(\left(\frac{1}{m+1}\lambda_{0}+\frac{m}{m+1}\lambda_{1}\right)\right.\nonumber \\
 & \qquad\qquad\qquad\qquad\qquad\left.-\left(\frac{1}{m+1}\lambda_{0}+\frac{m}{m+1}((1-\epsilon)\lambda_{0}+\epsilon\lambda_{2})\right)\right)\nonumber \\
 & =\orth\left(-\frac{m}{m+1}((1-\epsilon)\lambda_{0}+\epsilon\lambda_{2}-\lambda_{1})\right)\nonumber \\
 & \geq c'\sqrt{r},
\end{align}
where the first step uses Proposition~\ref{prop:orth}\ref{item:orth-difference-of-tensors},
and the last step is a restatement of~(\ref{eq:Lambda01-orth}).
Analogously,
\begin{align}
\orth B & \geq\orth(((1-\epsilon)\lambda_{0}+\epsilon\lambda_{2})-\lambda_{1})\nonumber \\
 & \geq c'\sqrt{r},
\end{align}
where the first and second steps use Proposition~\ref{prop:orth}\ref{item:orth-difference-of-tensors}
and~(\ref{eq:Lambda01-orth}), respectively. Finally,
\begin{align}
\orth C & =\orth((-\epsilon\lambda_{0}+\epsilon\lambda_{2})^{\otimes m})\nonumber \\
 & =m\orth(-\epsilon\lambda_{0}+\epsilon\lambda_{2})\nonumber \\
 & \geq m,\label{eq:orth-C-inproof}
\end{align}
where the second step applies Proposition~\ref{prop:orth}\ref{item:orth-tensor},
and the third step is valid because $\langle-\epsilon\lambda_{0}+\epsilon\lambda_{2},1\rangle=-\epsilon\langle\lambda_{0},1\rangle+\epsilon\langle\lambda_{2},1\rangle=-\epsilon+\epsilon=0.$
By~(\ref{eq:orth-Psi1-Psi2-inproof})\textendash (\ref{eq:orth-C-inproof}),
the proof is complete.
\end{proof}

\subsection{\label{subsec:An-amplification-theorem-for-smooth-thrdeg}An amplification
theorem for smooth threshold degree}

We have reached the technical centerpiece of our sign-rank analysis,
an amplification theorem for smooth threshold degree. This result
is qualitatively stronger than the amplification theorems for threshold
degree in Section~\ref{subsec:Hardness-amplification-for}, which
do not preserve smoothness. We prove the new amplification theorem
by manipulating locally smooth dual objects to achieve the desired
global behavior, an approach unrelated to our work in Section~\ref{subsec:Hardness-amplification-for}.
A detailed statement of our result follows.
\begin{thm}
\label{thm:min-smooth-amplification}There is an absolute constant
$C\geq1$ such that \\
\uline{for all:\mbox{$\rule[-4mm]{0mm}{9mm}$}}

positive integers $n,m,r,R,\theta$ with $R\geq r$ and $\theta\geq Cnm\log(2nm);$ 

real numbers $\gamma\in[0,1];$ 

functions $f\colon\zoon\to\zoo;$ 

probability distributions $\Lambda^{*}$ on $\{0,1,2,\ldots,R\}^{mn}|_{\leq\theta};$
and 

positive integers $d$ with
\begin{equation}
d\leq\frac{1}{C}\min\left\{ m\degthr(f,\gamma),\,\sqrt{r}\degthr(f,\gamma),\,\frac{\theta}{\sqrt{r}\log(2nmR)}\right\} ,\label{eq:sign-rank-def-d}
\end{equation}
\uline{one has:}
\begin{align}
 & \orth((-1)^{f\circ\MP_{m,R}^{*}}\cdot\Lambda)\geq d,\label{eq:srank-orthog}\\
 & \Lambda\geq\gamma\cdot(CnmR)^{-9d}\;\Lambda^{*}\label{eq:srank-min-smoothness}
\end{align}
for some $\Lambda\in\Distribution(\{0,1,2,\ldots,R\}^{mn}|_{\leq\theta})$.
\end{thm}

\begin{proof}
Let $0<c<1$ be the constant from Theorem~\ref{thm:dual-MP-smooth}.
Take $C\geq1/c$ to be a sufficiently large absolute constant. By
hypothesis,
\begin{equation}
\theta\geq Cnm\log(2nm).\label{eq:theta-large}
\end{equation}
Abbreviate
\begin{align}
X & =\{0,1,2,\ldots,R\}^{nm},\nonumber \\
\delta & =2^{-c\theta/(2\sqrt{r})}.\label{eq:srank-delta-def}
\end{align}
 The following inequalities are straightforward to verify:
\begin{align}
 & d<\frac{1}{3}\min\{\theta-nm,nmR\},\label{eq:srank-d-small}\\
 & \theta\geq\frac{8\e nm(1+\ln(nm))}{c},\label{eq:theta-above-minimal-threshold}\\
 & \frac{2^{3d+1}}{c^{4d+1}}\binom{nm+d}{d}^{3}\binom{nmR}{d}\frac{\delta}{1-\delta}<\frac{1}{2},\label{eq:srank-blowup-smaller-than-probability}\\
 & 2^{3d+1}\left(\frac{3m}{c}\right)^{4d+1}\binom{nm+d}{d}^{3}\binom{nmR}{d}\leq\frac{(CnmR)^{9d}}{4}.\label{eq:srank-bound-on-blowup}
\end{align}
For example,~(\ref{eq:srank-d-small}) holds because $d\leq nm/C$
by~(\ref{eq:sign-rank-def-d}) and $\theta\geq Cnm\log(2nm)$ by~(\ref{eq:theta-large}).
Inequalities~(\ref{eq:theta-above-minimal-threshold})\textendash (\ref{eq:srank-bound-on-blowup})
follow analogously from~(\ref{eq:sign-rank-def-d}) and~(\ref{eq:theta-large})
for a large enough constant $C$. The rest of the proof splits neatly
into four major steps.

\textsc{~}

\textsc{Step 1: Key distributions.} Theorem~\ref{thm:dual-MP-smooth}
provides probability distributions $\Lambda_{0}$ and $\Lambda_{1}$
such that
\begin{align}
 & \supp\Lambda_{i}=(\MP_{m,R}^{*})^{-1}(i),\qquad\qquad\qquad\qquad\qquad\qquad i=0,1,\label{eq:srank-Lambda-i-supp-restated}\\
 & \orth(\Lambda_{0}-\Lambda_{1})\geq\min\{m,c\sqrt{r}\},\label{eq:srank-Lambda-orth-restated}\\
 & \frac{\Lambda_{0}+\Lambda_{1}}{2}\in\SmoothFunction\left(\frac{m}{c},\{0,1,2,\ldots,R\}^{m}\right),\label{eq:srank-Lambda-smooth-restated}\\
 & \Lambda_{0},\Lambda_{1}\in\conv\left(\left\{ \lambda\in\Smooth\left(1,\frac{1}{c},1\right):\phantom{\frac{C}{(t+1)^{2}\,2^{t/(C\sqrt{r})}}}\right.\phantom{\left\{ \frac{1}{c(t+1)^{2}\,2^{-ct/\sqrt{r}}}\right\} ^{\otimes m}}\right.\nonumber \\
 & \quad\qquad\qquad\qquad\qquad\left.\left.\lambda(t)\leq\frac{1}{c(t+1)^{2}\,2^{ct/\sqrt{r}}}\text{ for }t\in\NN\right\} ^{\otimes m}\right).\label{eq:srank-Lambda-z-smooth-restated}
\end{align}
Consider the probability distributions
\begin{align*}
\Lambda_{z} & =\bigotimes_{i=1}^{n}\Lambda_{z_{i}}, &  & z\in\zoon.
\end{align*}
Then
\begin{align}
\Lambda_{z} & \in\conv\left(\left\{ \lambda\in\Smooth\left(1,\frac{1}{c},1\right):\lambda(t)\leq\frac{1}{c(t+1)^{2}\,2^{ct/\sqrt{r}}}\text{ for }t\in\NN\right\} ^{\otimes mn}\right)\nonumber \\
 & \subseteq\conv\left(\Smooth\left(1,\frac{1}{c},1\right)^{\otimes mn}\cap\phantom{\left\{ \frac{1}{c(t+1)^{2}\,2^{-ct/\sqrt{r}}}\right\} ^{\otimes m}}\right.\nonumber \\
 & \quad\qquad\qquad\qquad\left.\left\{ \lambda\in\Distribution(\NN):\lambda(t)\leq\frac{1}{c(t+1)^{2}\,2^{ct/\sqrt{r}}}\text{ for }t\in\NN\right\} ^{\otimes mn}\right)\nonumber \\
 & \subseteq\conv\left(\Smooth\left(1,\frac{1}{c},1\right)^{\otimes mn}\cap\left\{ \Lambda\in\Distribution(\NN^{mn}):\Lambda(\NN^{nm}|_{>\theta})\leq2^{-c\theta/(2\sqrt{r})}\right\} \right)\nonumber \\
 & \subseteq\conv\left(\Smooth\left(1,\frac{1}{c},1\right)^{\otimes mn}\cap\left\{ \Lambda\in\Distribution(\NN^{mn}):\Lambda(\NN^{nm}|_{>\theta})\leq\delta\right\} \right)\nonumber \\
 & \subseteq\conv\left(\Smooth\left(nm,\frac{1}{c},nm\right)\cap\{\Lambda\in\Distribution(\NN^{mn}):\Lambda(\NN^{nm}|_{>\theta})\leq\delta\}\right),\label{eq:srank-lambda-z-smooth-and-concentrated}
\end{align}
where the first step uses~(\ref{eq:tensor-conv-conv-tensor}) and~(\ref{eq:srank-Lambda-z-smooth-restated});
the third step is valid by~(\ref{eq:theta-above-minimal-threshold})
and Lemma~\ref{lem:concentration-of-measure}; the fourth step is
a substitution from~(\ref{eq:srank-delta-def}); and the last step
is an application of Proposition~\ref{prop:smooth-tensor}. 

\textsc{~}

\textsc{Step 2: Restricting the support.} By~(\ref{eq:srank-d-small}),
(\ref{eq:srank-lambda-z-smooth-and-concentrated}), and Corollary~\ref{cor:Zeroizer-for-distributions-convex-hull},
there is a real function $\tilde{\Lambda}_{z}\colon\NN^{nm}\to\Re$
such that
\begin{align}
 & \orth(\Lambda_{z}-\tilde{\Lambda}_{z})>d,\label{eq:srank-Lambda-tilde-Lambda-orth}\\
 & \supp\tilde{\Lambda}_{z}\subseteq\NN^{nm}|_{\leq\theta},\label{eq:tilde-Lambda-low-Hamming}\\
 & \supp\tilde{\Lambda}_{z}\subseteq\supp\Lambda_{z},\label{eq:tilde-Lambda-inside-support}
\end{align}
and
\[
|\Lambda_{z}-\tilde{\Lambda}_{z}|\leq\frac{2^{3d+1}}{c^{4d+1}}\binom{nm+d}{d}^{3}\binom{\diam(\supp\Lambda_{z})}{d}\frac{\delta}{1-\delta}\cdot\Lambda_{z}\qquad\text{on }\NN^{nm}|_{\leq\theta}.
\]
In view of (\ref{eq:srank-blowup-smaller-than-probability}) and~$\diam(\supp\Lambda_{z})\leq nmR$,
the last equation simplifies to
\begin{equation}
|\Lambda_{z}-\tilde{\Lambda}_{z}|\leq\frac{1}{2}\Lambda_{z}\quad\text{on }\NN^{nm}|_{\leq\theta}.\label{eq:Lambda-tilde-Lambda-approx}
\end{equation}
Properties~(\ref{eq:tilde-Lambda-low-Hamming}) and~(\ref{eq:Lambda-tilde-Lambda-approx})
imply that $\tilde{\Lambda}_{z}$ is a nonnegative function, which
along with~(\ref{eq:srank-Lambda-tilde-Lambda-orth}) and Proposition~\ref{prop:distribution-criterion}
implies that $\tilde{\Lambda}_{z}$ is a probability distribution.
Combining this fact with~(\ref{eq:srank-Lambda-i-supp-restated}),
(\ref{eq:tilde-Lambda-low-Hamming}), and (\ref{eq:tilde-Lambda-inside-support})
gives
\begin{equation}
\tilde{\Lambda}_{z}\in\Distribution\left(\NN^{nm}|_{\leq\theta}\cap\prod_{i=1}^{n}(\MP_{m,R}^{*})^{-1}(z_{i})\right),\qquad\qquad z\in\zoon.\label{eq:srank-tilde-Lambda-support}
\end{equation}
In particular, the $\tilde{\Lambda}_{z}$ are supported on disjoint
sets of inputs.

~

\textsc{Step 3: Ensuring min-smoothness.} Recall from (\ref{eq:srank-tilde-Lambda-support})
that each of the probability distributions $\tilde{\Lambda}_{z}$
is supported on a subset of $X|_{\leq\theta}.$ Consider the function
$\Phi\colon X|_{\leq\theta}\to\Re$ given by
\begin{align*}
\Phi & =2^{-n}\sum_{z\in\zoon}(-1)^{f(z)}\tilde{\Lambda}_{z}.
\end{align*}
Again by~(\ref{eq:srank-tilde-Lambda-support}), the support of $\tilde{\Lambda}_{z}$
is contained in $\prod_{i=1}^{n}(\MP_{m,R}^{*})^{-1}(z_{i}).$ This
means in particular that $f\circ\MP_{m,R}^{*}=f(z)$ on the support
of $\tilde{\Lambda}_{z}$, whence
\begin{equation}
(-1)^{f(z)}\tilde{\Lambda}_{z}=(-1)^{f\circ\MP_{m,R}^{*}}\cdot\tilde{\Lambda}_{z}\label{eq:srank-Lambda-z-f-on-support}
\end{equation}
everywhere on $X|_{\leq\theta}.$ Making this substitution in the
defining equation for $\Phi,$ we find that
\begin{equation}
(-1)^{f\circ\MP_{m,R}^{*}}\cdot\Phi\geq0.\label{eq:srank-Phi-sign}
\end{equation}
The fact that the $\tilde{\Lambda}_{z}$ are supported on pairwise
disjoint sets of inputs forces 
\begin{align}
 & |\Phi|=2^{-n}\sum_{z\in\zoon}\tilde{\Lambda}_{z}\label{eq:srank-abs-Phi}
\end{align}
and in particular
\begin{equation}
\|\Phi\|_{1}=1.\label{eq:srank-Phi-ell1}
\end{equation}

We now examine the smoothness of $\Phi.$ For this, consider the probability
distribution
\begin{align}
\Lambda & =2^{-n}\sum_{z\in\zoon}\Lambda_{z}.\label{eq:srank-lambda}
\end{align}
Comparing equations (\ref{eq:srank-abs-Phi}) and~(\ref{eq:srank-lambda})
term by term and using the upper bound (\ref{eq:Lambda-tilde-Lambda-approx}),
we find that $|\Lambda-|\Phi||\leq\frac{1}{2}\Lambda$ on $X|_{\leq\theta}$.
Equivalently,
\begin{equation}
\frac{1}{2}\Lambda\leq|\Phi|\leq\frac{3}{2}\Lambda\quad\text{on }X|_{\leq\theta.}\label{eq:srank-Lambda-tilde-Lambda-pointwise}
\end{equation}
But
\begin{align}
\Lambda & =\left(\frac{1}{2}\Lambda_{0}+\frac{1}{2}\Lambda_{1}\right)^{\otimes n}\nonumber \\
 & \in\SmoothFunction\left(\frac{m}{c},\{0,1,2,\ldots,R\}^{m}\right)^{\otimes n}\nonumber \\
 & \subseteq\SmoothFunction\left(\frac{m}{c},\{0,1,2,\ldots,R\}^{mn}\right),\label{eq:srank-Lambda-n-smooth}
\end{align}
where the last two steps are valid by~(\ref{eq:srank-Lambda-smooth-restated})
and Proposition~\ref{prop:smooth}\ref{enu:smooth-tensor}, respectively.
Combining (\ref{eq:srank-Lambda-tilde-Lambda-pointwise}) and~(\ref{eq:srank-Lambda-n-smooth}),
we conclude that $\Phi$ is $(3m/c)$-smooth on $X|_{\leq\theta}$.
As a result,~(\ref{eq:srank-d-small}) and Lemma~\ref{lem:SMOOTH-redistribute}
provide a function $\Phi^{*}\colon X|_{\leq\theta}\to\Re$ with 
\begin{align}
 & \orth(\Phi-\Phi^{*})>d,\label{eq:srank-smoothed-distribution-indistinguishable}\\
 & \|\Phi^{*}\|_{1}\leq2\|\Phi\|_{1},\\
 & \Phi\cdot\Phi^{*}\geq0,\label{eq:srank-smoothed-distribution-sign}\\
 & |\Phi^{*}|\geq\left(2^{3d+1}\left(\frac{3m}{c}\right)^{4d+1}\binom{nm+d}{d}^{3}\binom{\diam(\supp\Phi)}{d}\right)^{-1}\|\Phi\|_{1}\,\Lambda^{*}.\label{eq:srank-smoothed-distribution-factor}
\end{align}
In view of~(\ref{eq:srank-Phi-ell1}), the second property simplifies
to
\begin{equation}
\|\Phi^{*}\|_{1}\leq2.\label{eq:srank-smoothed-distribution-norm}
\end{equation}
Recall that on $X|_{\leq\theta},$ the function $\Phi$ is $(3m/c)$-smooth
and, by~(\ref{eq:srank-Phi-ell1}), is not identically zero. Therefore,
$\Phi$ must be nonzero at every point of $X|_{\leq\theta},$ which
includes the support of $\Phi^{*}.$ As a result, (\ref{eq:srank-Phi-sign})
and~(\ref{eq:srank-smoothed-distribution-sign}) imply that
\begin{equation}
(-1)^{f\circ\MP_{m,R}^{*}}\cdot\Phi^{*}\geq0.\label{eq:srank-Phi-star-sign}
\end{equation}
Finally, using~$\diam(\supp\Phi)\leq nmR$ along with the bounds~(\ref{eq:srank-bound-on-blowup})
and (\ref{eq:srank-Phi-ell1}), we can restate~(\ref{eq:srank-smoothed-distribution-factor})
as
\begin{equation}
|\Phi^{*}|\geq4(CnmR)^{-9d}\Lambda^{*}.\label{eq:srank-phi-star-min-smooth}
\end{equation}

\textsc{~}

\textsc{Step 4: The final construction.} By the definition of smooth
threshold degree, there is a probability distribution $\mu$ on $\zoon$
such that 
\begin{align}
 & \orth((-1)^{f}\cdot\mu)\geq\degthr(f,\gamma),\label{eq:srank-mu-orthog}\\
 & \mu(z)\geq\gamma\cdot2^{-n}, &  & z\in\zoon.\label{eq:srank-mu-minsmooth}
\end{align}
Define
\[
\Phi_{\text{final}}=\sum_{z\in\zoon}\mu(z)(-1)^{f(z)}\tilde{\Lambda}_{z}-\gamma\Phi+\gamma\Phi^{*}.
\]
The right-hand side is a linear combination of functions on $X|_{\leq\theta},$
whence
\begin{equation}
\supp(\Phi_{\text{final}})\subseteq X|_{\leq\theta}.\label{eq:srank-final-support}
\end{equation}
Moreover,
\begin{align}
\|\Phi_{\text{final}}\|_{1} & \leq\sum_{z\in\zoon}\mu(z)\|\tilde{\Lambda}_{z}\|_{1}+\gamma\|\Phi\|_{1}+\gamma\|\Phi^{*}\|_{1}\nonumber \\
 & \leq1+3\gamma\nonumber \\
 & \leq4,\label{eq:srank-final-ell1}
\end{align}
where the first step applies the triangle inequality, and the second
step uses~(\ref{eq:srank-tilde-Lambda-support}), (\ref{eq:srank-Phi-ell1})
and~(\ref{eq:srank-smoothed-distribution-norm}). Continuing,
\begin{align}
 & (-1)^{f\circ\MP_{m,R}^{*}}\cdot\Phi_{\text{final}}\nonumber \\
 & \qquad=(-1)^{f\circ\MP_{m,R}^{*}}\cdot\left(\sum_{z\in\zoon}(\mu(z)-\gamma2^{-n})(-1)^{f(z)}\tilde{\Lambda}_{z}+\gamma\Phi^{*}\right)\nonumber \\
 & \qquad=\!\sum_{z\in\zoon}\!(\mu(z)-\gamma2^{-n})(-1)^{f\circ\MP_{m,R}^{*}}\cdot(-1)^{f(z)}\tilde{\Lambda}_{z}+\gamma(-1)^{f\circ\MP_{m,R}^{*}}\cdot\Phi^{*}\nonumber \\
 & \qquad=\sum_{z\in\zoon}(\mu(z)-\gamma2^{-n})\tilde{\Lambda}_{z}+\gamma|\Phi^{*}|\label{eq:srank-final-minsmooth-intermediate}\\
 & \qquad\geq\gamma|\Phi^{*}|\nonumber \\
 & \qquad\geq4\gamma(CnmR)^{-9d}\Lambda^{*},\label{eq:srank-final-minsmooth}
\end{align}
where the first step applies the definition of $\Phi$; the third
step uses~(\ref{eq:srank-Lambda-z-f-on-support}) and~(\ref{eq:srank-Phi-star-sign});
the fourth step follows from~(\ref{eq:srank-mu-minsmooth}); and
the fifth step substitutes the lower bound from~(\ref{eq:srank-phi-star-min-smooth}).
Now
\begin{equation}
\Phi_{\text{final}}\not\equiv0\label{eq:srank-final-nonzero}
\end{equation}
follows from~(\ref{eq:srank-final-minsmooth-intermediate}) if $\gamma=0$,
and from~(\ref{eq:srank-final-minsmooth}) if $\gamma>0.$ 

It remains to examine the orthogonal content of $\Phi_{\text{final}}.$
For this, write
\begin{multline*}
\Phi_{\text{final}}=\sum_{z\in\zoon}\mu(z)(-1)^{f(z)}\Lambda_{z}+\sum_{z\in\zoon}\mu(z)(-1)^{f(z)}(\tilde{\Lambda}_{z}-\Lambda_{z})\\
+\gamma(\Phi^{*}-\Phi).
\end{multline*}
Then
\begin{align}
\orth(\Phi_{\text{final}}) & \geq\min\left\{ \orth\left(\sum_{z\in\zoon}\mu(z)(-1)^{f(z)}\Lambda_{z}\right),\right.\nonumber \\
 & \left.\phantom{\orth\left(\sum_{z\in\zoon}\right),}\min_{z}\{\orth(\tilde{\Lambda}_{z}-\Lambda_{z})\},\;\;\orth(\Phi^{*}-\Phi)\right\} \nonumber \\
 & \geq\min\left\{ \orth\left(\sum_{z\in\zoon}\mu(z)(-1)^{f(z)}\Lambda_{z}\right),d\right\} \nonumber \\
 & \geq\min\left\{ \orth\left(\sum_{z\in\zoon}\mu(z)(-1)^{f(z)}\bigotimes_{i=1}^{n}\Lambda_{z_{i}}\right),d\right\} \nonumber \\
 & \geq\min\left\{ \orth(\mu\cdot(-1)^{f})\orth(\Lambda_{1}-\Lambda_{0}),d\right\} \nonumber \\
 & \geq\min\{\degthr(f,\gamma)\min\{m,c\sqrt{r}\},d\}\nonumber \\
 & =d,\label{eq:srank-final-orth}
\end{align}
where the first step applies Proposition~\ref{prop:orth}\ref{item:orth-sum};
the second step follows from~(\ref{eq:srank-Lambda-tilde-Lambda-orth})
and~(\ref{eq:srank-smoothed-distribution-indistinguishable}); the
third step is valid by the definition of $\Lambda_{z}$; the fourth
step applies Corollary~\ref{cor:orth-block-composition}; the fifth
step substitutes the lower bounds from~(\ref{eq:srank-Lambda-orth-restated})
and~(\ref{eq:srank-mu-orthog}); and the final step uses~(\ref{eq:sign-rank-def-d}).

To complete the proof, let
\[
\Lambda=\frac{\Phi_{\text{final}}}{\|\Phi_{\text{final}}\|_{1}}\cdot(-1)^{f\circ\MP_{m,R}^{*}},
\]
where the right-hand side is well-defined by~(\ref{eq:srank-final-nonzero}).
Then $\|\Lambda\|_{1}=1$ by definition. Moreover, (\ref{eq:srank-final-support})
and~(\ref{eq:srank-final-minsmooth}) guarantee that $\Lambda$ is
a nonnegative function with support contained in $X|_{\leq\theta},$
so that $\Lambda\in\Distribution(X|_{\leq\theta}).$ The orthogonality
property~(\ref{eq:srank-orthog}) follows from~(\ref{eq:srank-final-orth}),
whereas the min-smoothness property~(\ref{eq:srank-min-smoothness})
follows from~(\ref{eq:srank-final-ell1}) and~(\ref{eq:srank-final-minsmooth}).
\end{proof}
\noindent We now translate the new amplification theorem from $\NN^{nm}|_{\leq\theta}$
to the hypercube, using the input transformation scheme of Theorem~\ref{thm:input-compression}. 
\begin{thm}
\label{thm:degthr-composition-min-smooth-Boolean-input}Let $C\geq1$
be the absolute constant from Theorem~\emph{\ref{thm:min-smooth-amplification}.}
Fix positive integers $n,m,\theta$ with $\theta\geq Cnm\log(2nm).$
Then there is an $($explicitly given$)$ transformation $H\colon\zoo^{6\theta\lceil\log(nm+1)\rceil}\to\zoon,$
computable by an AND-OR-AND circuit of polynomial size with bottom
fan-in at most $6\lceil\log(nm+1)\rceil,$ such that
\begin{align}
\degthr(f\circ H,\gamma\theta^{-30d}) & \geq d\lceil\log(nm+1)+1\rceil,\label{eq:min-smooth-composition}\\
\degthr(f\circ\neg H,\gamma\theta^{-30d}) & \geq d\lceil\log(nm+1)+1\rceil\label{eq:min-smooth-composition-negated}
\end{align}
for all Boolean functions $f\colon\zoon\to\zoo,$ all real numbers
$\gamma\in[0,1],$ and all positive integers
\[
d\leq\frac{1}{C}\min\left\{ m\degthr(f,\gamma),\frac{\theta}{4m\log(2\theta)}\right\} .
\]
\end{thm}

\begin{proof}
Negating a function's input bits has no effect on its $\gamma$-smooth
threshold degree for any $0\leq\gamma\leq1$, so that $f(x_{1},x_{2},\ldots,x_{n})$
and $f(\neg x_{1},\neg x_{2},\ldots,\neg x_{n})$ both have $\gamma$-smooth
threshold degree $\degthr(f,\gamma).$ Therefore, proving~(\ref{eq:min-smooth-composition})
for all $f$ will also settle~(\ref{eq:min-smooth-composition-negated})
for all $f.$ In what follows, we focus on the former.

Theorem~\ref{thm:input-compression} constructs an explicit surjection
$G\colon\zoo^{N}\to\NN^{nm}|_{\leq\theta}$ on $N=6\theta\lceil\log(nm+1)\rceil$
variables with the following two properties: 
\begin{enumerate}[topsep=3mm]
\item for every coordinate $i=1,2,\ldots,nm,$ the mapping $x\mapsto\OR_{\theta}^{*}(G(x)_{i})$
is computable by a DNF formula of size $(nm\theta)^{O(1)}=\theta^{O(1)}$
with bottom fan-in at most $6\lceil\log(nm+1)\rceil$;
\item for any polynomial $p,$ the map $v\mapsto\Exp_{G^{-1}(v)}p$ is a
polynomial on $\NN^{nm}|_{\leq\theta}$ of degree at most $(\deg p)/\lceil\log(nm+1)+1\rceil$.
\end{enumerate}
Consider the composition $F=(f\circ\MP_{m,\theta}^{*})\circ G.$ Then
\begin{align*}
F & =(f\circ(\AND_{m}\circ\OR_{\theta}^{*}))\circ G\\
 & =f\circ((\underbrace{\AND_{m}\circ\OR_{\theta}^{*},\ldots,\AND_{m}\circ\OR_{\theta}^{*}}_{n})\circ G),
\end{align*}
which by property~(i) of $G$ means that $F$ is the composition
of $f$ and an AND-OR-AND circuit $H$ of size $(nm\theta)^{O(1)}=\theta^{O(1)}$
and bottom fan-in at most $6\lceil\log(nm+1)\rceil$. Hence, the proof
will be complete once we show that
\begin{equation}
\degthr(F,\gamma\theta^{-30d})\geq d\lceil\log(nm+1)+1\rceil.\label{eq:degthr-composition-needed-bound-min-smooth}
\end{equation}

Define $r=m^{2}$ and $R=\max\{\theta,r\},$ and consider the probability
distribution on $\{0,1,2,\ldots,R\}^{nm}|_{\leq\theta}=\NN^{nm}|_{\leq\theta}$
given by $\Lambda^{*}(v)=|G^{-1}(v)|/2^{N}.$ Then Theorem~\ref{thm:min-smooth-amplification}
constructs a probability distribution $\Lambda$ on $\NN^{nm}|_{\leq\theta}$
such that
\begin{align}
 & \orth((-1)^{f\circ\MP_{m,R}^{*}}\cdot\Lambda)\geq d,\label{eq:srank-orthog-min-smooth-restated-R}\\
 & \Lambda\geq\gamma\theta^{-30d}\;\Lambda^{*}.\label{eq:srank-min-smoothness-min-smooth-restated}
\end{align}
In view of $R\geq\theta,$ inequality~(\ref{eq:srank-orthog-min-smooth-restated-R})
can be restated as
\begin{equation}
\orth((-1)^{f\circ\MP_{m,\theta}^{*}}\cdot\Lambda)\geq d.\label{eq:srank-orthog-min-smooth-restated}
\end{equation}
Define
\[
\lambda=\sum_{v\in\NN^{nm}|_{\leq\theta}}\Lambda(v)\cdot\frac{\1_{G^{-1}(v)}}{|G^{-1}(v)|},
\]
where $\1_{G^{-1}(v)}$ denotes as usual the characteristic function
of the set $G^{-1}(v).$ Clearly, $\lambda$ is a probability distribution
on $\zoo^{N}$. Moreover,
\begin{align}
\lambda & \geq\gamma\theta^{-30d}\sum_{v\in\NN^{nm}|_{\leq\theta}}\Lambda^{*}(v)\cdot\frac{\1_{G^{-1}(v)}}{|G^{-1}(v)|}\nonumber \\
 & =\gamma\theta^{-30d}\sum_{v\in\NN^{nm}|_{\leq\theta}}\frac{|G^{-1}(v)|}{2^{N}}\cdot\frac{\1_{G^{-1}(v)}}{|G^{-1}(v)|}\nonumber \\
 & =\gamma\theta^{-30d}\cdot\frac{\1_{\zoo^{N}}}{2^{N}},\label{eq:amplification-min-smooth-lambda-min-smooth}
\end{align}
where the first two steps use~(\ref{eq:srank-min-smoothness-min-smooth-restated})
and the definition of $\Lambda^{*},$ respectively. 

Finally, we examine the orthogonal content of $(-1)^{F}\cdot\lambda.$
Let $p\colon\Re^{N}\to\Re$ be any polynomial of degree less than
$d\lceil\log(nm+1)+1\rceil.$ Then by property~(ii) of $G,$ the
mapping $p^{*}\colon v\mapsto\Exp_{G^{-1}(v)}p$ is a polynomial on
$\NN^{nm}|_{\leq\theta}$ of degree less than $d.$ As a result, 
\begin{align*}
\langle(-1)^{F}\cdot\lambda,p\rangle & =\langle(-1)^{(f\circ\MP_{m,\theta}^{*})\circ G}\cdot\lambda,p\rangle\\
 & =\sum_{v\in\NN^{nm}|_{\leq\theta}}\,\sum_{G^{-1}(v)}(-1)^{(f\circ\MP_{m,\theta}^{*})\circ G}\cdot\lambda\cdot p\\
 & =\sum_{v\in\NN^{nm}|_{\leq\theta}}\,(-1)^{(f\circ\MP_{m,\theta}^{*})(v)}\sum_{G^{-1}(v)}\lambda\cdot p\\
 & =\sum_{v\in\NN^{nm}|_{\leq\theta}}\,(-1)^{(f\circ\MP_{m,\theta}^{*})(v)}\Lambda(v)\Exp_{G^{-1}(v)}p\\
 & =\langle(-1)^{f\circ\MP_{m,\theta}^{*}}\cdot\Lambda,p^{*}\rangle\\
 & =0,
\end{align*}
where the last step uses~(\ref{eq:srank-orthog-min-smooth-restated})
and $\deg p^{*}<d.$ We conclude that $\orth((-1)^{F}\cdot\lambda)\geq d\lceil\log(nm+1)+1\rceil,$
which along with~(\ref{eq:amplification-min-smooth-lambda-min-smooth})
settles~(\ref{eq:degthr-composition-needed-bound-min-smooth}).
\end{proof}

\subsection{\label{subsec:The-smooth-threshold-degree-AC0}The smooth threshold
degree of AC\protect\textsuperscript{0}}

By Theorem~\ref{thm:smooth-MP}, the composition $\AND_{n^{1/3}}\circ\OR_{n^{2/3}}$
has $\exp(-O(n^{1/3}))$-smooth threshold degree $\Omega(n^{1/3}).$
The objective of this section is to strengthen this result to a near-optimal
bound. For any $\epsilon>0,$ we will construct a constant-depth circuit
$f\colon\zoon\to\zoo$ with $\exp(-O(n^{1-\epsilon}))$-smooth threshold
degree $\Omega(n^{1-\epsilon}).$ This construction is likely to find
applications in future work, in addition to its use in this paper
to obtain a lower bound on the sign-rank of $\classAC^{0}$. The proof
proceeds by induction, with the amplification theorem for smooth threshold
degree (Theorem~\ref{thm:degthr-composition-min-smooth-Boolean-input})
applied repeatedly to construct increasingly harder circuits. To simplify
the exposition, we isolate the inductive step in the following lemma.
\begin{lem}
\label{lem:smooth-inductive-step}Let $f\colon\zoon\to\zoo$ be a
Boolean circuit of size $s,$ depth $d\geq0,$ and smooth threshold
degree
\begin{equation}
\degthr\left(f,\exp\left(-c'\cdot\frac{n^{1-\alpha}}{\log^{\beta}n}\right)\right)\geq c''\cdot\frac{n^{1-\alpha}}{\log^{\beta}n},\label{eq:f-smooth-thrdeg}
\end{equation}
for some constants $\alpha\in[0,1],$ $\beta\geq0,$ $c'>0,$ and
$c''>0.$ Then $f$ can be transformed in polynomial time into a Boolean
circuit $F\colon\zoo^{N}\to\zoo$ on $N=\Theta(n^{1+\alpha}\log^{2+\beta}n)$
variables that has size $s+N^{O(1)},$ depth at most $d+3,$ bottom
fan-in $O(\log N),$ and smooth threshold degree
\begin{equation}
\degthr\left(F,\exp\left(-C'\cdot\frac{N^{\frac{1}{1+\alpha}}}{\log^{\frac{1-\alpha+\beta}{1+\alpha}}N}\right)\right)\geq C''\cdot\frac{N^{\frac{1}{1+\alpha}}}{\log^{\frac{1-\alpha+\beta}{1+\alpha}}N},\label{eq:F-smooth-thrdeg}
\end{equation}
where $C',C''>0$ are constants that depend only on $\alpha,\beta,c',c''$.
Moreover, if $d\geq1$ and $f$ is monotone with AND gates at the
bottom, then the depth of $F$ is at most $d+2.$
\end{lem}

\begin{proof}
Let $C\geq1$ be the absolute constant from Theorem~\ref{thm:min-smooth-amplification}.
Apply Theorem~\ref{thm:degthr-composition-min-smooth-Boolean-input}
with
\begin{align*}
m & =\lceil n^{\alpha}\log^{\beta}n\rceil,\\
\theta & =\lceil Cnm\log(2nm)\rceil,\\
\gamma & =\exp\left(-c'\cdot\frac{n^{1-\alpha}}{\log^{\beta}n}\right)
\end{align*}
to obtain a function $H\colon\zoo^{N}\to\zoo^{n}$ on $N=\Theta(n^{1+\alpha}\log^{2+\beta}n)$
variables such that the composition $F=f\circ H$ satisfies~(\ref{eq:F-smooth-thrdeg})
for some $C',C''>0$ that depend only on $\alpha,\beta,c',c'',$ and
furthermore $H$ is computable by an AND-OR-AND circuit of polynomial
size and bottom fan-in $O(\log N)$. Clearly, the composition $F=f\circ H$
is a circuit of size $s+N^{O(1)}$, depth at most $d+3,$ and bottom
fan-in $O(\log N)$. Moreover, if $d\geq1$ and the circuit for $f$
is monotone with AND gates at the bottom level, then the bottom level
of $f$ can be merged with the top level of $H$ to reduce the depth
of $F=f\circ H$ to at most $(d+3)-1=d+2.$
\end{proof}
\begin{cor}
\label{cor:smooth-inductive-step}Fix constants $\alpha\in[0,1],$
$\beta\geq0,$ $c'>0,$ $c''>0,$ and $d\geq0.$ Let $\{f_{n}\}_{n=1}^{\infty}$
be a Boolean circuit family in which $f_{n}\colon\zoon\to\zoo$ has
polynomial size, depth at most $d,$ and smooth threshold degree
\begin{equation}
\degthr\left(f_{n},\exp\left(-c'\cdot\frac{n^{1-\alpha}}{\log^{\beta}n}\right)\right)\geq c''\cdot\frac{n^{1-\alpha}}{\log^{\beta}n}\label{eq:f-smooth-thrdeg-1}
\end{equation}
for all $n\geq2.$ Then there is a Boolean circuit family $\{F_{N}\}_{N=1}^{\infty}$
in which $F_{N}\colon\zoo^{N}\to\zoo$ has polynomial size, depth
at most $d+3,$ bottom fan-in $O(\log N),$ and smooth threshold degree
\begin{equation}
\degthr\left(F_{N},\exp\left(-C'\cdot\frac{N^{\frac{1}{1+\alpha}}}{\log^{\frac{1-\alpha+\beta}{1+\alpha}}N}\right)\right)\geq C''\cdot\frac{N^{\frac{1}{1+\alpha}}}{\log^{\frac{1-\alpha+\beta}{1+\alpha}}N}\label{eq:F-smooth-thrdeg-1}
\end{equation}
for all $N\geq2,$ where $C',C''>0$ are constants that depend only
on $\alpha,\beta,c',c''$. Moreover, if $d\geq1$ and each $f_{n}$
is monotone with AND gates at the bottom, then the depth of each $F_{N}$
is at most $d+2.$
\end{cor}

\begin{proof}
It suffices to construct $F_{N}$ for $N$ larger than a constant
of our choosing. The first few circuits of the family $\{F_{N}\}_{N=1}^{\infty}$
can then be taken to be the ``dictator'' functions $x\mapsto x_{1},$
each of which has $1/2$-smooth threshold degree~$1$ by Fact~\ref{fact:fact-degthr-min-smooth}. 

Let $c=c(\alpha,\beta,c',c'')>0$ be a sufficiently small constant.
For any given integer $N$ larger than a certain constant, apply Lemma~\ref{lem:smooth-inductive-step}
to $f_{\lfloor(cN/\log^{2+\beta}N)^{1/(1+\alpha)}\rfloor}$ to obtain
a circuit $F$ on at most $N$ variables with smooth threshold degree~(\ref{eq:F-smooth-thrdeg}),
for some constants $C',C''>0$ that depend only on $\alpha,\beta,c',c''$.
Then Proposition~\ref{prop:smooth-degthr-padding} forces~(\ref{eq:F-smooth-thrdeg-1})
for the circuit $F_{N}\colon\zoo^{N}\to\zoo$ given by $F_{N}(x_{1},x_{2},\ldots,x_{N})=F(x_{1},x_{2},\ldots).$
The claims regarding the size, depth, and fan-ins of each $F_{N}$
follow directly from Lemma~\ref{lem:smooth-inductive-step}.
\end{proof}
We now obtain our lower bounds on the smooth threshold degree of $\classAC^{0}$.
We present two incomparable theorems here, both of which apply Corollary~\ref{cor:smooth-inductive-step}
in a recursive manner but with different base cases.
\begin{thm}
\label{thm:degthr-ac0-min-smooth-3k}Let $k\geq0$ be a given integer.
Then there is an $($explicitly given$)$ Boolean circuit family $\{f_{n}\}_{n=1}^{\infty},$
where $f_{n}\colon\zoon\to\zoo$ has polynomial size, depth~$3k,$
bottom fan-in $O(\log n),$ and smooth threshold degree
\begin{equation}
\degthr\left(f_{n},\exp\left(-c'\cdot\frac{n^{1-\frac{1}{k+1}}}{\log^{\frac{k(k-1)}{2(k+1)}}n}\right)\right)\geq c''\cdot\frac{n^{1-\frac{1}{k+1}}}{\log^{\frac{k(k-1)}{2(k+1)}}n}\label{eq:degthr-ac0-min-smooth}
\end{equation}
for some constants $c',c''>0$ and all $n\geq2.$
\end{thm}

\begin{proof}
The proof is by induction on $k.$ The base case $k=0$ corresponds
to the family of ``dictator'' functions $x\mapsto x_{1},$ each
of which has $1/2$-smooth threshold degree~$1$ by Fact~\ref{fact:fact-degthr-min-smooth}.
For the inductive step, fix an explicit circuit family $\{f_{n}\}_{n=1}^{\infty}$
in which $f_{n}\colon\zoon\to\zoo$ has polynomial size, depth~$3k,$
and smooth threshold degree~(\ref{eq:degthr-ac0-min-smooth}) for
some constants $c',c''>0$. Then taking $\alpha=\frac{1}{k+1}$ and
$\beta=\frac{k(k-1)}{2(k+1)}$ in Corollary~\ref{cor:smooth-inductive-step}
produces an explicit circuit family $\{F_{n}\}_{n=1}^{\infty}$ in
which $F_{n}\colon\zoon\to\zoo$ has polynomial size, depth~$3k+3=3(k+1),$
bottom fan-in $O(\log n),$ and smooth threshold degree 
\[
\degthr\left(F_{n},\exp\left(-C'\cdot\frac{n^{\frac{k+1}{k+2}}}{\log^{\frac{k(k+1)}{2(k+2)}}n}\right)\right)\geq C''\cdot\frac{n^{\frac{k+1}{k+2}}}{\log^{\frac{k(k+1)}{2(k+2)}}n}
\]
for some constants $C',C''>0$ and all $n\geq2.$
\end{proof}
\begin{thm}
\label{thm:degthr-ac0-min-smooth-3k-plus-1}Let $k\geq1$ be a given
integer. Then there is an $($explicitly given$)$ Boolean circuit
family $\{f_{n}\}_{n=1}^{\infty},$ where $f_{n}\colon\zoon\to\zoo$
has polynomial size, depth~$3k+1,$ bottom fan-in $O(\log n),$ and
smooth threshold degree
\begin{equation}
\degthr\left(f_{n},\exp\left(-c'\cdot\frac{n^{1-\frac{2}{2k+3}}}{\log^{\frac{k^{2}}{2k+3}}n}\right)\right)\geq c''\cdot\frac{n^{1-\frac{2}{2k+3}}}{\log^{\frac{k^{2}}{2k+3}}n}\label{eq:degthr-ac0-min-smooth-3}
\end{equation}
for some constants $c',c''>0$ and all $n\geq2.$
\end{thm}

\begin{proof}
As with Theorem~\ref{thm:degthr-ac0-min-smooth-3k}, the proof is
by induction on $k.$ For the base case $k=1,$ consider the family
$\{g_{n}\}_{n=1}^{\infty}$ in which $g_{n}\colon\zoon\to\zoo$ is
given by 
\[
g_{n}(x)=\bigvee_{i=1}^{\lfloor n^{1/3}\rfloor}\bigwedge_{j=1}^{\lfloor n^{2/3}\rfloor}x_{i,j}.
\]
 Then
\begin{align*}
\degthr(g_{n},12^{-\lfloor n^{1/3}\rfloor-1}) & =\degthr(\OR_{\lfloor n^{1/3}\rfloor}\circ\AND_{\lfloor n^{2/3}\rfloor},12^{-\lfloor n^{1/3}\rfloor-1})\\
 & =\degthr(\MP_{\lfloor n^{1/3}\rfloor,\lfloor n^{2/3}\rfloor},12^{-\lfloor n^{1/3}\rfloor-1})\\
 & =\Omega(n^{1/3}),
\end{align*}
where the first step uses Proposition~\ref{prop:smooth-degthr-padding};
the second step is valid because a function's smooth threshold degree
remains unchanged when one negates the function or its input variables;
and the last step follows from Theorem~\ref{thm:smooth-MP}. Applying
Corollary~\ref{cor:smooth-inductive-step} to the circuit family
$\{g_{n}\}_{n=1}^{\infty}$ with $\alpha=2/3$ and $\beta=0$ yields
an explicit circuit family $\{G_{n}\}_{n=1}^{\infty}$ in which $G_{n}\colon\zoon\to\zoo$
has polynomial size, depth $2+2=4,$ bottom fan-in $O(\log n),$ and
smooth threshold degree
\[
\degthr\left(G_{n},\exp\left(-C'\cdot\frac{n^{3/5}}{\log^{1/5}n}\right)\right)\geq C''\cdot\frac{n^{3/5}}{\log^{1/5}n}
\]
for some constants $C',C''>0$ and all $n\geq2$. This new circuit
family $\{G_{n}\}_{n=1}^{\infty}$ establishes the base case.

For the inductive step, fix an integer $k\geq1$ and an explicit circuit
family $\{f_{n}\}_{n=1}^{\infty}$ in which $f_{n}\colon\zoon\to\zoo$
has polynomial size, depth $3k+1,$ and smooth threshold degree~(\ref{eq:degthr-ac0-min-smooth-3})
for some constants $c',c''>0.$ Applying Corollary~\ref{cor:smooth-inductive-step}
with $\alpha=2/(2k+3)$ and $\beta=k^{2}/(2k+3)$ yields an explicit
circuit family $\{F_{n}\}_{n=1}^{\infty}$, where $F_{n}\colon\zoon\to\zoo$
has polynomial size, depth $(3k+1)+3=3(k+1)+1,$ bottom fan-in $O(\log n),$
and smooth threshold degree
\[
\degthr\left(F_{n},\exp\left(-C'''\cdot\frac{n^{\frac{2k+3}{2k+5}}}{\log^{\frac{(k+1)^{2}}{2k+5}}n}\right)\right)\geq C''''\cdot\frac{n^{\frac{2k+3}{2k+5}}}{\log^{\frac{(k+1)^{2}}{2k+5}}n}
\]
for some constants $C''',C''''>0$ and all $n\geq2$. This completes
the inductive step.
\end{proof}

\subsection{\label{subsec:The-sign-rank-of-AC0}The sign-rank of AC\protect\textsuperscript{0}}

We have reached our main result on the sign-rank and unbounded-error
communication complexity of constant-depth circuits. The proof amounts
to lifting, by means of Theorem~\ref{thm:thrdeg-to-sign-rank}, the
lower bounds on smooth threshold degree in Theorems~\ref{thm:degthr-ac0-min-smooth-3k}
and~\ref{thm:degthr-ac0-min-smooth-3k-plus-1} to sign-rank lower
bounds. 
\begin{thm}
\label{thm:sign-rank-ac0-3k}Let $k\geq1$ be a given integer. Then
there is an $($explicitly given$)$ Boolean circuit family $\{F_{n}\}_{n=1}^{\infty},$
where $F_{n}\colon\zoon\times\zoon\to\zoo$ has polynomial size, depth
$3k,$ bottom fan-in $O(\log n),$ sign-rank
\begin{equation}
\srank(F_{n})=\exp\left(\Omega\left(n^{1-\frac{1}{k+1}}\cdot(\log n)^{-\frac{k(k-1)}{2(k+1)}}\right)\right),\label{eq:Fn-sign-rank}
\end{equation}
and unbounded-error communication complexity
\begin{equation}
\upp(F_{n})=\Omega\left(n^{1-\frac{1}{k+1}}\cdot(\log n)^{-\frac{k(k-1)}{2(k+1)}}\right).\label{eq:Fn-upp}
\end{equation}
\end{thm}

\begin{proof}
Theorem~\ref{thm:degthr-ac0-min-smooth-3k} constructs a circuit
family $\{f_{n}\}_{n=1}^{\infty}$ in which $f_{n}\colon\zoon\to\zoo$
has polynomial size, depth $3k,$ bottom fan-in $O(\log n),$ and
smooth threshold degree~(\ref{eq:degthr-ac0-min-smooth}) for some
constants $c',c''>0$ and all $n\geq2.$ Abbreviate $m=2\lceil\exp(4c'/c'')\rceil.$
For any $n\geq m,$ define $F_{n}=f_{\lfloor n/m\rfloor}\circ\OR_{m}\circ\AND_{2}$.
Then~(\ref{eq:Fn-sign-rank}) is immediate from~(\ref{eq:degthr-ac0-min-smooth})
and Theorem~\ref{thm:thrdeg-to-sign-rank}. Combining~(\ref{eq:Fn-sign-rank})
with Theorem~\ref{thm:srank-upp} settles~(\ref{eq:Fn-upp}).

It remains to analyze the circuit complexity of $F_{n}.$ We defined
$F_{n}$ formally as a circuit of depth~$3k+2$ in which the bottom
four levels have fan-ins $n^{O(1)},$ $O(\log n),$ $m,$ and $2,$
in that order. Since $m$ is a constant independent of $n$, these
four levels can be computed by a circuit of polynomial size, depth~$2$,
and bottom fan-in $O(\log n).$ This optimization reduces the depth
of $F_{n}$ to~$(3k+2)-4+2=3k$ while keeping the bottom fan-in at
$O(\log n).$
\end{proof}
\noindent We now similarly lift Theorem~\ref{thm:degthr-ac0-min-smooth-3k-plus-1}
to a lower bound on sign-rank and unbounded-error communication complexity.
\begin{thm}
\label{thm:sign-rank-ac0-3k-plus-1}Let $k\geq1$ be a given integer.
Then there is an $($explicitly given$)$ Boolean circuit family $\{F_{n}\}_{n=1}^{\infty},$
where $F_{n}\colon\zoon\times\zoon\to\zoo$ has polynomial size, depth
$3k+1,$ bottom fan-in $O(\log n),$ sign-rank
\[
\srank(F_{n})=\exp\left(\Omega\left(n^{1-\frac{2}{2k+3}}\cdot(\log n)^{-\frac{k^{2}}{2k+3}}\right)\right),
\]
and unbounded-error communication complexity
\[
\upp(F_{n})=\Omega\left(n^{1-\frac{2}{2k+3}}\cdot(\log n)^{-\frac{k^{2}}{2k+3}}\right).
\]
\end{thm}

\begin{proof}
The proof is analogous to that of Theorem~\ref{thm:sign-rank-ac0-3k},
with the only difference that the appeal to Theorem~\ref{thm:degthr-ac0-min-smooth-3k}
should be replaced with an appeal to Theorem~\ref{thm:degthr-ac0-min-smooth-3k-plus-1}.
\end{proof}
\noindent Theorems~\ref{thm:sign-rank-ac0-3k} and~\ref{thm:sign-rank-ac0-3k-plus-1}
settle Theorems~\ref{thm:MAIN-sign-rank-3k},~\ref{thm:MAIN-sign-rank-3k-plus-1},
and~\ref{thm:MAIN-unbounded} in the introduction.

\section*{Acknowledgments}

The authors are thankful to Mark Bun and Justin Thaler for valuable
comments on an earlier version of this paper.\bibliographystyle{siamplain}
\bibliography{refs}

\appendix

\section{\label{sec:Sign-rank-and-smooth-thrdeg}Sign-rank and smooth threshold
degree}

The purpose of this appendix is to prove Theorem~\ref{thm:thrdeg-to-sign-rank},
implicit in~\cite{sherstov07symm-sign-rank,RS07dc-dnf}. We closely
follow the treatment in those earlier papers. Sections~\ref{subsec:Fourier-transform}\textendash \ref{subsec:Pattern-matrices}
cover necessary technical background, followed by the proof proper
in Section~\ref{subsec:Sign-rank-and-smooth}.

\subsection{\label{subsec:Fourier-transform}Fourier transform}

Consider the real vector space of functions $\zoon\to\Re.$ For $S\subseteq\oneton,$
define $\chi_{S}\colon\zoon\to\{-1,+1\}$ by $\chi_{S}(x)=(-1)^{\sum_{i\in S}x_{i}}.$
Then 
\[
\langle\chi_{S},\chi_{T}\rangle=\begin{cases}
2^{n} & \text{if }S=T,\\
0 & \text{otherwise.}
\end{cases}
\]
Thus, $\{\chi_{S}\}_{S\subseteq\oneton}$ is an orthogonal basis for
the vector space in question. In particular, every function $\phi\colon\zoon\to\Re$
has a unique representation of the form 
\begin{align*}
\phi=\sum_{S\subseteq\oneton}\hat{\phi}(S)\chi_{S}
\end{align*}
for some reals $\hat{\phi}(S),$ where by orthogonality $\hat{\phi}(S)=2^{-n}\langle\phi,\chi_{S}\rangle$.
The reals $\hat{\phi}(S)$ are called the \emph{Fourier coefficients
of $\phi,$} and the mapping $\phi\mapsto\hat{\phi}$ is the \emph{Fourier
transform of $f.$} The following fact is immediate from the definition
of $\hat{\phi}(S).$ 
\begin{prop}
\label{prop:fourier-coeff-bound} Let $\phi\colon\zoon\to\Re$ be
given. Then 
\[
\max_{S\subseteq\{1,2,\ldots,n\}}|\hat{\phi}(S)|\leq2^{-n}\|\phi\|_{1}.
\]
\end{prop}

The linear subspace of real polynomials on $\zoon$ of degree at most
$d$ is easily seen to be $\Span\{\chi_{S}:|S|\leq d\}.$ Its orthogonal
complement, $\Span\{\chi_{S}\colon|S|>d\}$, is then the linear subspace
of functions that have zero inner product with every polynomial of
degree at most $d.$ As a result, the orthogonal content of a nonzero
function $\phi\colon\zoon\to\Re$ is given by
\begin{align}
\orth\phi & =\min\{|S|:\hat{\phi}(S)\ne0\}, &  & \phi\not\equiv0.\label{eq:orth-fourier}
\end{align}

\subsection{Forster's bound}

The \emph{spectral norm} of a real matrix $A=[A_{xy}]_{x\in X,y\in Y}$
is given by 
\[
\|A\|=\max_{v\in\Re^{|Y|},\;\|v\|_{2}=1}\|Av\|_{2},
\]
where $\|\cdot\|_{2}$ is the Euclidean norm on vectors. The first
strong lower bound on the sign-rank of an explicit matrix was obtained
by Forster~\cite{forster02linear}, who proved that 
\[
\srank(A)\geq\frac{\sqrt{|X|\,|Y|}}{\|A\|}
\]
for any matrix $A=[A_{xy}]_{x\in X,y\in Y}$ with $\pm1$ entries.
Forster's result has seen a number of generalizations, including the
following theorem~\cite[Theorem~3]{forster01relations}.
\begin{thm}[Forster et al.]
\label{thm:forster} Let $A=[A_{xy}]_{x\in X,y\in Y}$ be a real
matrix without zero entries. Then 
\[
\srank(A)\geq\frac{\sqrt{|X|\,|Y|}}{\|A\|}\;\min_{x,y}|A_{xy}|.
\]
\end{thm}

\subsection{\label{subsec:Pattern-matrices}Spectral norm of pattern matrices}

\emph{Pattern matrices} were introduced in~\cite{sherstov07ac-majmaj,sherstov07quantum}
and proved useful in obtaining strong lower bounds on communication
complexity. Relevant definitions and results from~\cite{sherstov07quantum}
follow. Let $n$ and $N$ be positive integers with $n\mid N.$ Partition
$\{1,2,\ldots,N\}$ into $n$ contiguous blocks, each with $N/n$
elements: 
\begin{multline*}
\{1,2,\ldots,N\}=\left\{ 1,2,\dots,\frac{N}{n}\right\} \cup\left\{ \frac{N}{n}+1,\dots,\frac{2N}{n}\right\} \\
\cup\cdots\cup\left\{ \frac{(n-1)N}{n}+1,\dots,N\right\} .
\end{multline*}
Now, let $\VV(N,n)$ denote the family of subsets $V\subseteq\{1,2,\ldots,N\}$
that have exactly one element in each of these blocks (in particular,
$|V|=n$). Clearly, $|\VV(N,n)|=(N/n)^{n}.$ For a function $\phi\colon\zoo^{n}\to\Re,$
the \emph{$(N,n,\phi)$-pattern matrix} is the real matrix $A$ given
by
\[
A=\Big[\phi(x|_{V}\oplus w)\Big]_{x\in\zoo^{N},\,(V,w)\in\VV(N,n)\times\zoo^{n}}\;.
\]
In words, $A$ is the matrix of size $2^{N}$~by~$(N/n)^{n}2^{n}$
whose rows are indexed by strings $x\in\zoo^{N},$ whose columns are
indexed by pairs $(V,w)\in\VV(N,n)\times\zoo^{n},$ and whose entries
are given by $A_{x,(V,w)}=\phi(x|_{V}\oplus w).$ We will need the
following expression for the spectral norm of a pattern matrix~\cite[Theorem~4.3]{sherstov07quantum}.
\begin{thm}[Sherstov]
\label{thm:pattern-spectrum}Let $\phi\colon\zoo^{n}\to\Re$ be given.
Let $A$ be the $(N,n,\phi)$-pattern matrix. Then 
\[
\|A\|\;=\;\sqrt{2^{N+n}\left(\frac{N}{n}\right)^{n}}\;\max_{S\subseteq\{1,2,\ldots,n\}}\left\{ |\hat{\phi}(S)|\left(\frac{n}{N}\right)^{|S|/2}\right\} .
\]
\end{thm}

\subsection{\label{subsec:Sign-rank-and-smooth}Proof of Theorem~\ref{thm:thrdeg-to-sign-rank}}

We are now in a position to prove Theorem~\ref{thm:thrdeg-to-sign-rank}.
We will derive it from the following more general result, stated in
terms of pattern matrices.
\begin{thm}
\label{thm:sign-rank-pm}Let $f\colon\zoon\to\zoo$ be given. Suppose
that $\degthr(f,\gamma)\geq d,$ where $\gamma$ and $d$ are positive
reals. Then for any integer $T\geq1,$ the $(Tn,n,(-1)^{f})$-pattern
matrix has sign-rank at least $\gamma T^{d/2}$.
\end{thm}

\begin{proof}
By the definition of smooth threshold degree, there is a probability
distribution $\mu$ on $\zoon$ such that 
\begin{align}
 & \mu(x)\geq\gamma\,2^{-n}, &  & x\in\zoon,\label{eq:srank-mu-smooth}\\
 & \orth((-1)^{f}\cdot\mu)\geq d.\label{eq:srank-mu-orth}
\end{align}
Abbreviate $\phi=(-1)^{f}\cdot\mu.$ Let $F$ and $\Phi$ denote the
$(Tn,n,(-1)^{f})$- and $(Tn,n,\phi)$-pattern matrices, respectively.
By~(\ref{eq:orth-fourier}) and~(\ref{eq:srank-mu-orth}),
\begin{align}
\hat{\phi}(S) & =0, &  & |S|<d.\label{eq:srank-phi-low-order}
\end{align}
The remaining Fourier coefficients of $\phi$ can be bounded using
Proposition~\ref{prop:fourier-coeff-bound}:
\begin{align}
|\hat{\phi}(S)| & \leq2^{-n}, &  & S\subseteq\oneton.\label{eq:srank-phi-high-order}
\end{align}
Now 
\begin{align*}
\srank(F) & =\srank(\Phi)\\
 & \geq\frac{\sqrt{2^{Tn+n}\,T^{n}}}{\|\Phi\|}\cdot\gamma\,2^{-n}\\
 & =\frac{\gamma\,2^{-n}}{\max_{S}\{|\hat{\phi}(S)|\,T^{-|S|/2}\}}\\
 & \geq\gamma T^{d/2},
\end{align*}
where the first step is valid because $F$ and $\Phi$ have the same
sign pattern; the second step uses~(\ref{eq:srank-mu-smooth}) and
Theorem~\ref{thm:forster}; the third step applies Theorem~\ref{thm:pattern-spectrum};
and the final step substitutes the upper bounds from~(\ref{eq:srank-phi-low-order})
and~(\ref{eq:srank-phi-high-order}).
\end{proof}
We have reached the main result of this appendix.
\begin{thm*}[restatement of Theorem~\ref{thm:thrdeg-to-sign-rank}]
Let $f\colon\zoon\to\zoo$ be given. Suppose that $\degthr(f,\gamma)\geq d,$
where $\gamma$ and $d$ are positive reals. Fix an integer $m\geq2$
and define $F\colon\zoo^{mn}\times\zoo^{mn}\to\zoo$ by $F(x,y)=f\circ\OR_{m}\circ\AND_{2}.$
Then
\[
\srank(F)\geq\gamma\left\lfloor \frac{m}{2}\right\rfloor ^{d/2}.
\]
\end{thm*}
\begin{proof}
The result is immediate from Theorem~\ref{thm:sign-rank-pm} since
the $(\lfloor m/2\rfloor n,n,(-1)^{f})$-pattern matrix is a submatrix
of $[(-1)^{F(x,y)}]_{x,y}.$
\end{proof}

\section{\label{sec:dual-OR}A dual object for OR}

The purpose of this appendix is to prove Theorem~\ref{thm:dual-OR},
which gives a dual polynomial for the OR function with a number of
additional properties. The treatment here closely follows earlier
work by Špalek~\cite{spalek08dual-or}, Bun and Thaler~\cite{bun-thaler13and-or-tree,bun-thaler17adeg-ac0,BKT17poly-strikes-back},
and Sherstov~\cite{sherstov14sign-deg-ac0,sherstov15asymmetry}.
We start with a well-known binomial identity~\cite{concrete-mathematics}.
\begin{fact}
\label{fact:binom-orthog}For every univariate polynomial $p$ of
degree less than $n,$
\[
\sum_{t=0}^{n}(-1)^{t}{n \choose t}p(t)=0.
\]
\end{fact}

The next lemma constructs a dual polynomial for OR that has the sign
behavior claimed in Theorem~\ref{thm:dual-OR} but may lack some
of the metric properties. The lemma is an adaptation of~\cite[Lemma~A.2]{sherstov14sign-deg-ac0}.
\begin{lem}
\label{lem:construction-omega-Big-weight-at-zero}Let $\epsilon$
be given, $0<\epsilon<1$. Then for some constant $c=c(\epsilon)\in(0,1)$
and every integer $n\geq1,$ there is an $($explicitly given$)$
function $\omega\colon\{0,1,2,\dots,n\}\to\Re$ such that
\begin{align}
 & \omega(0)>\frac{1-\epsilon}{2}\cdot\|\omega\|_{1},\label{eq:construction-omega-at-zero}\\
 & |\omega(t)|\leq\frac{1}{ct^{2}\,2^{ct/\sqrt{n}}}\cdot\|\omega\|_{1} &  & (t=1,2,\ldots,n),\label{eq:construction-omega-upper-bound}\\
 & (-1)^{t}\omega(t)\geq0 &  & (t=0,1,2,\dots,n),\label{eq:construction-omega-sign}\\
 & \orth\omega\geq c\sqrt{n}.\label{eq:construction-omega-orthog}
\end{align}
\end{lem}

\begin{rem}
\noindent \label{rem:monotonicity}It is helpful to keep in mind that
properties~(\ref{eq:construction-omega-at-zero})\textendash (\ref{eq:construction-omega-orthog})
are logically monotonic in $c.$ In other words, establishing these
properties for a given constant $c>0$ also establishes them for all
smaller positive constants.
\end{rem}

\begin{proof}[Proof of Lemma~\emph{\ref{lem:construction-omega-Big-weight-at-zero}.}]
Let $\Delta=8\lceil1/\epsilon\rceil+3.$ If $n\leq\Delta,$ the requirements
of the lemma hold for the function $\omega:(0,1,2,3,\ldots,n)\mapsto(1,-1,0,0,\ldots,0)$
and all $c\in(0,1/\Delta].$ In what follows, we treat the complementary
case $n>\Delta.$ 

Define $d=\lfloor\sqrt{n/\Delta}\rfloor$ and let $S=\{1,\frac{\Delta+1}{2}\}\cup\{i^{2}\Delta:i=0,1,2,\dots,d\},$
so that $S\subseteq\{0,1,2,\dots,n\}.$ Consider the function $\omega\colon\{0,1,2,\dots,n\}\to\Re$
given by
\[
\omega(t)=\frac{(-1)^{n+t+|S|+1}}{n!}{n \choose t}\prod_{\substack{i=0,1,2,\dots,n:\\
i\notin S
}
}(t-i).
\]
Fact~\ref{fact:binom-orthog} implies that
\begin{align}
\orth\omega & >d+1\nonumber \\
 & \geq\sqrt{\frac{n}{\Delta}}.\label{eq:omega-orth}
\end{align}
A routine calculation reveals that
\begin{equation}
\omega(t)=\begin{cases}
(-1)^{|\{i\in S:i<t\}|}\prod_{i\in S\setminus\{t\}}\frac{1}{|t-i|} & \text{if \ensuremath{t\in S,}}\\
0 & \text{otherwise.}
\end{cases}\label{eq:construction-omega-cases}
\end{equation}
It follows that 
\begin{align}
\frac{\omega(0)}{|\omega(1)|} & =\frac{\Delta-1}{\Delta+1}\prod_{i=1}^{d}\frac{i^{2}\Delta-1}{i^{2}\Delta}\nonumber \\
 & \geq1-\frac{2}{\Delta+1}-\sum_{i=1}^{d}\frac{1}{i^{2}\Delta}\nonumber \\
 & >1-\frac{2}{\Delta+1}-\frac{1}{\Delta}\sum_{i=1}^{\infty}\frac{1}{i^{2}}\nonumber \\
 & >1-\frac{4}{\Delta}.\label{eq:omega-0-vs-1}
\end{align}
An analogous application of~(\ref{eq:construction-omega-cases})
shows that
\begin{align}
\frac{|\omega(\frac{\Delta+1}{2})|}{|\omega(0)|} & \leq\frac{\frac{\Delta+1}{2}}{\frac{\Delta+1}{2}\cdot(\frac{\Delta+1}{2}-1)}\frac{\Delta^{d}d!\,d!}{(\Delta-\frac{\Delta+1}{2})\cdot\frac{1}{2}\Delta^{d-1}(d-1)!\,(d+1)!}\nonumber \\
 & =\frac{8\Delta d}{(\Delta-1)^{2}(d+1)}\nonumber \\
 & \leq\frac{8\Delta}{(\Delta-1)^{2}}.
\end{align}
Finally, for $i=1,2,\ldots,d,$
\begin{align}
\frac{|\omega(i^{2}\Delta)|}{|\omega(0)|} & =\frac{\frac{\Delta+1}{2}}{(i^{2}\Delta-1)(i^{2}\Delta-\frac{\Delta+1}{2})}\cdot\frac{d!\,d!\,\Delta^{d}}{\frac{1}{2}\cdot(d-i)!\,(d+i)!\,\Delta^{d}}\nonumber \\
 & \leq\frac{2(\Delta+1)}{i^{4}(\Delta-1)^{2}}\cdot\frac{d!\,d!}{(d-i)!\,(d+i)!}\nonumber \\
 & =\frac{2(\Delta+1)}{i^{4}(\Delta-1)^{2}}\cdot\frac{d}{d+i}\cdot\frac{d-1}{d+i-1}\cdot\cdots\cdot\frac{d-i+1}{d+1}\nonumber \\
 & \leq\frac{2(\Delta+1)}{i^{4}(\Delta-1)^{2}}\cdot\left(1-\frac{i}{d+i}\right)^{i}\nonumber \\
 & \leq\frac{2(\Delta+1)}{i^{4}(\Delta-1)^{2}}\cdot\exp\left(-\frac{i^{2}}{d+i}\right)\nonumber \\
 & \leq\frac{2(\Delta+1)}{i^{4}(\Delta-1)^{2}}\cdot\exp\left(-\frac{i^{2}}{2d}\right)\nonumber \\
 & \leq\frac{2(\Delta+1)}{i^{4}(\Delta-1)^{2}}\cdot\exp\left(-\frac{i^{2}}{2\sqrt{n/\Delta}}\right).\label{eq:omega-decay}
\end{align}
Now,
\begin{align}
\frac{\|\omega\|_{1}}{\omega(0)} & =1+\frac{|\omega(1)|}{\omega(0)}+\frac{|\omega(\frac{\Delta+1}{2})|}{\omega(0)}+\sum_{i=1}^{d}\frac{|\omega(i^{2}\Delta)|}{\omega(0)}\nonumber \\
 & \leq1+\left(1-\frac{4}{\Delta}\right)^{-1}+\frac{8\Delta}{(\Delta-1)^{2}}+\sum_{i=1}^{\infty}\frac{2(\Delta+1)}{i^{4}(\Delta-1)^{2}}\nonumber \\
 & =1+\left(1-\frac{4}{\Delta}\right)^{-1}+\frac{8\Delta}{(\Delta-1)^{2}}+\frac{\pi^{4}(\Delta+1)}{45(\Delta-1)^{2}}\nonumber \\
 & \leq\frac{2}{1-\frac{8}{\Delta}}\nonumber \\
 & <\frac{2}{1-\epsilon},\label{eq:omega-large-at-0}
\end{align}
where the second step uses~(\ref{eq:omega-0-vs-1})\textendash (\ref{eq:omega-decay}),
and the last step substitutes the definition of $\Delta$. Now~(\ref{eq:construction-omega-at-zero})
follows from~(\ref{eq:omega-large-at-0}) and $\omega(0)>0$. Moreover,
for $c=c(\Delta)>0$ small enough, (\ref{eq:construction-omega-orthog})
follows from~(\ref{eq:omega-orth}), whereas (\ref{eq:construction-omega-upper-bound})
follows from~(\ref{eq:omega-decay}) and the fact that $\omega$
vanishes outside the union $\{1,\frac{\Delta+1}{2}\}\cup\{i^{2}\Delta:i=0,1,2,\dots,d\}.$

It remains to verify that $\omega$ has the desired sign behavior.
Since $\omega$ vanishes outside $S,$ the requirement (\ref{eq:construction-omega-sign})
holds trivially at those points. For $t\in S,$ it follows from (\ref{eq:construction-omega-cases})
that 
\begin{align*}
 & \sign\omega(1)=-1,\\
 & \sign\omega\!\left({\textstyle \frac{\Delta+1}{2}}\right)=1,\\
 & \sign\omega(i^{2}\Delta)=(-1)^{i}, &  & i=0,1,2,\ldots,d.
\end{align*}
Since $\Delta\in4\ZZ+3$ by definition, we conclude that $\sign\omega(t)=(-1)^{t}$
for all $t\in S.$ This settles (\ref{eq:construction-omega-sign})
and completes the proof.
\end{proof}
We have reached the main result of this section.
\begin{thm*}[restatement of Theorem~\ref{thm:dual-OR}]
Let $0<\epsilon<1$ be given. Then for some constants $c',c''\in(0,1)$
and all integers $N\geq n\geq1,$ there is an $($explicitly given$)$
function $\psi\colon\{0,1,2,\ldots,N\}\to\Re$ such that
\begin{align}
 & \psi(0)>\frac{1-\epsilon}{2},\label{eq:psi-at-0}\\
 & \|\psi\|_{1}=1,\\
 & \orth\psi\geq c'\sqrt{n},\\
 & \sign\psi(t)=(-1)^{t}, &  & t=0,1,2,\ldots,N,\\
 & |\psi(t)|\in\left[\frac{c'}{(t+1)^{2}\,2^{c''t/\sqrt{n}}},\;\frac{1}{c'(t+1)^{2}\,2^{c''t/\sqrt{n}}}\right], &  & t=0,1,2,\ldots,N.\label{eq:psi-pointwise}
\end{align}
\end{thm*}
\begin{proof}
For a sufficiently small constant $c>0$ and all $n\geq1,$ Lemma~\ref{lem:construction-omega-Big-weight-at-zero}
and Remark~\ref{rem:monotonicity} ensure the existence of a function
$\omega\colon\{0,1,2,\dots,\lceil n/2\rceil\}\to\Re$ such that
\begin{align}
 & \|\omega\|_{1}=1,\label{eq:construction-thm-omegas-ell1}\\
 & \omega(0)>\frac{1}{2}\left(1-\frac{\epsilon}{6}\right),\label{eq:construction-thm-omega-at-0}\\
 & |\omega(t)|\leq\frac{1}{ct^{2}\,2^{ct/\sqrt{n}}} &  & (t=1,2,\ldots,\lceil n/2\rceil),\label{eq:construction-thm-omega-pointwise}\\
 & (-1)^{t}\omega(t)\geq0 &  & (t=0,1,2,\ldots,\lceil n/2\rceil),\label{eq:construction-thm-omega-sign}\\
 & \orth\omega\geq c\sqrt{n}.\label{eq:construction-thm-omega-orth}
\end{align}
For convenience, extend $\omega$ to all of $\ZZ$ by defining it
to be zero outside its original domain. Define $\Psi\colon\{0,1,2,\dots,N\}\to\Re$
by 
\begin{multline*}
\Psi(t)=\omega(t)+\delta\left(\sum_{i=1}^{N-\lceil n/2\rceil}\frac{(-1)^{i}}{i^{2}\,2^{ci/\sqrt{n}}}\omega(t-i)\right.\\
\left.+\sum_{i=N-\lceil n/2\rceil+1}^{N}\frac{(-1)^{i}}{i^{2}\,2^{ci/\sqrt{n}}}\omega(-t+i)\right),
\end{multline*}
where
\[
\delta=\frac{5\epsilon}{\pi^{2}(1-\epsilon)}.
\]
By~(\ref{eq:construction-thm-omega-orth}) and Proposition~\ref{prop:orth}\ref{item:orth-sum},
\begin{equation}
\orth\Psi\geq c\sqrt{n}.\label{eq:Psi-orth}
\end{equation}
We now move on to metric properties of $\Psi.$ Multiplying the defining
equation for $\Psi$ on both sides by $(-1)^{t}$ and applying~(\ref{eq:construction-thm-omega-sign}),
we arrive at
\begin{multline}
(-1)^{t}\Psi(t)=|\omega(t)|+\delta\left(\sum_{i=1}^{N-\lceil n/2\rceil}\frac{|\omega(t-i)|}{i^{2}\,2^{ci/\sqrt{n}}}+\sum_{i=N-\lceil n/2\rceil+1}^{N}\frac{|\omega(-t+i)|}{i^{2}\,2^{ci/\sqrt{n}}}\right),\\
t=0,1,2,\ldots,N.\label{eq:Psi-sign}
\end{multline}
Summing over $t$ gives
\begin{align}
\|\Psi\|_{1} & =\|\omega\|_{1}+\delta\sum_{i=1}^{N}\frac{1}{i^{2}\,2^{ci/\sqrt{n}}}\|\omega\|_{1}\nonumber \\
 & =1+\delta\sum_{i=1}^{N}\frac{1}{i^{2}\,2^{ci/\sqrt{n}}}\nonumber \\
 & \in\left[1,1+\delta\sum_{i=1}^{\infty}\frac{1}{i^{2}}\right]\nonumber \\
 & =\left[1,\frac{6-\epsilon}{6(1-\epsilon)}\right],\label{eq:Psi-ell1}
\end{align}
where the second step uses~(\ref{eq:construction-thm-omegas-ell1}).
We also have
\begin{align}
\Psi(0) & \geq\omega(0)\nonumber \\
 & >\frac{6-\epsilon}{12},\label{eq:Psi-at-0}
\end{align}
where the first and second steps use~(\ref{eq:Psi-sign}) and~(\ref{eq:construction-thm-omega-at-0}),
respectively.

We now estimate $|\Psi(t)|$ for each $t=1,2,\ldots,N.$ For a lower
bound, we have
\begin{align}
|\Psi(t)| & =|\omega(t)|+\delta\left(\sum_{i=1}^{N-\lceil n/2\rceil}\frac{|\omega(t-i)|}{i^{2}\,2^{ci/\sqrt{n}}}+\sum_{i=N-\lceil n/2\rceil+1}^{N}\frac{|\omega(-t+i)|}{i^{2}\,2^{ci/\sqrt{n}}}\right)\nonumber \\
 & \geq\delta\cdot\frac{|\omega(0)|}{t^{2}\,2^{ct/\sqrt{n}}}\nonumber \\
 & \geq\frac{5\epsilon}{\pi^{2}(1-\epsilon)}\cdot\frac{6-\epsilon}{12}\cdot\frac{1}{t^{2}\,2^{ct/\sqrt{n}}},\label{eq:Psi-at-t-lower}
\end{align}
where the first and last steps use~(\ref{eq:Psi-sign}) and~(\ref{eq:construction-thm-omega-at-0}),
respectively. The upper bound on $|\Psi(t)|$ is somewhat more technical.
To begin with, we have the following bound for every positive integer
$t$:

\begin{align}
\sum_{i=1}^{t-1}\frac{1}{(t-i)^{2}\,i^{2}} & =\sum_{i=1}^{t-1}\frac{1}{\max\{(t-i)^{2},i^{2}\}\,\min\{(t-i)^{2},i^{2}\}}\nonumber \\
 & \leq\frac{1}{(t/2)^{2}}\sum_{i=1}^{t-1}\frac{1}{\min\{(t-i)^{2},i^{2}\}}\nonumber \\
 & \leq\frac{1}{(t/2)^{2}}\cdot2\sum_{i=1}^{\infty}\frac{1}{i^{2}}\nonumber \\
 & =\frac{4\pi^{2}}{3t^{2}}.\label{eq:inverse-two-squares}
\end{align}
Continuing,
\begin{align}
\sum_{i=1}^{\infty}\frac{|\omega(t-i)|}{i^{2}\,2^{ci/\sqrt{n}}} & =\frac{|\omega(0)|}{t^{2}\,2^{ct/\sqrt{n}}}+\sum_{i=1}^{t-1}\frac{|\omega(t-i)|}{i^{2}\,2^{ci/\sqrt{n}}}\nonumber \\
 & \leq\frac{1}{t^{2}\,2^{ct/\sqrt{n}}}+\sum_{i=1}^{t-1}\frac{1}{c(t-i)^{2}\,i^{2}\,2^{ct/\sqrt{n}}}\nonumber \\
 & \leq\frac{1}{t^{2}\,2^{ct/\sqrt{n}}}\left(1+\frac{4\pi^{2}}{3c}\right),\label{eq:shift-right}
\end{align}
where the second step uses~(\ref{eq:construction-thm-omegas-ell1})
and~(\ref{eq:construction-thm-omega-pointwise}), and the third step
substitutes the bound from~(\ref{eq:inverse-two-squares}). Analogously,
\begin{align}
\sum_{i=1}^{\infty}\frac{|\omega(-t+i)|}{i^{2}\,2^{ci/\sqrt{n}}} & =\frac{|\omega(0)|}{t^{2}\,2^{ct/\sqrt{n}}}+\sum_{i=t+1}^{\infty}\frac{|\omega(-t+i)|}{i^{2}\,2^{ci/\sqrt{n}}}\nonumber \\
 & \leq\frac{1}{t^{2}\,2^{ct/\sqrt{n}}}+\sum_{i=t+1}^{\infty}\frac{1}{c(t-i)^{2}\,i^{2}\,2^{ci/\sqrt{n}}}\nonumber \\
 & \leq\frac{1}{t^{2}\,2^{ct/\sqrt{n}}}\left(1+\sum_{i=t+1}^{\infty}\frac{1}{c(t-i)^{2}}\right)\nonumber \\
 & =\frac{1}{t^{2}\,2^{ct/\sqrt{n}}}\left(1+\frac{\pi^{2}}{6c}\right),\label{eq:shift-left}
\end{align}
where the second step uses~(\ref{eq:construction-thm-omegas-ell1})
and~(\ref{eq:construction-thm-omega-pointwise}). Now for every integer
$t\geq1,$
\begin{align}
|\Psi(t)| & \leq|\omega(t)|+\delta\left(\sum_{i=1}^{\infty}\frac{|\omega(t-i)|}{i^{2}\,2^{ci/\sqrt{n}}}+\sum_{i=1}^{\infty}\frac{|\omega(-t+i)|}{i^{2}\,2^{ci/\sqrt{n}}}\right)\nonumber \\
 & \leq\frac{1}{ct^{2}\,2^{ct/\sqrt{n}}}\left(1+2c\delta+\frac{4\pi^{2}\delta}{3}+\frac{\pi^{2}\delta}{6}\right),\label{eq:Psi-at-t-upper}
\end{align}
where the first step is immediate from the defining equation for $\Psi,$
and the second step uses~(\ref{eq:construction-thm-omega-pointwise}),
(\ref{eq:shift-right}), and~(\ref{eq:shift-left}). 

To complete the proof, let $\psi\colon\{0,1,2,\ldots,N\}\to\Re$ be
given by $\psi=\Psi/\|\Psi\|_{1}.$ Then for $c''=c$ and small enough
$c'=c'(c,\epsilon,\delta)>0$, properties~(\ref{eq:psi-at-0})\textendash (\ref{eq:psi-pointwise})
follow directly from~(\ref{eq:Psi-orth})\textendash (\ref{eq:Psi-at-t-lower})
and (\ref{eq:Psi-at-t-upper}). 
\end{proof}

\end{document}